\documentclass[ba]{imsart}

\usepackage{natbib}
\RequirePackage[colorlinks,citecolor=blue,urlcolor=blue]{hyperref}

\usepackage{amsmath}
\usepackage{amsfonts}
\usepackage{amssymb}
\usepackage{graphicx}
\usepackage{bm}
\usepackage{dsfont}
\usepackage{color}
\usepackage{enumitem}
\usepackage{pbox}

\usepackage{multirow}

\pubyear{0000}
\volume{TBA}
\issue{TBA}
\firstpage{1}
\lastpage{1}

\startlocaldefs
\providecommand{\U}[1]{\protect\rule{.1in}{.1in}}
\newtheorem{theorem}{Theorem}

\newtheorem{proposition}[theorem]{Proposition}

\newenvironment{proof}[1][Proof]{\noindent\textbf{#1.} }{\ \rule{0.5em}{0.5em}}
\usepackage{booktabs}
\usepackage{lipsum,float}
\usepackage{subcaption}
\usepackage{url}
\newcommand{\obsim}[1]{\overset{#1}{\boldsymbol\sim}}

\newcommand{\rhogg}{\rho_{\text{\tt GG}}}
\newcommand{\psigg}{\psi_{\text{\tt GG}}}
\newcommand{\1}[1]{\mathds{1}_{#1}}

\newcommand{\revision}[1]{#1}

\newcommand{\qnts}[3]{\parbox{1.6cm}{\centering\scriptsize #1 \\ (#2, #3)}}

\DeclareMathOperator{\GBFRY}{GBFRY}
\DeclareMathOperator{\Gammadist}{Gamma}
\DeclareMathOperator{\Betadist}{Beta}
\DeclareMathOperator{\Poisson}{Poisson}
\DeclareMathOperator{\GG}{GG}
\DeclareMathOperator{\GGP}{GGP}
\DeclareMathOperator{\NGGP}{NGGP}
\DeclareMathOperator{\Pareto}{Pareto}
\DeclareMathOperator{\Unif}{Unif}
\endlocaldefs

\begin{document}

\begin{frontmatter}

% "Title of the Paper"
\title{The Normal-Generalised Gamma-Pareto process: A novel pure-jump L\'evy process with flexible tail and jump-activity properties}
\runtitle{The Normal-Generalised Gamma-Pareto process}

% indicate corresponding author with \corref{}
% \author{\fnms{John} \snm{Smith}\thanksref{t1}\corref{}\ead[label=e1]{smith@foo.com}\ead[label=e2,url]{www.foo.com}}
% \thankstext{t1}{Thanks to somebody}
% \address{line 1\\ line 2\\ \printead{e1}\\ \printead{e2}}

\author{\fnms{Fadhel} \snm{Ayed}\ead[label=e1]{fadhel.ayed@gmail.com}}
\address{Department of Statistics, University of Oxford  \printead{e1}}
\and
\author{\fnms{Juho} \snm{Lee}\ead[label=e2]{juholee@kaist.ac.kr}}
\address{Kim Jaechul Graduate School of AI, KAIST  \printead{e2}}
\and
\author{\fnms{Fran\c cois} \snm{Caron}\ead[label=e3]{caron@stats.ox.ac.uk}}
\address{Department of Statistics, University of Oxford  \printead{e3}}
\runauthor{F. Ayed et al.}

\begin{abstract}
We propose a novel family of self-decomposable L\'evy
processes where one can control separately the tail behavior and the jump activity of the process, via two different parameters. Crucially, we show that one can sample exactly increments of
this process, at any time scale; this allows the implementation of likelihood-free Markov chain Monte Carlo algorithms for (asymptotically) exact posterior inference. We use this novel process in L\'evy-based stochastic volatility models to
predict the returns of stock market data, and show that the proposed class of models leads to superior predictive performances compared to classical alternatives.
\end{abstract}

%\begin{keyword}[class=MSC]
%\kwd[Primary ]{}
%\kwd{}
%\kwd[; secondary ]{}
%\end{keyword}


% history:
% \received{\smonth{1} \syear{0000}}

%\tableofcontents

\end{frontmatter}

% Main text entry area
\section{Introduction}

Pure-jump L\'evy processes are a flexible class of stochastic processes that have found a wide range of applications, including scalable Markov chain Monte Carlo~\citep{SimSekli2017}, tracking~\citep{Zhang2018} or the analysis of phylogenetic
traits~\citep{Landis2012,Landis2017}. Finance is probably the main domain of application, as it is widely accepted that asset prices contain jumps, and such models have been used as building blocks of complex dynamic models of asset or option prices~\citep{Madan1998,Barndorff-Nielsen2001,Carr2002,Cont2004,Huang2004,Jing2012}.%\medskip

Let $\left(  X_{t}\right)  _{t\geq0}$ be a real-valued pure-jump L\'evy process. The process is said to have \textit{heavy, power-law tails} if, for any $t,\Delta>0$,
\begin{align}
\Pr(|X_{t+\Delta}-X_t|>x)\obsim{x\to\infty} \Delta C_1  x^{-2\tau}
\end{align}
for some power-law exponent $\tau>0$ and some constant $C_1>0$. That is, for large $x$, the survival function of the increments approximately behaves as a power function.

Many financial time series, such as historical asset returns, exhibit heavy-tails. Empirical evidence seems to indicate that the returns have nonetheless finite variance, hence corresponding to a power-law exponent $\tau>1$ \revision{~\citep[Section 7.3]{Cont2004}}. Starting from the early work of \cite{Mandelbrot1963} with the stable distribution, various infinite-divisible distributions, closely related to L\'evy processes, have been proposed to capture power-law tails. Examples include the student t~\citep{Blattberg1974} or Pareto~\citep{Champagnat2013} distributions; other models with (non power-law) semi-heavy tails such as the normal inverse
Gaussian~\citep{Barndorff-Nielsen1997}, generalized
hyperbolic~\citep{Eberlein1998} and tempered stable
distributions~\citep{Cont1997,Carr2002} have also been proposed; see \citep[Section 7.3]{Cont2004} for
a review. \medskip

Another quantity of interest of the L\'evy process is the Blumenthal-Getoor (BG) index $\beta\in[0,2]$, also known as fractional order. It is defined by
\begin{align}
\beta=\inf \left \{r>0~\mid\ \sum_{i\geq 1} |J_i|^r \1{\theta_i\leq 1}<\infty\right \}
\end{align}
\revision{where $(J_i,\theta_i)_{i\geq 1}$ is the set of jump sizes and jump times}. The BG index measures the level of activity of the jumps: as the value of $\beta$ increases, small jumps tend to become more and more frequent. It is also related to the smoothness properties of the time series~\citep[Section 7.3]{Cont2004} and therefore provides interpretable information on the process and its properties. A number of papers have proposed and analysed (model-free) estimators of this index \citep{AitSahalia2009,Belomestny2010,Woerner2011,Belomestny2013}. Some L\'evy  processes, such as the normal-tempered stable or tempered stable processes, can capture the whole range $[0,2)$ via a tuning parameter; other processes, such as the variance gamma ($\beta=0$), normal-inverse Gaussian, student t or generalised hyperbolic ($\beta=1$), have a fixed BG index.\smallskip

For a pure-jump L\'evy process, a typical way to obtain a given power-law exponent and BG index is to assume the regular variation of the tail intensity of the L\'evy measure $\nu$ characterising the L\'evy process, such that
\begin{align*}
  \int_{\left\vert s\right\vert \geq x}\nu(ds)\obsim{x\to\infty} x^{-2\tau} C_1 ~~\text{ and }~~
  \int_{\left\vert s\right\vert \geq x}\nu(ds)\obsim{x\to 0}x^{-\beta}\ell(1/x)
\end{align*}
for some slowly varying function $\ell$, that is such that $\lim_{t\to\infty}\ell(ct)/\ell(t)= 1$ for all $c>0$. While many L\'evy measures have been proposed in the literature, no tractable model is able to capture both the whole range of power-law exponent  $\tau>0$ and BG index $\beta\in[0,2)$. Normal-tempered stable process for example (which includes as special case the variance-gamma and normal-inverse Gaussian) capture the whole range of the $\beta$ index but have light tails. The class of generalised hyperbolic processes can capture heavy tails, but has a fixed BG index equal to 1. The normal-stable process can capture both heavy tails and different BG indices, but the same parameter controls both properties, and the process has infinite variance.\smallskip

In this paper, we introduce a novel four-parameter pure-jump L\'evy process, called \textit{normal generalised gamma-Pareto} (NGGP) process, with the following properties.
\begin{itemize}[noitemsep,nolistsep]
\item The model can capture power-law, heavy tails with a single parameter $\tau>0$; for $\tau>1$, the process has finite variance;
\item Another parameter $\sigma\in(-\infty,1)$ controls the BG index and therefore the activity of the jumps, with $\beta=\max(0,2\sigma)$. The process is finite-activity for $\sigma<0$; it is infinite-activity for $\sigma\geq 0$; it is of bounded variation if $\sigma<1/2$ and of unbounded variation if $\sigma\in[1/2,1)$;
\item The other two parameters respectively are inverse scale and time scale parameters. \revision{More details on the interpretability of the parameters can be found in Section \ref{sec:GGP}};
\item One can sample exactly from the distribution of the increments, at any time scale; this enables the use of likelihood-free Markov chain Monte Carlo methods for inference;
\item The distribution of the increments of the L\'evy process is self-decomposable.
\end{itemize}
The NGGP process is obtained via Brownian subordination, using the subordinator introduced by~\cite{Ayed2019} for modeling power-law properties of text data. We derive a number of properties of the NGGP process and use the proposed model to predict the stock prices of some financial assets. We consider two L\'evy based stochastic volatility models: an exponential L\'evy model, and an Ornstein-Uhlenbeck based model L\'evy-driven stochastic volatility model~\citep{Barndorff-Nielsen2001}. We show that, compared to other L\'evy processes, the proposed model is both able to capture the heavy-tail and small-jump behaviours.

The article is organised as follows. In Section \ref{sec:model}, we introduce the generalised gamma-Pareto subordinator, its properties, and the associated subordinated Brownian process. In Section~\ref{sec:modelfinance} we describe two L\'evy process based stochastic volatility models, and describe how to perform asymptotically exact posterior inference under our L\'evy process with both models. In Section~\ref{sec:experiments} we present experimental results on the modelling of stock prices and show that our model provides a very good fit to the data and good predictive performances compared to classical alternatives.

\paragraph{Notations. } We use the notation $a_n\obsim{n\to \infty}b_n$ for $\lim_{n\to \infty} a_n/b_n =1$. For a random variable $X$, the notation $X\sim F$ indicates that $X$ has distribution $F$. $\Gammadist(a,b)$ denotes the gamma distribution with shape parameter $a$ and inverse scale parameter $b$. $\Poisson(\lambda)$ denotes the standard Poisson distribution with rate $\lambda$.

\section{The NGGP process}
\label{sec:model}

\subsection{Generalised gamma subordinator}

A generalised gamma (GG) subordinator (almost surely increasing L\'evy process) $(Y_t)_{t\geq 0}$ has L\'evy intensity~\citep{Hougaard1986,Aalen1992,Brix1999}
\begin{equation}
\rhogg(w;\eta,\sigma,c)=\frac{\eta}{\Gamma(1-\sigma)}w^{-1-\sigma}e^{-cw},~~~w>0\label{eq:GGintensity}
\end{equation}
where $\eta>0$ and $\sigma\in(-\infty,1)$, $c>0$ or $c=0,\sigma\in(0,1)$. The subordinator is finite-activity for $\sigma<0$ and infinite-activity if $\sigma\in[0,1)$. It admits as special cases the gamma process ($\sigma=0$), inverse-Gaussian process ($\sigma=1/2$) and stable process ($c=0$). When $\sigma> 0$, the process belongs to the general family of tempered stable processes introduced by~\cite{Rosinski2007} and some authors referred to this process simply as a tempered stable process~\citep{BarndorffNielsen2002,Liang2015}. Ignoring the drift term, $Y_t$ has Laplace transform
\begin{equation}
\mathbb E[e^{-\vartheta Y_t}]=\exp\left ( -t\psigg(\vartheta;\eta,\sigma,c)\right )
\end{equation}
where the Laplace exponent is given by
\begin{equation}
\psigg(\vartheta;\eta,\sigma,c)=\left \{
\begin{array}{ll}
  \frac{\eta }{\sigma}\left [(\vartheta+c)^\sigma -c^\sigma \right ] & \sigma\neq 0\\
  \eta \log(1+\vartheta/c) & \sigma = 0.
\end{array}\right .
\end{equation}
$Y_t$ is said to have the generalised gamma distribution with parameters $(\eta t,\sigma,c)$, and we write $Y_t\sim \GG(t\eta,\sigma,c)$. For $\sigma=0$, $Y_t\sim\Gammadist(\eta t,c)$, while for $\sigma<0$, $Y_t$ is a compound Poisson-gamma distribution with
$$
Y_t\overset{d}{=}\sum_{k=1}^{K_t} Y_{t,k}
$$
where $K_t\sim\Poisson(\eta t \frac{c^\sigma}{-\sigma})$ and $Y_{t,k}\sim \Gammadist(-\sigma,c)$ for $k=1,\ldots,K_t$. For $\sigma>0$, $Y_t$ is an exponentially tilted stable random variable, for which exact samplers exist~\citep{Devroye2009,Hofert2011}.

\subsection{Generalised gamma-Pareto subordinator}
\label{sec:GGP}
\subsubsection{Definition}
Let $Z=\left(  Z_{t}\right)  _{t\geq0}$ be a
subordinator with no drift and L\'evy intensity
% The subordinator is called a\fc{need to find a better name} generalized BFRY (GBFRY) subordinator, or incomplete-gamma tempered stable process, when the L\'evy measure is absolutely continuous with intensity
\begin{equation}
\rho(w)=\frac{\eta }{c^{\tau}\Gamma(1-\sigma)}w^{-1-\tau}\left [\gamma(\tau-\sigma+1,cw)+(cw)^{\tau-\sigma}e^{-cw}\right ],~~~w>0\label{eq:GBFRYintensity}
\end{equation}
where $\eta>0$, $c>0$, $\sigma\in(-\infty,1)$, $\tau>0$ and $\gamma(s,x)=\int_0^x t^{s-1}e^{-t}dt$ is the lower incomplete gamma function. For $\tau>\sigma$, using the identity (43) in Appendix A, the L\'evy intensity takes the simpler form
\begin{equation}
\rho(w)=\frac{\eta (\tau-\sigma)}{c^{\tau}\Gamma(1-\sigma)}w^{-1-\tau}\gamma(\tau-\sigma,cw)\label{eq:GBFRYintensity2}
\end{equation}
which is the form in which \cite{Ayed2019} introduced the process, with a slightly different parameterisation. The L\'evy intensity admits the following representation as a mixture of generalised gamma process
\begin{equation}
\rho(w)=\int_0^\infty u^{-1} \rhogg\left (\frac{w}{u};\frac{\eta (\tau-\sigma)}{\tau c^{\sigma}},\sigma,c\right ) f_U(u)du\label{eq:GBFRYintensity3}
\end{equation}
where $f_U(u)=\tau u^{-1-\tau}\1{u\geq 1}$ is the probability density function of a Pareto random variable $\Pareto(\tau,1)$ with support $[1,\infty)$ and power-law exponent $\tau>0$, and $\rhogg$ is the L\'evy intensity of a GG subordinator, defined in Equation~\eqref{eq:GGintensity}.  We will thereafter refer to the subordinator with intensity~\eqref{eq:GBFRYintensity} as a \textit{Generalised Gamma-Pareto} (GGP) process. For $x>0$, let
\begin{align}
\bar\rho(x)&=\int_x^\infty \rho(w)dw
\end{align}
be the tail L\'evy intensity. For $t>0$, we denote $F_{Z_t}$ the cumulative distribution function of the random variable $Z_t$, with Laplace transform
$$
\mathbb E[e^{-\vartheta Z_t}]=e^{-t\psi(\vartheta)}
$$
where $\psi$ is the Laplace exponent which is given by, noting that $\gamma(\tau-\sigma+1,cw)=w^{\tau-\sigma+1}\int_0^c u^{\tau-\sigma}e^{-wu}du$,
\begin{align}
\psi(\vartheta)&=\int_0^\infty (1-e^{-w\vartheta})\rho(w)dw=\frac{\eta }{c^\tau }\left [\frac{c^\tau}{\tau}-\int_0^c (u+\vartheta)^{\sigma-1}u^{\tau-\sigma}du+\frac{c^{\tau-\sigma}}{\sigma}((\vartheta+c)^\sigma-c^\sigma) \right ].\label{eq:laplaceexponent}
\end{align}
$Z_t$ is said to have $\text{GGP}(t\eta, \sigma,\tau,c)$ distribution.

\subsubsection{Properties}

We derive here a number of properties of the L\'evy process $Z$ and of the GGP distribution.

\paragraph{Positive stable process.} The positive stable process with L\'evy intensity $\frac{\eta}{\Gamma(1-\sigma)}w^{-1-\sigma}$ is obtained as a special case when $c=1$, $\sigma=\tau\in(0,1)$.

\paragraph{Scaled GG process.} Let
$$
Y_t=\sum_{i\geq 1} W_i\1{\theta_i\leq t}
$$
where $\{(W_i,\theta_i)\}_{i\geq 1}$ are the jump sizes and times of a GG subordinator. Then, for $\tau>\sigma$, the representation \eqref{eq:GBFRYintensity3} implies that
$$
Z_t\overset{d}{=}\sum_{i\geq 1} W_i U_i \1{\theta_i\leq t}
$$
where $U_i\sim \Pareto(\tau,1)$. The jump sizes of the GGP subordinator are obtained by scaling the jumps of a GG subordinator with independent Pareto random variables.

\paragraph{Moments and cumulants.} We have $\mathbb E[Z_t^m]<\infty$ for $m<\tau$ and $\mathbb E[Z_t^m]=\infty$ otherwise. For $1\leq m<\tau$, the $m$th cumulant is given by
$$
\kappa_m(Z_t)=t \int_0^\infty w^m\rho(w)dw=\frac{t\eta (\tau-\sigma)\Gamma(m-\sigma)}{c^m(\tau-m)\Gamma(1-\sigma)}.
$$
In particular, for $\tau>1$
\begin{align*}
\mathbb E[Z_t]=\frac{t\eta (\tau-\sigma)}{c(\tau-1)}
\end{align*}
and for $\tau>2$,
\begin{align*}
\text{var} (Z_t)=\frac{t\eta (\tau-\sigma)(1-\sigma)}{c^2(\tau-2)}.
\end{align*}

\paragraph{Inverse scale parameter.} If the intensity $\rho$ is of the form \eqref{eq:GBFRYintensity}  for some parameters $(\eta,\sigma,\tau,c)$, then $\rho(w/c)/c$ is also of the form \eqref{eq:GBFRYintensity} with parameters $(\eta,\sigma,\tau,1)$. $c$ is therefore an inverse scale parameter, and if $Z_t\sim\GGP(t\eta, \sigma, \tau, c)$ then $c Z_t\sim \GGP(t\eta, \sigma, \tau, 1)$.

\paragraph{Activity of the jumps and BG index.} The L\'evy intensity \eqref{eq:GBFRYintensity} satisfies $\int_0^\infty \rho(w)dw=\infty$ if $\sigma\geq 0$ and the subordinator is therefore infinite-activity. If $\sigma<0$, $\int_0^\infty \rho(w)dw<\infty$ and it is finite-activity. More precisely, as noted by \cite{Ayed2019}, the tail L\'evy intensity is regularly varying at 0
\begin{align}
\bar\rho(x)&\obsim{x\to 0} \ell(1/x)x^{-\alpha}\label{eq:RV0GBFRY}
\end{align}
where $\alpha=\max(0,\sigma)$ is the BG index, and the slowly varying function $\ell$ is defined by
\begin{align}
\ell(t)= \left \{\begin{array}{ll}
                             \frac{\eta}{c^\sigma\sigma\Gamma(1-\sigma)} & \sigma>0 \\
                             \eta\log (t) & \sigma=0\\
                             \frac{\eta(\tau-\sigma)}{-\sigma\tau}& \sigma<0.
                           \end{array}\right .
\end{align}

\revision{For any $x>0$ and $t\geq 0$, notice that $\bar\rho(x)=\mathbb E\left [\sum_{i\geq 1} \1{J_i\geq x} \1{\theta_i\in[t,t+1]}\right ]$, where $\{(J_i,\theta_i)\}_{i\geq 1}$ are the jump sizes and times of the GGP subordinator. Hence the BG index $\alpha$ controls the number of jumps above a certain threshold $x>0$ per time unit. It also tunes  a number of asymptotic properties of the Laplace exponent of the L\'evy measure and of the cumulative distribution function and small time distribution of the increments, as described below.}

It follows from the Abelian theorem~\cite[Proposition 17]{Gnedin2007} that the Laplace exponent satisfies
\begin{equation}
\psi(\vartheta)\obsim{\vartheta\to\infty}\Gamma(1-\alpha)\vartheta^\alpha\ell(\vartheta).\label{eq:asymppsi}
\end{equation}
For $\sigma=\alpha\in(0,1)$, the cumulative distribution function $F_{Z_t}$ satisfies~\cite[Theorem 8.2.2. p. 341]{Bingham1989}
\begin{align}
-\log F_{Z_t}(z)\obsim{z\to 0} (1-\alpha)\alpha^{\alpha/(1-\alpha)}\left (\frac{\eta t}{c^{\alpha}\alpha}\right )^{1/(1-\alpha)}z^{-\alpha/(1-\alpha)}.
\end{align}
Additionally, \revision{using \eqref{eq:asymppsi},} for small increments, we have, for all $\vartheta\geq 0$
$$
\mathbb E [e^{-\vartheta\frac{c Z_t}{(\eta t/\alpha)^{1/\alpha}}}]\overset{t\to 0}{\to} e^{-\vartheta^\alpha}
$$
hence $\frac{c Z_t}{(\eta t/\alpha)^{1/\alpha}}$ tends in distribution to a positive stable random variable with parameter $\alpha\in(0,1)$ as $t\to 0$.
For $\sigma<0$, the cdf has a discontinuity at 0 with
$$
F_{Z_t}(0)-F_{Z_t}(0_-)=\Pr(Z_t=0)=e^{-t\frac{\eta(\tau-\sigma)}{-\sigma\tau}}.
$$

\paragraph{Heavy tails and power-law behaviour.} As noted by \cite{Ayed2019}, the tail L\'evy intensity is regularly varying at infinity, with power-law exponent $\tau$. We have
\begin{align}
\bar\rho(x)&\obsim{x\to\infty} \frac{\eta \Gamma(\tau-\sigma+1)}{\tau c^\tau \Gamma(1-\sigma)}x^{-\tau}.\label{eq:RVinfGBFRY}
\end{align}

It follows from \cite[Theorem 8.2.1. page 341]{Bingham1989} that the survival function $1-F_{Z_t}(z)=\Pr(Z_t>z)$ satisfies
\begin{align}
1-F_{Z_t}(z)\obsim{z\to\infty} \frac{\eta t\Gamma(\tau-\sigma+1)}{\tau c^\tau \Gamma(1-\sigma)}z^{-\tau}
\end{align}
and the increments have heavy, power-law tails with exponent $\tau>0$.

\paragraph{Simulation of the increments.}

First note that if $\sigma<0$, the subordinator is a compound Poisson process with jump rate $\frac{\eta (\tau-\sigma)}{-\sigma\tau}$ and jumps being GBFRY distributed (see Section B in the Appendix) with parameters $ (-\sigma, \tau,c)$. We therefore have
$$
Z_t\overset{d}{=}\sum_{j=1}^{K_t} G_{t,j} U_{t,j}
$$
where $K_t\sim\Poisson(t\frac{\eta(\tau-\sigma)}{-\sigma\tau})$, $G_{t,j}\sim\Gammadist(-\sigma,c)$ and $U_{t,j}\sim\Pareto(\tau,1)$, $j=1,\ldots,K_n$ are independent random variables.
Consider now the case $\sigma\geq 0$. The L\'evy measure admits the two-components mixture representation
\begin{align}
\rho(w) &= \frac{\eta c^{-\sigma}}{\Gamma(1-\sigma
)}w^{-1-\sigma}e^{-cw} + \frac{\eta}{c^\tau\Gamma(1-\sigma)}
w^{-1-\tau} \gamma(\tau+1-\sigma, cw)\label{eq:mixture}\\
&= \frac{\eta c^{-\sigma}}{\Gamma(1-\sigma
)}w^{-1-\sigma}e^{-cw} + \frac{\eta(\tau-(\sigma-1))(1-\sigma)}{c^\tau\Gamma(1-(\sigma-1))(\tau-(\sigma -1))}
w^{-1-\tau} \gamma(\tau-(\sigma-1), cw)\nonumber
\end{align}

The first component of the mixture representation~\eqref{eq:mixture} is the L\'evy intensity of a GG subordinator \revision{with parameters $(\eta_1=\eta/c^\sigma,\sigma_1=\sigma,c_1=c)$}. The second component is the intensity of a GGP subordinator \revision{with parameters $(\eta_2=\eta(1-\sigma)/(\tau+1-\sigma),\sigma_2=\sigma-1,\tau_2=\tau,c_2=c)$}; as $\sigma_2=\sigma-1<0$, this subordinator is a finite-activity compound Poisson process, and one can sample its increments as described above. We can
therefore write%
\begin{align}
Z_{t}\overset{d}{=}Z_{t,1}+Z_{t,2}\label{eq:Zsum}
\end{align}
where $Z_{t,1}\sim \GG(\frac{\eta t}{c^\sigma},\sigma, c)$ is an exponentially tilted stable random variable for which exact samplers exist~\citep{Devroye2009,Hofert2011}, and
$$
Z_{t,2}\overset{d}{=}\sum_{j=1}^{\widetilde K_t} \widetilde G_{t,j}\widetilde U_{t,j}
$$
where $\widetilde K_t\sim\Poisson(\frac{\eta t}{\tau})$, $\widetilde G_{t,j}\sim\Gammadist(1-\sigma,c)$ and $\widetilde U_{t,j}\sim\Pareto(\tau,1)$.

\paragraph{Self-decomposability.}

\revision{Self-decomposable distributions, a subclass of infinitely-divisible distibutions, are closely related to stationary processes of Ornstein-Uhlenbeck type. Such models, described in Section \ref{sec:OUmodels}, have been extensively used for the modeling of financial times series, see e.g. \cite{Barndorff-Nielsen2001}.}

\begin{proposition}
The random variable $Z_t\GGP(t\eta, \sigma, \tau, c)$ is self-decomposable if $\sigma\geq 0$. That is, for any $a\in(0,1)$, there is $Z_t^{(a)}$ independent
of $Z_t$ such that
\[
Z_t\overset{d}{=}a Z_t + Z_t^{(a)}.
\]
\end{proposition}
\begin{proof}
Let $k(w) = w \rho(w)$. Consider first that $\tau > \sigma \geq 0$.
From equation~\eqref{eq:GBFRYintensity2}, we have $k(w)\propto w^{-\sigma}\int_0^c u^{\tau-\sigma-1}e^{-wu}du$ which is non-increasing.  Consider that $0 < \tau \leq \sigma$. From Equation~\eqref{eq:GBFRYintensity}, $k$ takes the form $k(w) \propto w^{-\tau}g(w)$ where
$$
g(w)= \gamma(\tau-\sigma+1, cw) + (cw)^{\tau-\sigma}  e^{-cw}.
$$
As $g'(w)  = (\tau - \sigma) c (cw)^{\tau-\sigma-1} e^{-cw}\leq 0$, it follows that $k$ is monotone decreasing.
Hence for all $\sigma\geq 0$, $k$ is monotone decreasing; using Proposition 15.3 p.485 in~\citep{Cont2004}, we conclude that the process is therefore self-decomposable.
\end{proof}

\revision{The self-decomposable random variable $Z_1$ admits the representation~\citep{Jurek2001}}
\begin{align}
Z_1\overset{d}{=}\int_0^\infty e^{-s}d\widetilde Z_s
\end{align}
where $(\widetilde Z_s)_{s\geq 0}$ is termed the background driving L\'evy process corresponding to the self-decomposable random variable $Z_1$ \cite[Section 2.2]{Barndorff-Nielsen2001}. $\widetilde Z_t$ has L\'evy intensity
\begin{align}
\widetilde \rho(w)&=-\rho(w)-w\rho'(w) =\frac{\eta \sigma }{c^{\sigma}\Gamma(1-\sigma)} w^{-1-\sigma}e^{-cw}+\frac{\eta\tau}{c^\tau\Gamma(1-\sigma)}%
w^{-1-\tau}\gamma(\tau-\sigma+1,cw).\label{eq:tilderhoGGP}
\end{align}
Importantly, for $\sigma=0$, the background L\'evy process is a finite-activity GGP process with intensity
\begin{align}
\widetilde \rho(w)=\frac{\eta\tau}{c^\tau}w^{-1-\tau}\gamma(\tau+1,cw).
\end{align}

\paragraph{Interpretability of the parameters.} In summary, each of the four parameters governs a different property of the GGP process.
\begin{itemize}[noitemsep,nolistsep]
\item $\eta>0$ is a time-scaling parameter: if $(Z_t)_{t\geq 0}$ is a GGP process with parameters $(1,\sigma,\tau,c)$, then $(Z_{\eta t})_{t\geq 0}$ is a GGP process with parameters $(\eta,\sigma,\tau,c)$;
\item $c>0$ is an inverse-scale parameter: if $(Z_t)$ is a GGP process with parameters $(\eta,\sigma,\tau,1)$, then $(Z_t/c)$ is a GGP process with parameters $(\eta,\sigma,\tau,c)$;
\item $\sigma\in(-\infty,1)$ tunes the activity of the jumps; the process is finite-activity if $\sigma<0$ and infinite-activity otherwise, with  corresponding BG index $\beta=\max(\sigma,0)$;
\item $\tau$ is the power-law exponent, controlling the tails of the distribution, with $\Pr(Z_t>z)\obsim{z\to\infty} C z^{-\tau}$ for some constant $C$.
\end{itemize}

\subsection{Normal GGP process}
\label{sec:NGGP}

\subsubsection{Definition}

Let $Z=\left(  Z_{t}\right)  _{t\geq0}$ be a GGP
subordinator with no drift and L\'evy intensity given by \eqref{eq:GBFRYintensity}. Let $B=\left(  B_{t}\right)  _{t\geq0}$ be a
Brownian motion on $\mathbb R$, independent from\ $Z$. The normal generalised gamma Pareto (NGGP) L\'evy process, taking values in $\mathbb R$, is defined via Brownian subordination by
\[
X_{t}= B_{Z_{t}}.
\]
For any $t>0$, Let $F_{X_t}$ denote the cumulative distribution function of the random variable $X_t$, with characteristic function~\cite[Section 4.2]{Cont2004}%
\[
\mathbb{E}[e^{i\lambda X_{t}}]=e^{t\Psi(\lambda)}%
\]
where the characteristic exponent is given by~\cite[Theorem 4.2]{Cont2004}
\begin{equation}\label{eq:characfunX}
\Psi(\vartheta)=-\psi(\vartheta^{2}/2)=\int_{\mathbb{R}%
}(e^{i\vartheta x}-1)\nu(x)dx
\end{equation}
where $\psi$ is defined in Equation~\eqref{eq:laplaceexponent} and $\nu$ is a L\'{e}vy intensity on $\mathbb{R}$ defined by
\[
\nu(x)=\int_{0}^{\infty}\frac{1}{\sqrt{2\pi w}}e^{-\frac{x^{2}%
}{2w}}\rho(w)dw.
\]
For any $x>0$, let
$$
\overline\nu(x)=\int_{|s|>x}\nu(s)ds
$$
denote the expected number of jumps of absolute value larger than $x$ in an unit-length interval. We also write $X_t\sim \NGGP(\eta t, \sigma, \tau, c)$.

\subsubsection{Properties}

Most of the properties here follow from the properties of the subordinator. By construction, we have, for all $t$
\begin{equation}
X_t\overset{d}{=}\sqrt{Z_t}\epsilon_t
\end{equation}
where $\epsilon_t\sim\mathcal N(0,1)$. $\sqrt{c}$ is therefore an inverse scale parameter, and if $X_t\sim\NGGP(\eta t,\sigma,\tau,c)$ then $\sqrt{c} X_t\sim\NGGP(\eta t,\sigma,\tau,1)$.

\paragraph{Moments and cumulants.} Let $1\leq  m<2\tau$. For $m$ odd, the $m$'th raw moment and cumulant of $X_t$ satisfy
\begin{align}
\mathbb E[X_t^m]=\kappa_m(X_t)&=0.
\end{align}
For $2\leq  m<2\tau$, $m$ even,
\begin{align*}
\kappa_m(X_t)&=2\int_{0}^\infty\int_0^\infty \frac{x^m}{\sqrt{2\pi w}}e^{-\frac{x^{2}%
}{2w}}\rho(w)dw dx\\
%&=\frac{1}{\sqrt{\pi}}\int_{0}^\infty\int_0^\infty (2wu)^{m/2} u^{-1/2} e^{-u}\rho(w)dw du\\
&=\frac{2^{m/2}}{\sqrt{\pi}}\kappa_{m/2}(Z_t)\Gamma((m+1)/2)\\
&=\frac{t\eta (\tau-\sigma)2^{m/2}\Gamma((m+1)/2)}{c^{m/2}(\tau-m/2)\sqrt{\pi}}
\end{align*}
It follows, for $\tau>1$, the L\'evy process has finite variance with
\begin{align*}
\text{var} (X_t)=\mathbb E[Z_t]=\frac{t\eta (\tau-\sigma)}{c(\tau-1)}<\infty.
\end{align*}
For $\tau>2$, the excess kurtosis is finite and given by
$$
\text{kurt}(X_t)=\frac{\kappa_4(X_t)}{\kappa_2(X_t)^2}=\frac{3\text{var}(Z_t)}{(\mathbb E[Z_t])^2}=\frac{3(\tau-1)^2}{t\eta(\tau-2)(\tau-\sigma)}.
$$

\paragraph{Activity of the jumps and BG index.} The L\'evy process $X$ is infinite-activity if $\sigma\geq0$ and finite-activity otherwise. Using Proposition \ref{prop:tauberian} in Appendix \ref{sec:app:taubersubbrownian}, the regular variation of $\overline\rho$ at 0 in Equation~\eqref{eq:RV0GBFRY} implies the regular variation of $\overline\nu$ at 0
\begin{align}
\bar\nu(x)&\obsim{x\to 0} \frac{2^{\alpha+1}\Gamma(\alpha+1/2)}{\sqrt{\pi}}\ell(1/x^2)x^{-2\alpha}.
\end{align}
The BG index of $X$ is therefore equal to $2\alpha=2\max(0,\sigma)\in[0,2)$. When $\sigma=\alpha>0$, combining \eqref{eq:characfunX} with \eqref{eq:asymppsi}, we obtain the small time limit%\fc{check, should be $X_t$ below}
$$
\mathbb E\left [\exp\left (i\vartheta \frac{X_t \sqrt{2c}}{(\eta t)^{1/(2\alpha)}}\right )\right]\overset{t\to 0}\rightarrow e^{-|\vartheta|^{2\alpha}}
$$
hence $\frac{X_t \sqrt{2c}}{(\eta t)^{1/(2\alpha)}}$ tends in distribution to a symmetric stable distribution with parameter $2\alpha\in(0,2)$ when $t$ tends to 0.

\paragraph{Heavy tails and power-law behaviour. }
Using Proposition \ref{prop:tauberian} in the Appendix, the regular variation of $\overline\rho$ at infinity in Equation~\eqref{eq:RVinfGBFRY} implies the regular variation of $\overline\nu$ at infinity
\begin{align}
\bar\nu(x)&\obsim{x\to \infty} C_1 x^{-2\tau}
\end{align}
where
\begin{align}
C_1=\frac{2^{\tau+1}\Gamma(\tau+1/2)}{\sqrt{\pi}}\frac{\eta \Gamma(\tau-\sigma+1)}{\tau c^\tau \Gamma(1-\sigma)}.
\label{eq:C1}
\end{align}

Additionally, we have
\begin{align}
\Pr(|X_t|>x)\obsim{x\to\infty} C_1 t x^{-2\tau}.
\end{align}
The increments have therefore heavy tails with power-law exponent $2\tau$.

\paragraph{Simulation of the increments.} As
$$X_t \overset{d}{=} \sqrt{Z_t}\epsilon_t$$
one can simulate increments exactly by sampling $Z_t$ from \eqref{eq:Zsum} and $\epsilon_t\sim\mathcal N(0,1)$.

\paragraph{Self-decomposability.} The self-decomposability of $X$ follows from the self-decomposability of the subordinator $Z$~\cite[Theorem 1]{Sato2001}.

\paragraph{Interpretability of the parameters.} In summary, each of the four parameters governs a different property of the NGGP process.
\begin{itemize}[noitemsep,nolistsep]
\item $\eta>0$ is a time-scaling parameter: if $(X_t)_{t\geq 0}$ is a NGGP process with parameters $(1,\sigma,\tau,c)$, then $(X_{\eta t})_{t\geq 0}$ is a NGGP process with parameters $(\eta,\sigma,\tau,c)$;
\item $\sqrt{c}>0$ is an inverse-scale parameter: if $(X_t)$ is a NGGP process with parameters $(\eta,\sigma,\tau,1)$, then $(Z_t/\sqrt{c})$ is a NGGP process with parameters $(\eta,\sigma,\tau,c)$;
\item $\sigma\in(-\infty,1)$ tunes the activity of the jumps; the process is finite-activity if $\sigma<0$ and infinite-activity otherwise, with  corresponding BG index $\beta=\max(2\sigma,0)$;
\item $2\tau$ is the power-law exponent, controlling the tails of the distribution, with $\Pr(X_t>z)\obsim{z\to\infty} tC_1 z^{-2\tau}$ for some constant $C_1$.
\end{itemize}

\subsection{Generalisations}

\subsubsection{More general subordinators}

One could consider more generally a L\'evy intensity of the form
\begin{align}
\rho(w)=\frac{\eta}{c^\sigma \Gamma(1-\sigma)}w^{-1-\sigma}\left (wc\int_0^1 u^{1-\sigma}h(u)e^{-ucw}du + h(1)e^{-cw}\right )\label{eq:generalLevy}
\end{align}
where $h:(0,1]\to(0,\infty)$ is a differentiable function which satisfies
\begin{align}
\int_0^1 h(u)du<\infty~~\text{ and }~~h(u)\obsim{u\to 0} u^{\tau-1}.
\end{align}
For $\tau>\sigma$, we have $u^{1-\sigma}h(u)\to 0$ and \eqref{eq:generalLevy} takes the alternative form
\begin{align}
\rho(w)=\frac{\eta}{c^\sigma \Gamma(1-\sigma)}w^{-1-\sigma}\int_0^1 (u^{1-\sigma}h(u))'e^{-uwc}du.\label{eq:generalLevy2}
\end{align}
The proposed subordinator \eqref{eq:GBFRYintensity} is obtained as a special case when $h(u)= u^{\tau-1}$. \smallskip

The L\'evy process is finite activity for $\sigma<0$ and infinite-activity for $\sigma\geq 0$. Using Karamata's theorem for regularly varying functions, the tail L\'evy intensity $\overline \rho$ of the L\'evy intensity is regularly varying at 0 with BG index $\alpha=\max(0,\sigma)$ and at infinity with tail index $\tau$.
The L\'evy intensity \eqref{eq:generalLevy} takes the form of a sum of a compound Poisson intensity and a generalised gamma intensity. The compound Poisson intensity can be written as
\begin{align*}
\rho_1(w)&=\frac{\eta}{c^{\sigma-1} \Gamma(1-\sigma)}\int_0^1 h(u) u(uw)^{-\sigma} e^{-uw}du\\
&=\frac{\eta(1-\sigma)\int_0^1 h(u)du}{c^{\sigma-1} }\int_0^1 \frac{h(u)}{\int_0^1 h(v)dv} u\rho_\text{GG}(uw;\sigma-1,c)du
\end{align*}
which is a mixture of (finite-activity) generalised gamma processes. It follows that if $Z$ is a subordinator with L\'evy intensity \eqref{eq:generalLevy}, we have
$$Z_t\overset{d}{=}Z_{t,1}+Z_{t,2}$$
where $Z_{t,1}\sim \GG(\frac{t\eta h(1)}{c^\sigma},\sigma,c)$ is an exponentially tilted stable random variable and
$$
Z_{t,2}\overset{d}{=}\sum_{j=1}^{\widetilde K_t} \widetilde G_{t,j}\widetilde U_{t,j}
$$
where $\widetilde K_t\sim\Poisson(\eta t\int_0^1 h(u)du )$, $\widetilde G_{t,j}\sim\Gammadist(1-\sigma,c)$ and $1/\widetilde U_{t,j}$ have probability density function $\frac{h(u)}{\int_0^1 h(v)dv}$.
Finally, if $u^{1-\sigma}h(u)$ is monotone increasing, with $\tau>\sigma$, then Equation~\eqref{eq:generalLevy2} implies that the L\'evy process is a tempered stable process if $\sigma>0$, and a generalised gamma convolution if $\sigma=0$; it is therefore self-decomposable.

%\fc{add discussion}

\subsubsection{Subordinated fractional Brownian motion}

Many of the properties of the GGP process extend to the NGGP due to the self-similarity properties of Brownian motion. A stochastic process $(X_t)_{t\geq 0}$ with $X_0=0$ almost surely is self-similar if there exists an index $H>0$ such that for all $a>0$, $$(X_{at})_{t\geq 0}\overset{d}{=}(a^H X_t)_{t\geq 0}.$$ The Brownian motion is self-similar with index $H=1/2$. Another popular class of self-similar processes are fractional Brownian motions.  A fractional Brownian motion $(F_t)_{t\geq 0}$ is a zero-mean Gaussian process with covariance
$$
R(t,s):=\mathbb E[F_t F_s]=\frac{1}{2}(t^{2H}+s^{2H}-|t-s|^{2H})
$$
Brownian motion is obtained as a special case for $H=1/2$. The process is self-similar with index $H$ and has stationary Gaussian increments. The increments are negatively correlated if $H<1/2$, are independent if $H=1/2$ and positively correlated if $H>1/2$.

 Most of the \revision{properties} described in Section \ref{sec:NGGP} can be similarly derived in the more general case where the Brownian motion is replaced by a fractional Brownian motion. Let $(Z_t)$ be a GGP process with parameters $(\eta,\sigma,\tau,c)$ and $(F_t)$ a fractional Brownian motion with index $H>\max(0,\sigma)/2$. The subordinated fractional Brownian process $$X_t=F_{Z_t}$$ satisfies the self-similar property
$$
X_t\overset{d}{=}Z_t^H \epsilon_t
$$
where $\epsilon_t\sim \mathcal N(0,1)$. It follows that
\begin{align}
\Pr(|X_t|>x)\obsim{x\to\infty} C_1 t x^{-\tau/H}
\end{align}
where the constant $C_1$ is defined in Equation \eqref{eq:C1}. Increments of the process at times $t_1,\ldots,t_n$ can be simulated exactly by first simulating $Z_{t_k}-Z_{t_{k-1}}\overset{d}{=}Z_{t_k-t_{k-1}}$ using \eqref{eq:Zsum}, for $k=1,\ldots,n$ then, conditional on $(Z_{t_1},\ldots,Z_{t_n})$ simulate
$$
X_{1},\ldots,X_n|Z_{t_1},\ldots,Z_{t_n} \sim \mathcal N(0, (R(Z_{t_i},Z_{t_j}))_{1\leq i,j\leq n}).
$$

\subsection{Comparison to other models and discussion}

\paragraph{Comparison to Ayed et al.} The (normalised) GGP process was introduced by \cite{Ayed2019} as a prior for random probability measures with power-law properties, and applied to the modeling of word frequencies. \cite{Ayed2019} introduced the form \eqref{eq:GBFRYintensity2} which is only valid for $\tau>\sigma$. The alternative form \eqref{eq:GBFRYintensity} we introduce here allows to deal with the case $0<\tau\leq\sigma$ as well; in particular, one obtains the stable process as a particular case. \cite{Ayed2019} showed that the tail L\'evy intensity of the GGP is regularly varying at 0 and infinity and deduced the asymptotic behaviour of large and small jumps. Here we derive a number of additional important properties of the process and of the distribution of the increments. We show that it is decomposable, and crucially, that one can sample exactly the increments at any time scale.  \cite{Ayed2019} used the name GBFRY process for the process, due to its form similar to the form of the GBFRY distribution (see Section B in the Appendix); however, as it is customary to give the same name to the process and to the distribution of the increments, which are not GBFRY distributed, we prefer here to use the name generalised gamma-Pareto.

\paragraph{Tempered stable process and generalised gamma convolutions.} If $\tau>\sigma>0$, the subordinator falls in the general class of tempered stable processes, introduced by~\cite{Rosinski2007}. Noting that $\gamma(\tau-\sigma,cw)=w^{\tau-\sigma+1}\int_0^c u^{\tau-\sigma}e^{-wu}du$, the model~\eqref{eq:GBFRYintensity} is indeed of the form $w^{-1-\sigma}q(w)$ where the so-called tempering function $q$ is given by
\begin{equation}
q(w)=\frac{\eta (\tau-\sigma)}{c^{\tau}\Gamma(1-\sigma)}\int_0^c u^{\tau-\sigma-1}e^{-wu}du.\label{eq:qtilting}
\end{equation}
By Bernstein's theorem, the function $q$ is completely monotone.

For $\sigma=0$, the subordinator belongs to the class of generalised gamma convolutions~\citep{Thorin1977,Bondesson1992,James2008}, of the form $w^{-1}\int_0^\infty e^{wu}U(du)$ with Thorin measure $U(du)= \frac{\eta \tau}{c^{\tau}} u^{\tau-1}1_{u\in(0,c)}du$.

The subordinator, for any $\sigma$, also falls into the extended Thorin class described by \cite{Grigelionis2007}, see also the discussion in Section 1.8 in \citep{James2008}.

\paragraph{Comparison to other models.} As mentioned in the introduction, a number of different L\'evy processes have been proposed in the literature. While each process can capture some range of the different tail and jump behaviour, none of them is flexible enough to capture the whole range of tail and jump-activity indices.
Variance gamma, normal inverse Gaussian, exponentially tilted stable and tempered stable process do not capture heavy tails; the normal stable process has infinite variance, and the same parameter tunes the activity of the jumps and the BG index; for generalised hyperbolic process, the BG index is fixed to 1.

A drawback of the proposed model is that, contrary to popular models such as the variance gamma or normal inverse Gaussian processes, the increments $X_t$ do not have an analytical probability density function $f_t$. This is balanced however by the fact that one can sample exactly from the distribution of the increments, and one can therefore resort to likelihood-free methods for posterior inference, as described in the next section. Table~\ref{tab:comparison} summarises the properties of the different models. Note that, as mentioned in~\cite[Section 4.6]{Cont2004}, the generalised hyperbolic and student $t$ are not closed under convolution, and so there is no analytic expression for $f_t$ at any given time $t>0$, which may be an issue if data are sampled irregularly.

Some interesting connections can be drawn with other classes of stochastic processes. If $\tau>\sigma$, due to the mixture form $\eqref{eq:GBFRYintensity3}$, the GGP distribution arises as the marginal distribution of a quantile clock process (see Theorem 3.1 by~\cite{James2011}) with parameters $(R,L)$ where $R$ is a Pareto random variable and $L$ a GG subordinator,

\newcommand{\cell}[2]{\parbox[c]{#1cm}{\centering #2}}

\begin{table}
\setlength{\tabcolsep}{3pt}
\scriptsize
\centering
\caption{Comparison between different L\'evy processes. VG: Variance Gamma; NIG: Normal inverse Gaussian; NGG: Normal generalised gamma; NS: Normal Stable; TS: Tempered Stable; St: Student t; GH: Generalised hyperbolic}
\begin{tabular}{@{}ccccccc@{}}
\toprule
Model & Heavy tails & \cell{2}{Finite 2nd moment} & \cell{1.2}{BG index} & \cell{2}{Tractable $f_t$ for any $t$} &  \cell{2.5}{Exact simulation from $f_t$ for any $t$} \\
\midrule
VG & No & Yes & $\beta=0$ & Yes &  Yes \\
NIG & No & Yes & $\beta=1$ & Yes &  Yes \\
NGG & No & Yes & $\beta\in(0,2)$ & No &  Yes\\
NS & Yes, $2\tau\in(0,2)$ & No & $\beta=2\tau\in(0,2)$ & No &  Yes\\
TS & No & Yes & $\beta \in [0,2)$ & No &  Yes \\
St & Yes, $2\tau\in(0,\infty)$ & Yes if $2\tau>2$ & $\beta=1$ & No &  No \\
GH & Depends & Depends & $\beta=1$ & No  & No\\
NGGP & Yes, $2\tau\in(0,\infty)$ & Yes if $\tau>1$ & $\beta=2\max(0,\sigma)\in[0,2)$ & No &  Yes \\
\bottomrule
\end{tabular}\label{tab:comparison}
\vspace{-2em}
\end{table}

\section{L\'evy-driven stochastic volatility models}
\label{sec:modelfinance}

Let $S_{t}$ denote the price of a financial asset, e.g. a market or a stock
index, at time $t$. Denote $X_t=\log\left  (\frac{S_t}{S_0}\right )$. Observations are obtained at fixed discrete times $t_{0}=0<t_{1}%
<t_{2}<\ldots t_{n},$ and we write, for $k=1,2\ldots$
\begin{equation}
Y_{k}:=\log\frac{S_{t_{k}}}{S_{t_{k-1}}}=X_{t_{k}}-X_{t_{k-1}}%
\end{equation}
the log-returns (or more shortly, called returns). Let $\Delta_k=t_k-t_{k-1}$ be the inter-arrival times between observations. We assume that
\begin{equation}
X_t=\mu_0 t + \mu_1 V^*_t + B_{V^*_t}
\end{equation}
where $\mu_0$ is the drift parameter, $\mu_1$ is the risk premium, $B_t$ is a Brownian motion, independent of the stochastic process $V^*_t$, which can be interpreted as the integrated stochastic volatility. For $k=1,\ldots,n$, let
$$
\overline V_k =V^*_{t_k}-V^*_{t_{k-1}}
$$
be the integrated stochastic volatility over the interval $(t_{k-1},t_k)$. The observations $(Y_1,\ldots,Y_n)$ are conditionally independent given $(\overline V_1,\ldots,\overline V_n)$, with
$$
Y_k\mid \overline V_k\sim \mathcal N(\mu_0\Delta_k+\mu_1\overline V_k,\overline V_k).
$$

We consider two different stochastic processes for the integrated volatility process $(V_t^*)$: a L\'evy process and a Ornstein-Uhlenbeck based model.

\subsection{Exponentiated L\'{e}vy process}

Assume that $(V^*_t)_{t\geq 0}$ is a subordinator with no drift with L\'evy intensity $\rho$ parameterised by a vector $\phi$. The integrated volatilities $(\overline V_1,\ldots,\overline V_n)$ are therefore conditionally independent, with
\begin{align}
\overline V_k\mid \phi&\sim F_{V^*_{\Delta_k}}
\end{align}
where $F_{V^*_t}$ denotes the distribution of $V^*_t$, with Laplace transform
$$
\int_0^\infty e^{-\lambda x} dF_{V^*_t}(x)=e^{-t\int_0^\infty (1-e^{-\lambda w})\rho(w)dw}.
$$
If $V^*_t$ is taken to be the GGP model with intensity \eqref{eq:GBFRYintensity}, then $B_{V^*_t}$ is a NGGP L\'evy process.

\subsection{Ornstein-Uhlenbeck based stochastic volatility model}
\label{sec:OUmodels}
We also consider a non-Gaussian Ornstein-Uhlenbeck based model~\citep{Barndorff-Nielsen2001} with
$$
V^*_t = \int_0^t V_t dt
$$
where the instantaneous stochastic volatility process $(V_{t})_{t\geq0}$ is
stationary and satisfies
\[
V_{t}=V_0 e^{-\lambda t}+ \int_{0}^{t}e^{\lambda(s-t)}dZ_{s}
\]
for some $\lambda>0$ and some background driving L\'evy process $Z_t$ with L\'evy measure $\widetilde \rho$. Additionally, for any $t>0$ the random variable $V_{t}\sim F$ is infinite-divisible and self-decomposable with Laplace transform
\[
\mathbb{E}[e^{-\vartheta V_{t}}]=e^{-\int_{0}^{\infty}(1-e^{-\vartheta w})\rho
(w)dw}
\]
where $\widetilde\rho$ and $\rho$ are related by the
expression
\begin{align*}
\widetilde\rho(w)=-\rho(w)-w\rho^{\prime}(w).
\end{align*}

To define the model, one can either define the mean measure $\widetilde \rho$ of the
subordinator $Z_{t}$, or choose the stationary (self-decomposable) distribution $F$ of $V_{t}$, hence $\rho$. In practice, the second approach is often chosen; examples
include the gamma~\citep{Roberts2004,Griffin2006,Fruehwirth-Schnatter2009},
generalized inverse Gaussian~\citep{Gander2007}, and exponentially tilted stable
distributions~\citep{Gander2007,Andrieu2010} as marginals.
The integrated stochastic volatilities over the interval $(t_{k-1},t_k)$ are obtained, for $k=1,\ldots,n$, by
\begin{align}
\overline{V}_{k}  &  =\int_{t_{k-1}}^{t_{k}}V_{t}dt=\lambda^{-1}\left(
Z_{\lambda t_{k}}-V_{t_{k}}-\left(  Z_{\lambda t_{k-1}}-V_{t_{k-1}}\right)
\right)\label{eq:OUmodel1}
\end{align}
where $\binom{V_{t_{k}}}{Z_{\lambda t_{k}}}$ follows a linear dynamic model with $Z_0=0$, $V_0\sim F$, and for $k=1,\ldots,n$,
\begin{align}
\binom{V_{t_{k}}}{Z_{\lambda t_{k}}}=\binom{e^{-\lambda\Delta_{k}}V_{t_{k-1}}%
}{Z_{\lambda t_{k-1}}}+\varepsilon_{k},~~~\text{ with }\varepsilon_{k}\overset{d}{=}\binom{e^{-\lambda\Delta_{k}%
}\int_{0}^{\Delta_{k}}e^{\lambda t}dZ_{\lambda t}}{\int_{0}^{\Delta_{k}%
}dZ_{\lambda t}}. \label{eq:OUmodel2}
\end{align}

Exact simulation of $(\overline V_1, \ldots,\overline V_n)$ from the model defined by Equations~(\ref{eq:OUmodel1}-\ref{eq:OUmodel2})  requires to be able to simulate from $F$ and simulate the independent random variables $(\varepsilon_1,\ldots,\varepsilon_n)$. We describe two models where exact simulation is possible.

\subsubsection{Model with gamma marginal distribution}

A classical choice~\citep{Barndorff-Nielsen2001,Roberts2004,Griffin2006,Fruehwirth-Schnatter2009} is to take $F=\Gammadist(\eta,c)$ as marginal distribution for $V_t$. This corresponds to
\begin{align}
\rho(w)=\eta w^{-1}e^{-cw},~~~\widetilde \rho(w)=\eta ce^{-cw}.
\end{align}
The background driving L\'evy measure $(Z_t)$ is therefore finite-activity, and one can sample exactly the state noise $\varepsilon_k$ as follows.

\begin{enumerate}[noitemsep,nolistsep]
\item Simulate $N_k\sim\text{Poisson}(\eta\lambda\Delta_{k})$.

\item For $j=1,\ldots,N$, simulate $E_{kj}\sim\text{Exp}(c),\theta_{kj}\sim
U(0,\Delta_{k})$.

\item Set $\varepsilon_{k}= \binom{e^{-\lambda\Delta_{k}}\sum_{j=1}^{N_k} e^{\lambda
\theta_{kj}} E_{kj}}{\sum_{j=1}^{N_k} E_{kj}}$.
\end{enumerate}

\subsubsection{Model with GGP marginal distribution}
\label{subsubsec:complex_nggp}

Let $\tau=\sigma> 0$ or $\tau>\sigma\geq 0$. As shown in Section~\ref{sec:GGP}, $\GGP(\eta,\sigma,\tau,c)$ is self-decomposable. If $V_t$ has marginal $F=\GGP(\eta,\sigma,\tau,c)$ distribution, this corresponds to $\rho$ be defined by Equation~\eqref{eq:GBFRYintensity}, and the L\'evy intensity $\widetilde \rho$ of the background driving L\'evy intensity $(Z_t)$ is given by Equation~\eqref{eq:tilderhoGGP}. For $\sigma>0$, $(Z_t)$ is infinite-activity, and one needs to resort to numerical methods to approximately sample $(\varepsilon_k)$. This could be done by using the representation of the process as a sum of GG process and a finite activity process as described in Section~\ref{sec:GGP}, and using a truncated series representation for simulating the GG process.

We focus here on the case $\sigma=0$, where $\widetilde \rho$ simplifies to
\begin{align}
\widetilde \rho(w)=\frac{\eta\tau}{c^\tau}w^{-1-\tau}\gamma(\tau+1,cw)
\end{align}
with $\int_0^\infty \widetilde \rho(w)dw=\eta$, and the background L\'evy process $(Z_t)$ is therefore finite-activity.
We can therefore simulate $\varepsilon_{k}$ exactly as follows.

\begin{enumerate}[noitemsep,nolistsep]
\item Simulate $N_k\sim\text{Poisson}(\eta\lambda\Delta_{k})$

\item For $j=1,\ldots,N$, simulate $E_{kj}\sim\text{Exp}(c),U_{kj}%
\sim\Pareto(\tau,1),\theta_{kj}\sim U(0,\Delta_{k})$

\item Set $\varepsilon_{k}= \binom{e^{-\lambda\Delta_{k}}\sum_{j=1}^{N_k} e^{\lambda
\theta_{kj}} E_{kj}U_{kj}}{\sum_{j=1}^{N_k} E_{kj}U_{kj}}$
\end{enumerate}

which is similar to the model with gamma marginals, with $E_{kj}U_{kj}$ in place
of $E_{kj}$.

\subsection{Posterior Inference}

Let $\phi$ denote the set of unknown parameters of both models. That is, $\phi$ includes the drift and risk premium parameters $\mu_0$ and $\mu_1$, the parameters of the L\'evy intensity and, for the Ornstein-Uhlenbeck based model, the discounting factor $\lambda>0$. Let $\pi(\phi)$ be some prior density. We aim at approximating the posterior density $\pi(\phi\mid y_1,\ldots,y_n)$. The marginal likelihood takes the form
\begin{equation}
p(y_1,\ldots,y_n\mid \phi)=\int_{\mathbb R_+^n} \left [\prod_{k=1}^n p(y_k\mid \overline v_k,\mu_0,\mu_1) \right ]d\boldsymbol F_n(\overline v_1,\ldots,\overline v_n)\label{eq:maginallikelihood}
\end{equation}
where
$$
\boldsymbol F_n(\overline v_1,\ldots,\overline v_n)=\Pr\left (\overline V_1\leq\overline v_1,\ldots, \overline V_n\leq\overline v_n\mid \phi\right)
$$
denotes the joint cumulative distribution function of the integrated variances. In the exponentiated L\'evy process, we have
$$
\boldsymbol F_n(\overline v_1,\ldots,\overline v_n)=\prod_{k=1}^n F_{V_{\Delta_k}}(\overline v_k).
$$
If $F_{V_t}$ does not admit a tractable probability density function, as it is the case for the proposed GGP model, neither $\boldsymbol F_n$ nor $p(y_1,\ldots,y_n\mid \phi)$ are tractable, preventing the implementation of a Metropolis-Hastings Markov chain Monte Carlo algorithm. The same applies for the Ornstein-Uhlenbeck model.\smallskip

We therefore propose to use a pseudo-marginal Markov chain Monte Carlo (MCMC) algorithm~\citep{Beaumont2003,Andrieu2009}, which only requires to simulate from $\boldsymbol F_n$. The pseudo-marginal algorithm replaces the untractable marginal likelihood \eqref{eq:maginallikelihood} by an unbiased estimator, yet admitting the posterior distribution of interest as invariant distribution. Let $q$ denote some proposal distribution for the parameters. At iteration $i$ of the algorithm, we have
\begin{enumerate}[noitemsep,nolistsep]
\item Sample $\phi^* \mid \phi^{(i-1)}\sim q(\cdot|\phi^{(i-1)})$
\item Compute an unbiased estimate $\widehat p(y_1,\ldots,y_n\mid \phi^* )$
\item With probability $$\min\left (1,\frac{\widehat p(y_1,\ldots,y_n\mid \phi^* )\pi(\phi^*)q(\phi^{(i-1)}\mid \phi^*)}{\widehat p(y_1,\ldots,y_n\mid \phi^{(i-1)} )\pi(\phi^{(i-1)})q(\phi^*|\phi^{(i-1)})}\right )$$
    set $\phi^{(i)}=\phi^*$ and $\widehat p(y_1,\ldots,y_n\mid \phi^{(i)})=\widehat p(y_1,\ldots,y_n\mid \phi^*)$. \\Otherwise, set $\phi^{(i)}=\phi^{(i-1)}$ and $\widehat p(y_1,\ldots,y_n\mid \phi^{(i)})=\widehat p(y_1,\ldots,y_n\mid \phi^{(i-1)})$.
\end{enumerate}

In the exponential L\'evy model, an unbiased estimator can be obtained via Monte Carlo approximation
$$
\widehat p(y_1,\ldots,y_n\mid \phi)=\prod_{k=1}^n  \frac{1}{n_p}\sum_{j=1}^{n_p} p(y_k | \overline v_k^{(j)},\mu^{(j)}_0,\mu^{(j)}_1)
$$
where $\overline v_k^{(j)}\sim F_{V_{\Delta_k}}$ for $k=1,\ldots,n$ and $j=1,\ldots,n_p$, with $n_p$ the number of Monte Carlo samples (called particles thereafter).

In the Ornstein-Uhlenbeck model, the marginal likelihood can be approximated with a (bootstrap) sequential Monte Carlo algorithm, a standard inference technique for this class of models~\citep{Andrieu2010,Jasra2011,Chopin2013}. The resulting algorithm is known in this case as a particle marginal Metropolis-Hastings algorithm~\citep{Andrieu2010}.

\section{Experiments}
\label{sec:experiments}\label{}

\paragraph{Priors.} In all the experiments, the drift $\mu_0$ and the premium $\mu_1$ parameters are set to zero. For the GGP model, we assume that we are in the infinite-activity regime, with $\sigma\geq 0$, and with finite variance, hence $\tau>1$. The priors are set as follows:
\begin{align*}
\eta\sim\Gammadist(0.1,0.1),~c\sim\Gammadist(0.1,0.1),~~(\tau-1)\sim\Gammadist(1,1),~~\sigma\sim\Unif(0,1).
\end{align*}
The more informative prior for $\tau$ reflects the empirical evidence that, for many financial datasets, the power-law exponent ($2\tau$ for the NGGP) is in the range $(2,5)$~\citep[Section 7.1]{Cont2004}. For the Ornstein-Uhlenbeck model, we additionally set $\lambda\sim\Gammadist(0.1,0.1)$.

\paragraph{Software.} To fit both models, we use the Particles Library\footnote{\url{https://github.com/nchopin/particles}}, which allows to perform posterior inference in state-space models using particle MCMC algorithms.
The code and datasets can be found on the anonymous github repository\footnote{\url{https://github.com/OxCSML-BayesNP/NGGP}}

\subsection{Exponentiated L\'evy model}

\subsubsection{Simulated datasets}

\revision{We first focus on the range of parameters that corresponds to processes with infinite activity ($\sigma \geq 0$) and finite variance ($\tau > 1$) as empirical evidence indicate that this is appropriate for financial applications \cite[Section 7.1]{Cont2004}.
We generate a synthetic dataset of $n=5\,000$ unit-spaced observations from the NGGP model, with parameters $\eta = 1$, $\sigma = 0.6$, $\tau = 3$ and $c=1$. The priors are as described at the beginning of this section. We run three independent MCMC chains with $n_{mcmc}=10\,000$ iterations each, of which $5\,000$ iterations are used for burn-in. The number of particles to compute the marginal likelihood estimates is set to $n_p=4\,000$. In Figure~\ref{fig:simulated_post} we report histograms and trace plots of the posterior samples for each of the four parameters. Trace plots suggest the convergence of the MCMC algorithm.}

\revision{
To give a more complete picture, we also investigate whether the parameters can be recovered in the other three quadrants finite/infinite activity and finite/infinite variance. We consider for this priors with support in $(0,\infty)$ for $\tau$ and $(-\infty,1)$ for $\sigma$. The estimated parameters and $95\%$ credible intervals are reported in Table~\ref{tab:four_quadrants}.
}

\begin{figure}[h]
\centering
\includegraphics[width=0.49\linewidth]{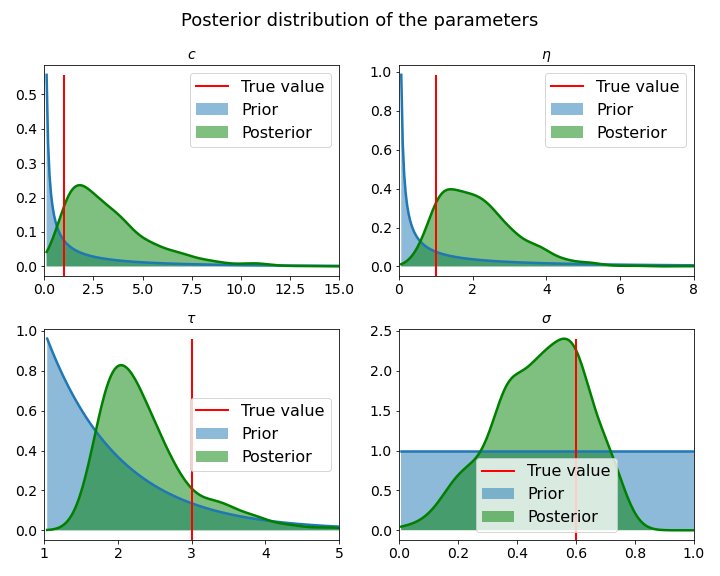}
\includegraphics[width=0.49\linewidth]{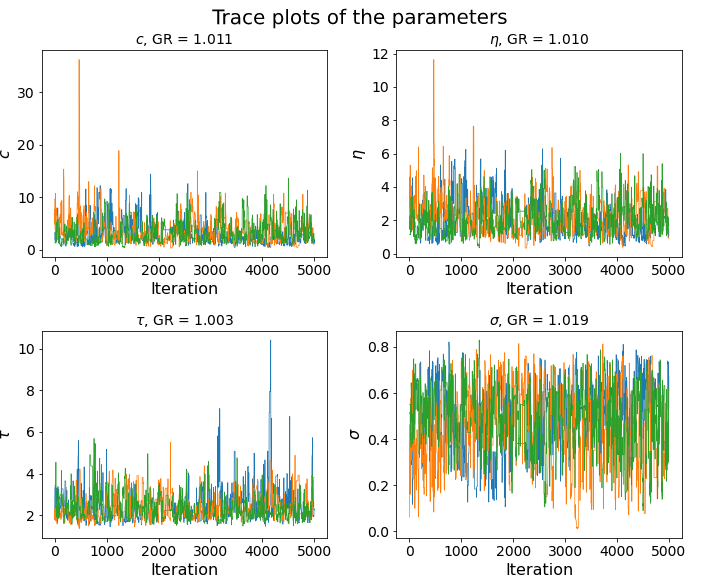}
\caption{Histograms (left) and trace plots (right) of the posterior distributions of the parameters on the simulated data experiment. The blue line represents the value of the parameter used to generate the data. We also report the Gelman-Rubin scores to assess convergence of the chains (the lower the better, the empirical threshold for convergence is $1.1$)}\label{fig:simulated_post}
\end{figure}

\begin{table}[h]
\caption{Recovering the finite-activity/infinite-variance, finite-activity/finite-variance, and infinite-activity/infinite-variance settings. Posterior mean and 95$\%$ credible regions (in parenthesis) are reported.}\label{tab:four_quadrants}
\begin{center}
\begin{tabular}{c| c|c|c|c}
    \hline
    \multirow{2}{*}{Regime} & \multicolumn{2}{c|}{True values} & \multicolumn{2}{c}{Estimated values}  \\
    & $\sigma$ & $\tau$ & $\sigma$ & $\tau$  \\
    \hline
    finite-activity/finite-variance & $-1$  &  $3$ & $-1.00$ $(-1.23, -0.82)$ &  $3.03$ $(2.01, 5.50)$ \\
    finite-activity/infinite-variance & $-1$ & $0.8 $ &  $-1.00$ $(-1.22, -0.80)$ & $0.77$ $(0.71, 0.84)$  \\
    infinite-activity/infinite-variance & $0.6$ & $0.8 $ & $0.68$ $(0.40, 0.87)$ & $0.75$ $(0.64, 0.87)$  \\
    \hline
\end{tabular}
\end{center}
\vspace{-2em}
\end{table}

Further details and additional synthetic experiments, exploring the sensitivity to the choice of the prior, are provided in Appendix D.

%Parameters: gbfry T 1500 eta 4.0 tau 2.5 sigma 0.2 c 2.0

\subsubsection{Real-world datasets.}

\paragraph{Dataset.} We consider a dataset composed of the time-series of the stock prices of six large technology companies: Apple, Amazon, Facebook, Google, Microsoft and Netflix. The data are sampled every minute from the 10th of July 2019 until the 22nd of January 2020, with approximately $50\,000$ time points. We subsample $1\,500$ observations as training data to estimate the parameters of each model, and use the rest of the observations as test data.

\paragraph{Other models.} We compare the fit of the NGGP model to classical L\'evy processes on the first dataset. The models compared are the normal-GG (also known as normal-tempered stable or normal-exponentially tilted stable model), with parameters $\eta$, $\sigma$ and $c$; two special cases of this model, the variance-gamma ($\sigma=0$) and normal-inverse-Gaussian ($\sigma=0.5$); the generalised hyperbolic (GH) model, with four parameters, and the student model, a special case of the GH model  with two parameters. We use vague $\Gammadist(0.1,0.1)$ priors on all parameters, except for the parameter $\sigma$ in the GG model, where a uniform prior on $[0,1]$ is used as for the GGP, and for the degrees of freedoms $\nu_0$ in the student t model, which tunes the power-law tail, where $(\nu_0/2-1)\sim \Gammadist(1,1)$ to reflect the prior assumptions on the tails (as for $\tau$ in the NGGP). Note that we can compare here to the GH and student models as the observations are equally spaced.

\paragraph{Results.} We run 3 MCMC chains in parallel, with $5\,000$ iterations ($2\,500$ burn-in) and $1\,500$ particles. The estimated parameters and 95\% credible intervals for the parameters of the NGGP are reported in Table~\ref{tab:post_simple_tech}. The posterior mean for $\tau$ is around $2$; this corresponds to a power-law exponent for $X_t$ of around $4$ which is in concordance with empirical observations~\citep[Section 7.1]{Cont2004}. One exception is the Amazon stock, where $\tau$ is closer to $1$, indicating a heavier tail. We first compare the models using the Kolmogorov-Smirnov (KS) statistics between the posterior predictive distribution and the empirical distribution of the test data. Results are reported in Table \ref{tab:simple_tech}. The KS statistics is rather insensitive to the tail of the distribution, and the performances are similar for most models considered. To investigate the goodness-of-fit to the tails of the distribution, we compare the ranked empirical squared log-return to their posterior predictive distribution. Both GH, VG and NIG, which have exponentially decaying tails, provide similar results, and we only report the results of the GH. Results for Apple, Amazon, Facebook in Figure~\ref{fig:simple_rank} (results for Google, Microsoft, Netflix are in Appendix \ref{sec:sensitivity}). We can see that the NGGP model successfully captures the behaviour of tails for the different datasets, while the GH fails to provide accurate posterior predictive for some datasets such as Facebook. The NS model, which has the same parameter to capture the jump-activity and the tail behaviour, underestimates the value of the tail exponent, and gives a poor fit.  The student t model tends to provide poor credible intervals, possibly due to the lack of flexibility of this two-parameter model.

\begin{table}
\caption{Posterior mean and 95\% credible interval for the four parameter of the NGGP model on the first dataset.}
\label{tab:post_simple_tech}

\centering
\scriptsize
\setlength{\tabcolsep}{3pt}
\begin{tabular}{@{} c c c c c}
\toprule
Data & $\eta$ & $\sigma$ & $\tau$ & $c$  \\
\midrule
Apple & 0.51, (0.13, 1.21) & 0.44, (0.21, 0.61) &  1.82, (1.11, 3.51) & 1.14, (0.13, 4.26) \\
Amazon & 0.64, (0.19, 1.40) & 0.39, (0.14, 0.58) & 1.18, (1.15, 3.59) & 1.46, (0.22, 4.61) \\
Facebook & 0.85, (0.34, 1.46) & 0.25, (0.04, 0.51) & 1.97, (1.25, 3.90) & 1.86, (0.38, 4.63) \\
Google & 0.18, (0.02, 0.71) & 0.64, (0.50, 0.73) & 1.93, (1.08, 4.16) & 0.34, (0.01, 1.98) \\
Microsoft & 0.27, (0.05, 0.81) & 0.55, (0.39, 0.66) & 1.98, (1.11, 4.35) & 0.56, (0.05, 2.51) \\
Netflix & 0.21, (0.07, 0.45) & 0.54, (0.44, 0.63) &  2.55, (1.25, 5.14) & 0.29, (0.06, 0.90)  \\
\bottomrule
\end{tabular}
\end{table}

\begin{table}
\caption{Kolmogorov-Smirnov distance between the empirical distribution of the test and the posterior predictive for different models on the first tech companies dataset (the smaller the better).}
\label{tab:simple_tech}
\centering
\scriptsize
\setlength{\tabcolsep}{3pt}
\begin{tabular}{@{} c c c c c c c c}
\toprule
Data & NGGP & NGG & GH & NIG & NS & VG & Student \\
\midrule
Apple & 0.0194 & 0.0196 &  0.0194 & 0.0194 & 0.0196 & 0.0218 & 0.0196\\
Amazon & 0.0087 & 0.0087 & 0.0085 & 0.0085 & 0.0159 & 0.0145 & 0.0092 \\
Facebook & 0.0181 & 0.0182 & 0.0182 & 0.0183 & 0.0245 & 0.1413 & 0.0181 \\
Google & 0.0205 & 0.0209 & 0.0197 & 0.0193 & 0.0237 & 0.0848 & 0.0200 \\
Microsoft & 0.0285 & 0.0285 & 0.0287 & 0.0286 & 0.0286 & 0.1567 & 0.0289\\
Netflix & 0.0079 & 0.0080 &  0.0080 & 0.0084 & 0.0098 & 0.0162 & 0.0079 \\
\midrule
Mean & 0.0172 & 0.0173 & 0.0171 & 0.0171 & 0.0204 & 0.0726 & 0.0173 \\
\bottomrule
\end{tabular}

\end{table}

\begin{figure}
\centering
\begin{subfigure}[t]{0.24\textwidth}
\centering
\includegraphics[width=\linewidth]{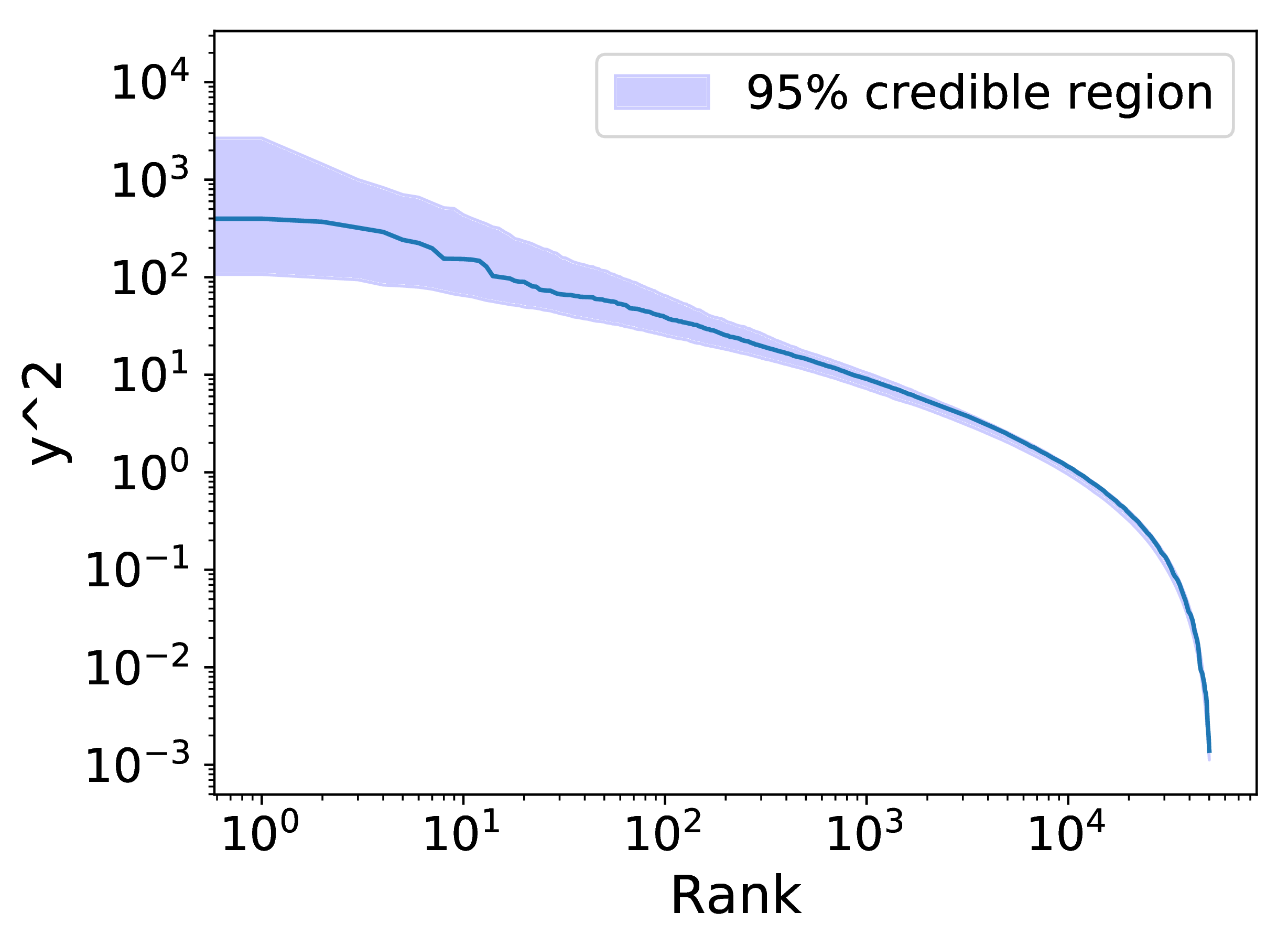}
%\caption[]{NGGP}
\end{subfigure}
\begin{subfigure}[t]{0.24\textwidth}
\centering
\includegraphics[width=\linewidth]{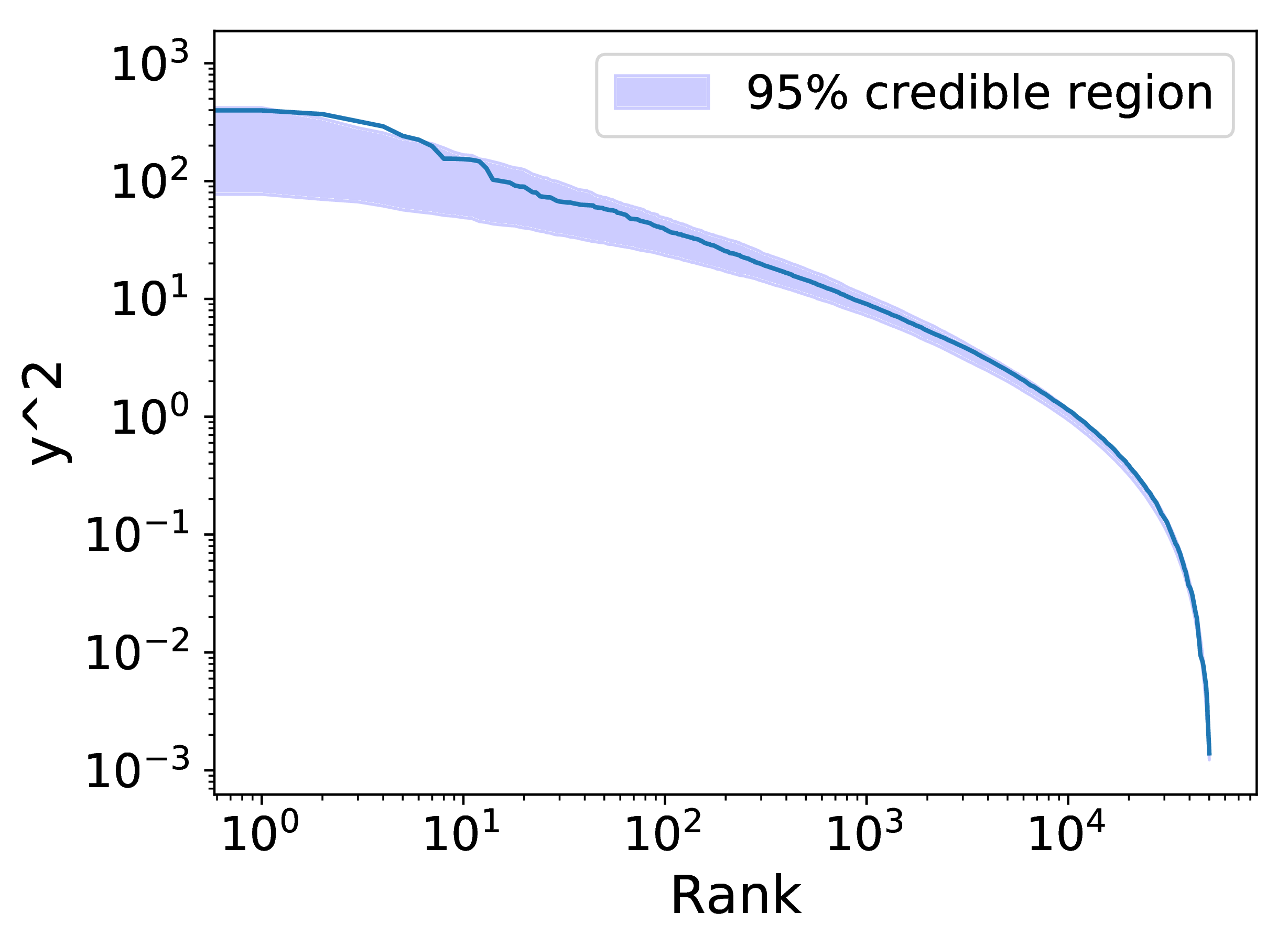}
%\caption{GHD}
\end{subfigure}
\begin{subfigure}[t]{0.24\textwidth}
\centering
\includegraphics[width=\linewidth]{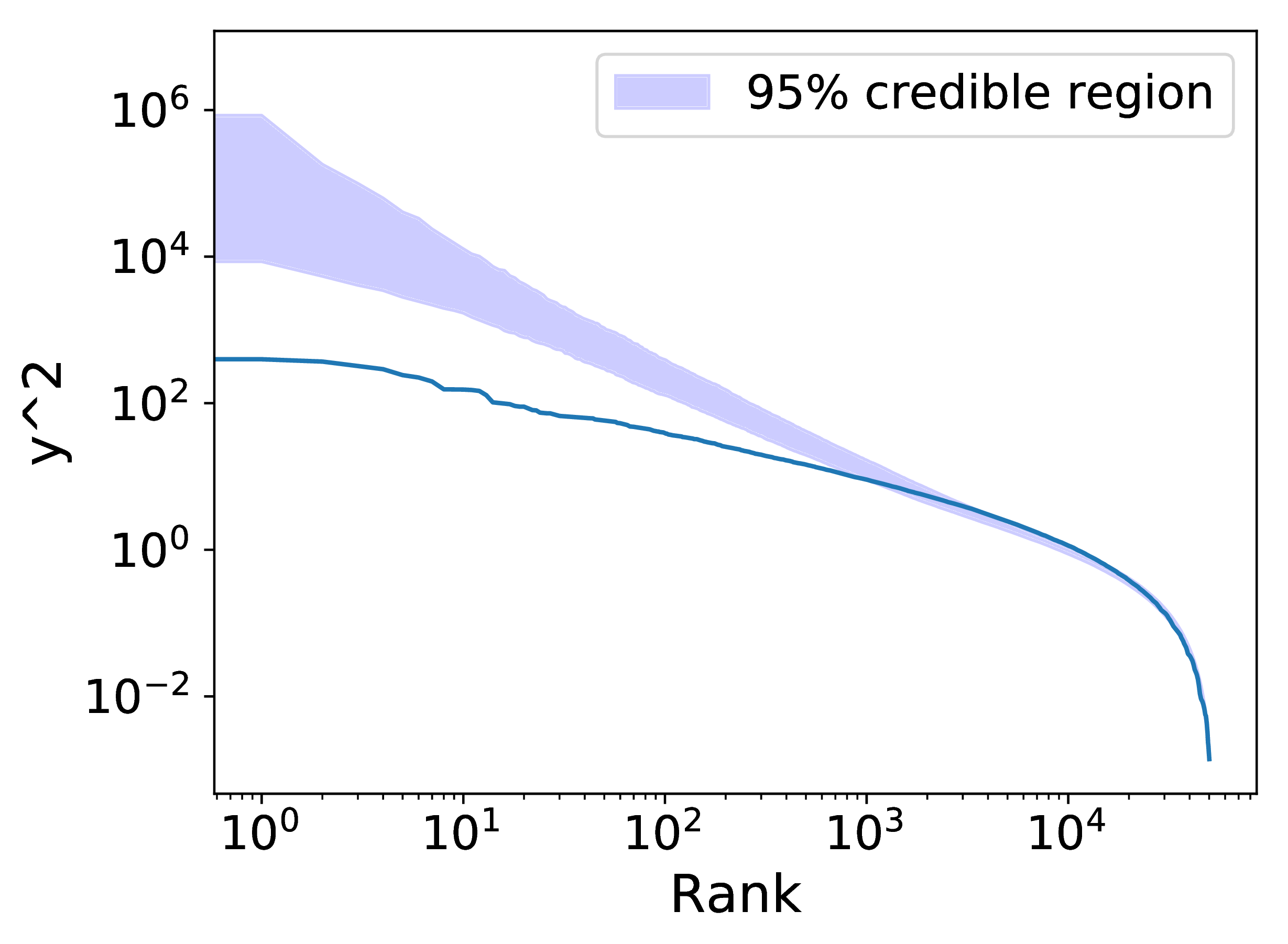}
%\caption{NS}
\end{subfigure}
\begin{subfigure}[t]{0.24\textwidth}
\centering
\includegraphics[width=\linewidth]{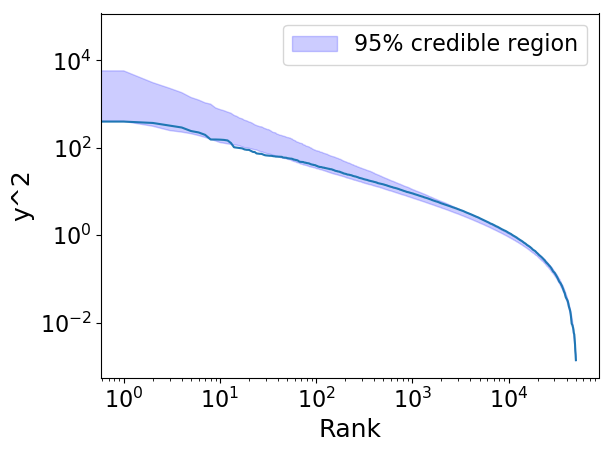}
%\caption{NS}
\end{subfigure}

\begin{subfigure}[t]{0.24\textwidth}
\centering
\includegraphics[width=\linewidth]{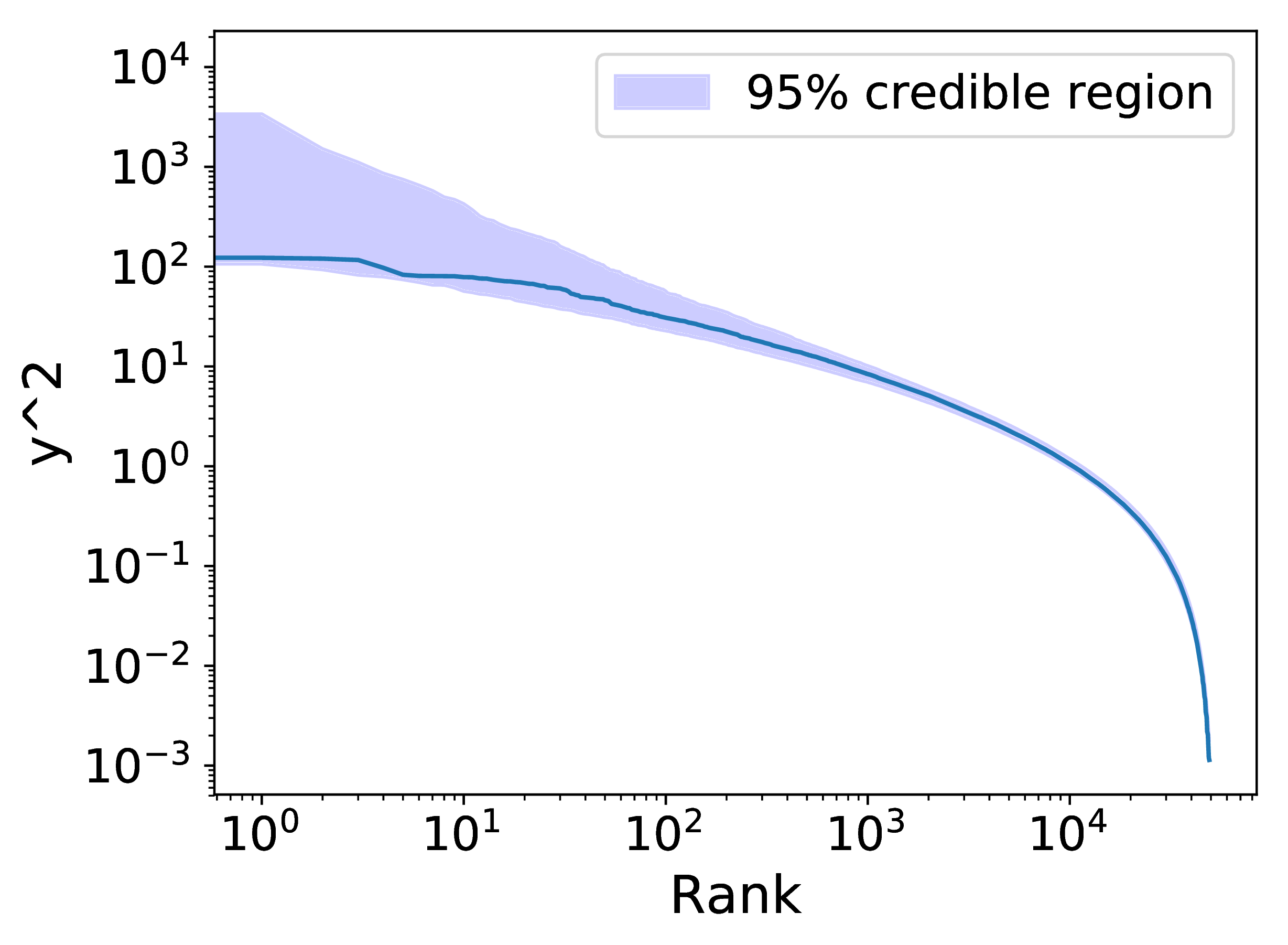}
%\caption{GBFRY}
\end{subfigure}
\begin{subfigure}[t]{0.24\textwidth}
\centering
\includegraphics[width=\linewidth]{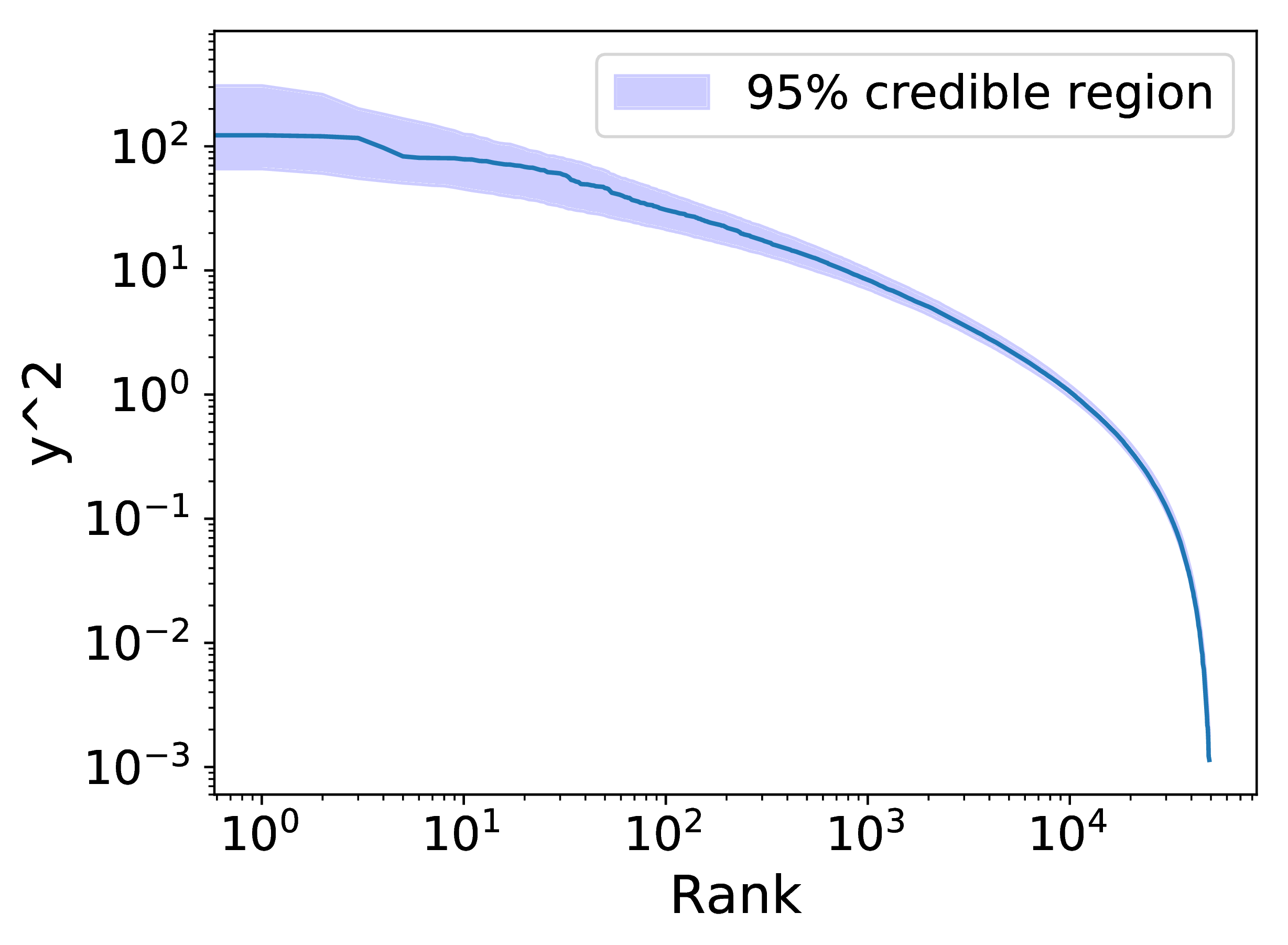}
%\caption{GHD}
\end{subfigure}
\begin{subfigure}[t]{0.24\textwidth}
\centering
\includegraphics[width=\linewidth]{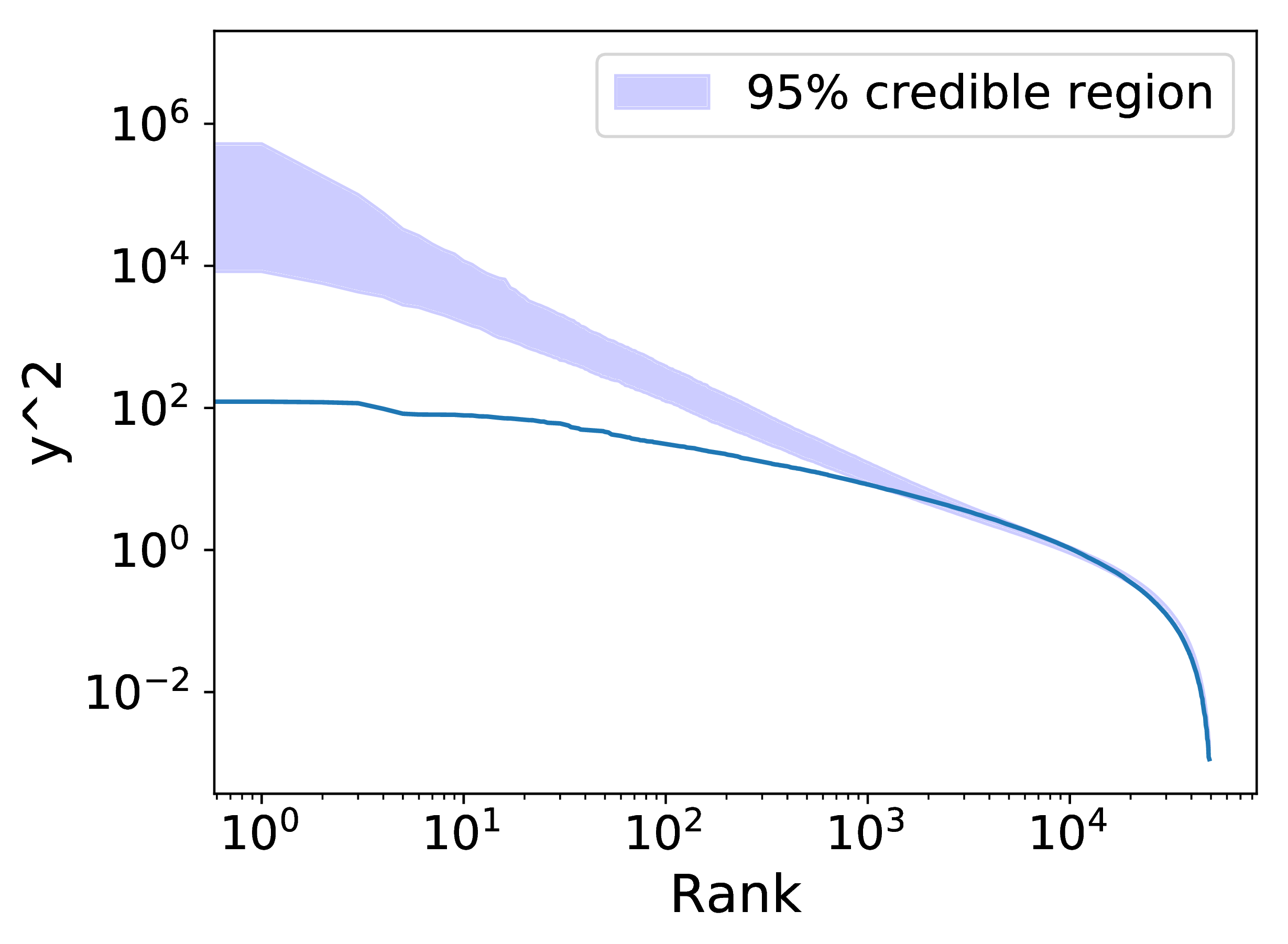}
%\caption{NS}
\end{subfigure}
\begin{subfigure}[t]{0.24\textwidth}
\centering
\includegraphics[width=\linewidth]{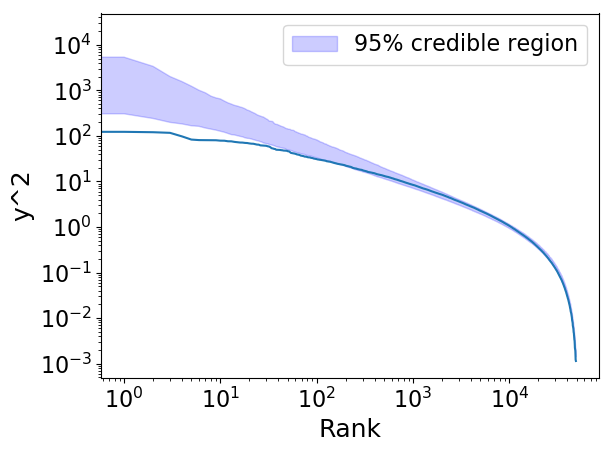}
%\caption{NS}
\end{subfigure}

\begin{subfigure}[t]{0.24\textwidth}
\centering
\includegraphics[width=\linewidth]{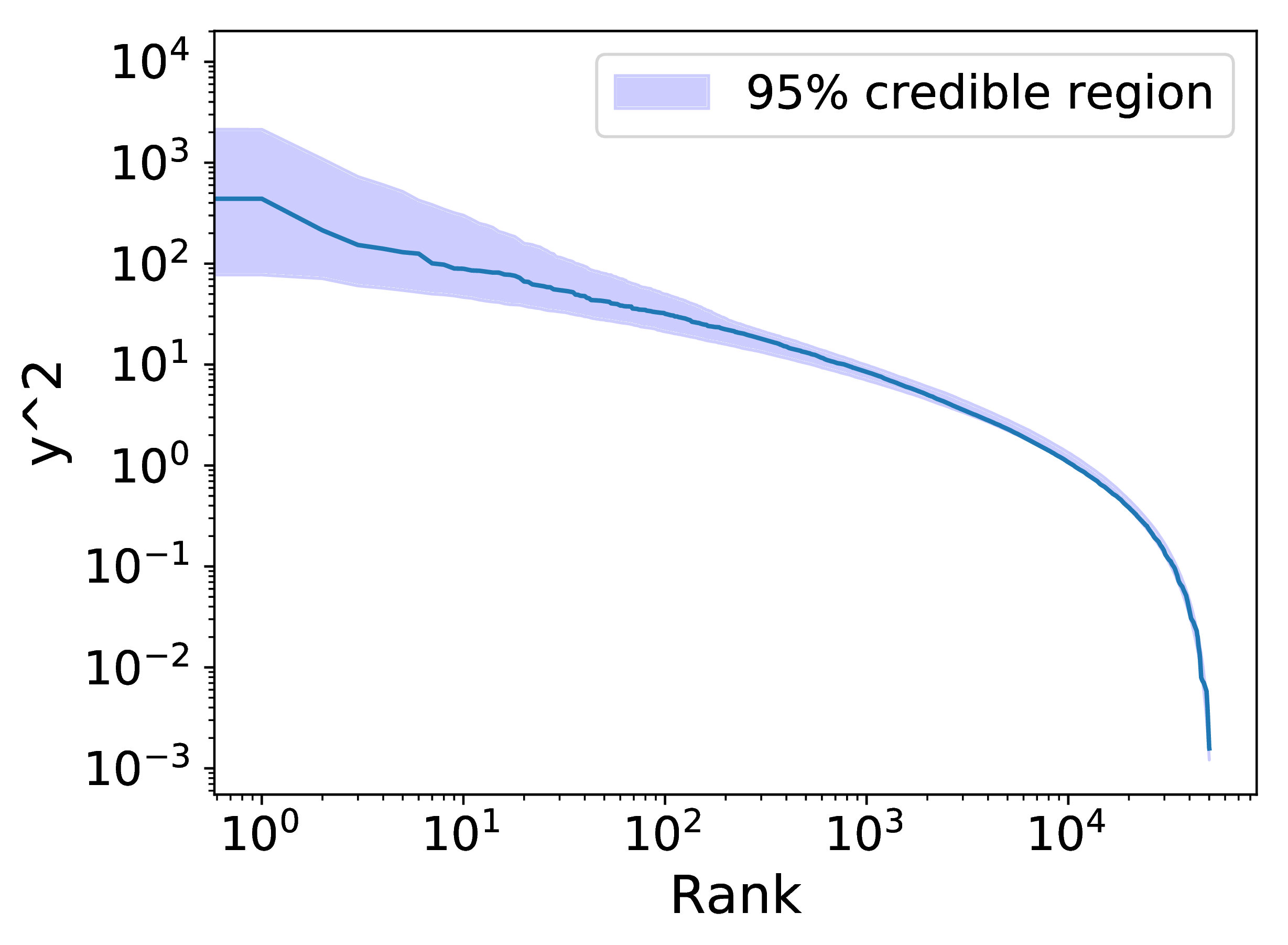}
\caption{NNGP}
\end{subfigure}
\begin{subfigure}[t]{0.24\textwidth}
\centering
\includegraphics[width=\linewidth]{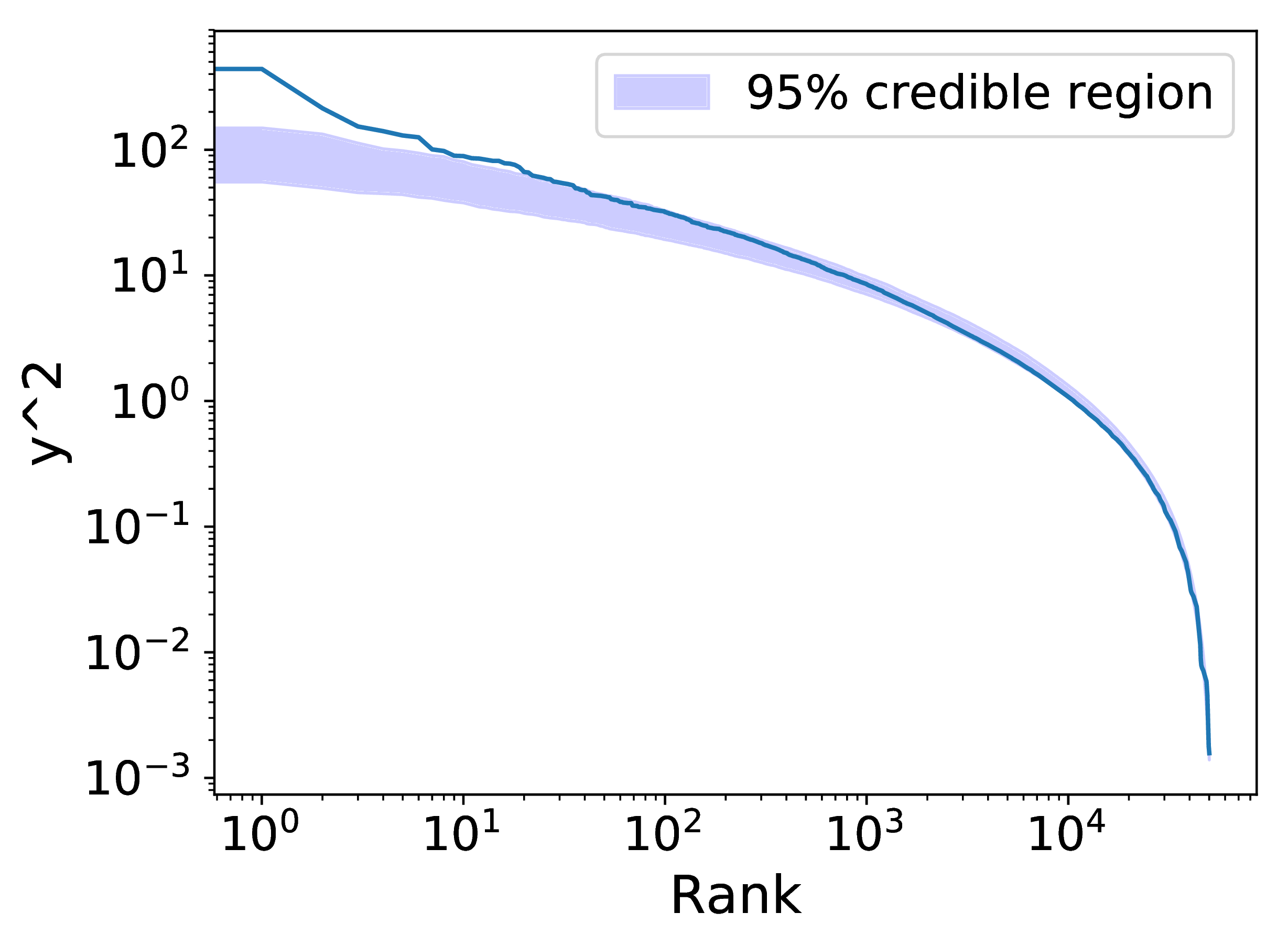}
\caption{GH}
\end{subfigure}
\begin{subfigure}[t]{0.24\textwidth}
\centering
\includegraphics[width=\linewidth]{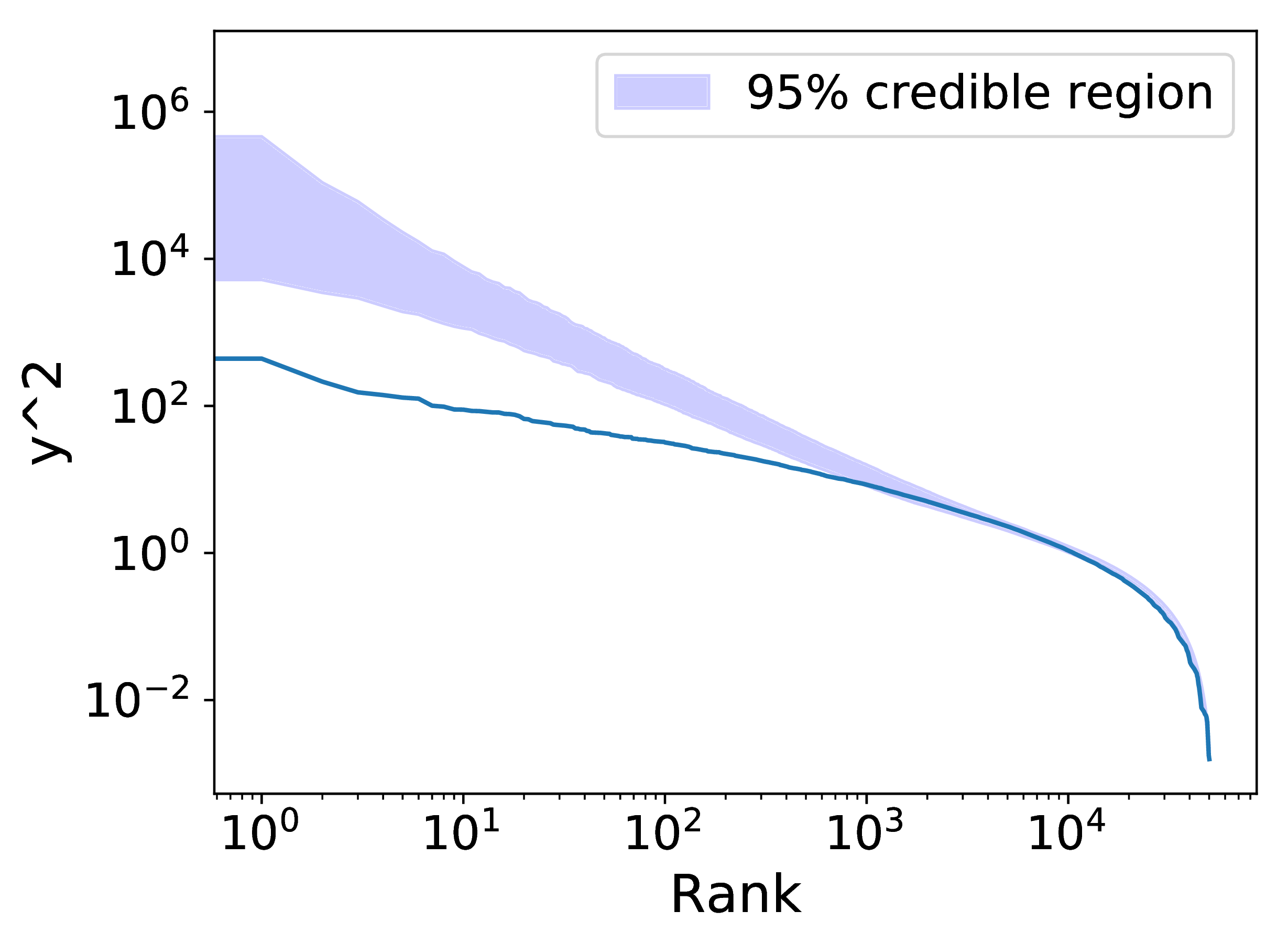}
\caption{NS}
\end{subfigure}
\begin{subfigure}[t]{0.24\textwidth}
\centering
\includegraphics[width=\linewidth]{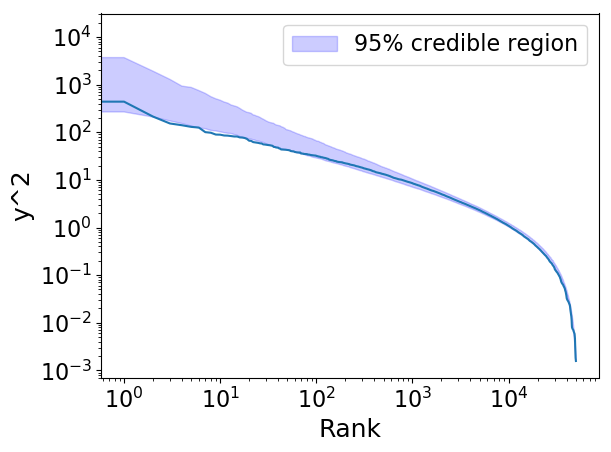}
\caption{Student}
\end{subfigure}
\caption{Ranked squared increments on the tech companies dataset. From top to bottom row: Apple, Amazon, Facebook. The line represents the ranked $y^2$ in the test dataset; the shaded area represent the 95\% credible region. Results are given for the  NGGP, GH,  NS and Student models in this order.}
\label{fig:simple_rank}
\end{figure}

\subsection{Ornstein-Uhlenbeck based model}

We now consider the Ornstein-Uhlenbeck based stochastic volatility model with NGGP marginal with $\sigma=0$.
As discussed in section~\ref{subsubsec:complex_nggp}, in this case, the simulation of the state noise $\varepsilon$ can be done exactly and exact posterior inference is possible.
We compare the model with NGGP to the model with normal-gamma marginal (NG), described in Section \ref{sec:modelfinance}, and demonstrate
that NGGP better captures the heavy-tails with minimal computational overhead compared to NG.

\begin{figure}
\centering
\includegraphics[width=0.24\linewidth]{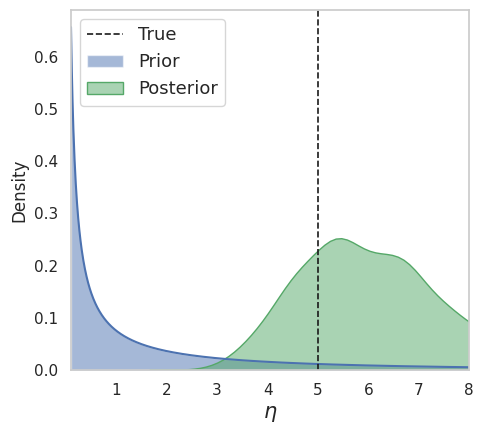}
\includegraphics[width=0.24\linewidth]{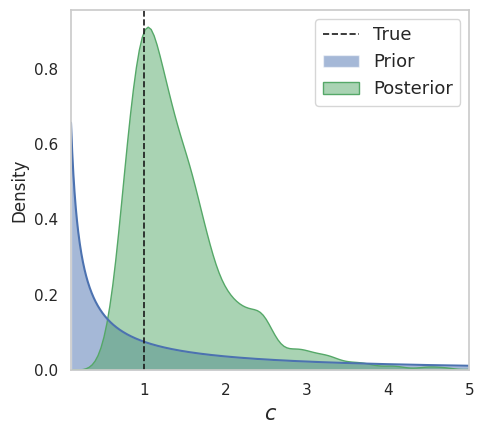}
\includegraphics[width=0.24\linewidth]{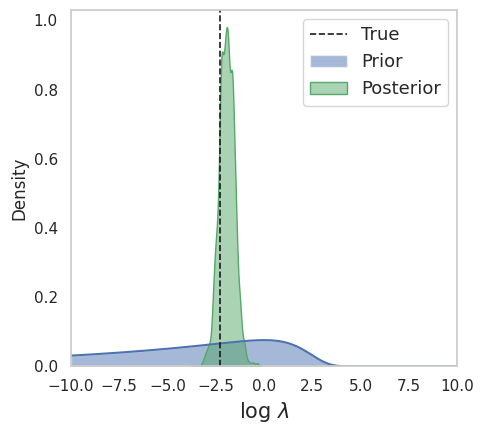}
\includegraphics[width=0.24\linewidth]{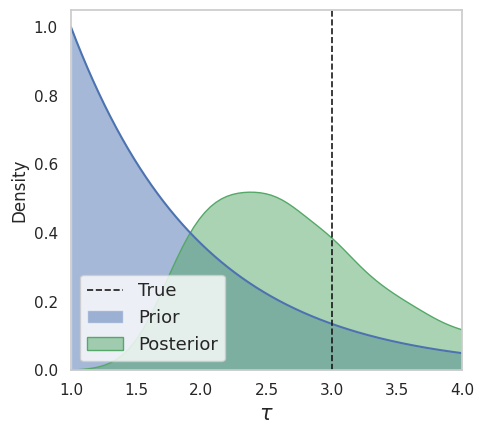}\\\includegraphics[width=0.24\linewidth]{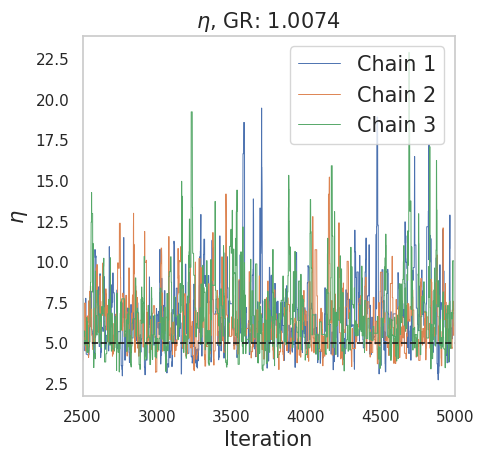}
\includegraphics[width=0.24\linewidth]{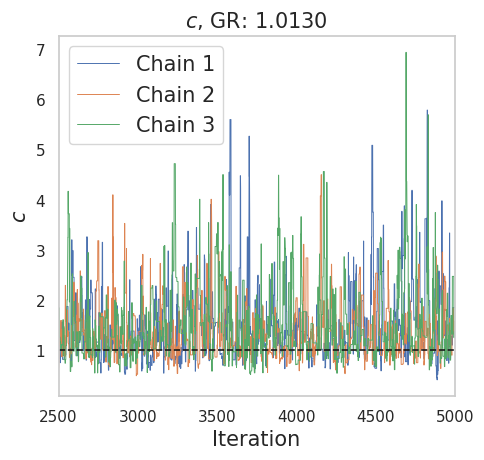}
\includegraphics[width=0.24\linewidth]{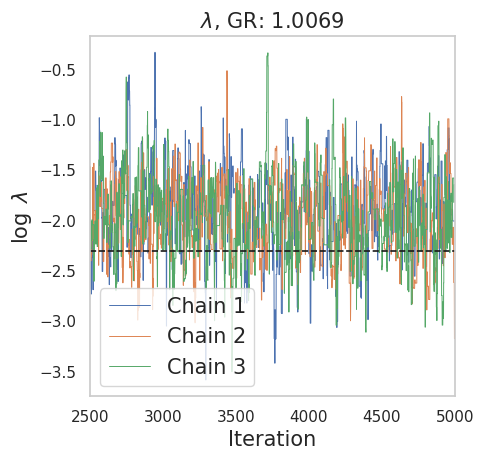}
\includegraphics[width=0.24\linewidth]{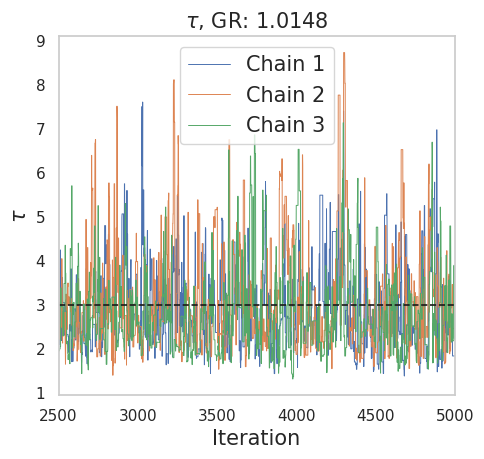}
\caption{Posterior samples of the parameters on simulated data from OU-basd stochastic volatility model with NGGP marginal.}
\label{fig:complex_simulated_3.0}
\vspace{-1em}
\end{figure}

\subsubsection{Simulated dataset}
\label{sec:OUsimulated}

We first demonstrate that our posterior sampler based on particle MCMC could successfully recover the true parameters on a simulated dataset.
We simulate data from the L\'evy-driven stochastic volatility model with NGGP marginal for $n=2\,000$ time-steps and parameters $\eta=5, c=1,\tau=3.0$ and $\lambda=0.1$.
We run three independent particle MCMC chains with $5\,000$ iterations ($2\,500$ burn-in) and $3\,000$ particles.
Figure~\ref{fig:complex_simulated_3.0} shows that our sampler successfully recovers the parameters. Trace plots suggest the convergence of the sampler. Figure~\ref{fig:complex_simulated_states} shows the posterior estimate and credible interval for the integrated volatility, together with the true value. Additional simulation results are reported in Appendix E, for data generated with a smaller value $\tau=1.5$. We also assess the sensitivity to the choice of prior, by reporting the posterior distributions under a different prior distribution for $\tau$.

\begin{figure}
\centering
\includegraphics[width=0.4\linewidth]{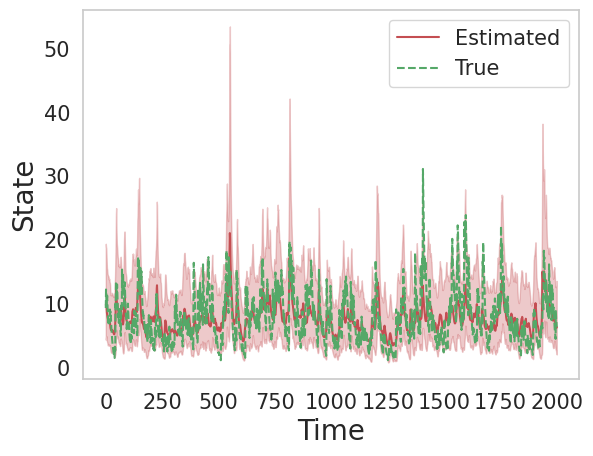}\\
\caption{Posterior mean (solid red line) and 95\% credible intervals (shaded area) of the integrated volatility. True volatility is in dashed green line.}
\label{fig:complex_simulated_states}
\vspace{-1em}
\end{figure}

\subsubsection{Real-world datasets}
\label{sec:ExperimentsOUreal}

\paragraph{Dataset.} The dataset is obtained from the Realized library\footnote{\url{https://realized.oxford-man.ox.ac.uk}}. We collected 14 daily stock data from 05-11-2007 to 07-10-2011 (around the time of subprime mortgage crisis), and fitted the L\'evy driven stochastic
volatility models on daily log return values. The data is accompanied with the estimates of the integrated variances $\overline v_k$ for each day, obtained with an estimator based on higher-frequency data; we use these values as ground-truth of the integrated variance, and note it $\overline v^{\text{true}}_k$.% to be able to compare the two different methods.

\paragraph{Evaluation metrics.}  For $k=1,\ldots,n$, let $\overline V_{k}^{(1)},\ldots,\overline V_{k}^{(n_s)}$ denote the posterior samples of the integrated variance over the $k$th interval, where $n_s$ is the number of MCMC iterations after burn-in. For $k=1,\ldots,n$ and any $\overline v_k\geq 0$, let $\widehat G_k(\overline v_k)=\frac{1}{n_s}\sum_{i=1}^{n_s} \1{\overline V_{k}^{(i)}\leq \overline v_k}$ be the Monte Carlo approximation of the posterior cumulative distribution function of the integrated variance $G_k(\overline v_k)=\Pr(\overline V_k\leq \overline v_k\mid y_1,\ldots,y_n)$. For $k=1,\ldots,n$, denote
$$
\zeta_k =1-\widehat G_k(\overline v^{\text{true}}_k)=\frac{1}{n_s}\sum_{i=1}^{n_s} \1{\overline V_{k}^{(i)}\geq \overline v_k^{\text{true}}}.
$$
In order to assess the goodness-of-fit of the model, we calculate the KS statistics between the empirical distribution of $(\zeta_1,\ldots,\zeta_n)$ and the distribution of a uniform random variable on $[0,1]$.
We also compare the fit of the model for different loss functions. Let $L$ be a loss function. If $L(x, y) = (x-y)^2$ is the $\ell_2$ loss, the Bayes estimator is the posterior mean. In case of $\ell_{1,\alpha}$ loss defined as
\[
L(x, y) =
\left\{ \begin{array}{ll}
x - y & \text{ if } x \geq y \\
\frac{1-\alpha}{\alpha}|x-y| & \text{ if } x < y
\end{array}
\right.,
\]
the Bayes estimator is the $\alpha$-quantile $G^{-1}_k(\alpha)$ of $G_k$. We assess the fit of each model by computing the average loss
$$
\frac{1}{n}\sum_{k=1}^n L(\overline v_k^{\text{true}}, \widehat {\overline v}_k^{L})
$$
where $\widehat {\overline v}_k^{L}$ is the Bayes estimator under the loss $L$.

\begin{figure}
\centering
%\begin{subfigure}{\linewidth}
%\centering
%\includegraphics[width=0.49\linewidth]{figures/complex/oxford/AEX_states.png}
%\vspace*{-2mm}
%\caption{AEX}
%\end{subfigure}
\begin{subfigure}{\linewidth}
\centering
\includegraphics[width=0.49\linewidth]{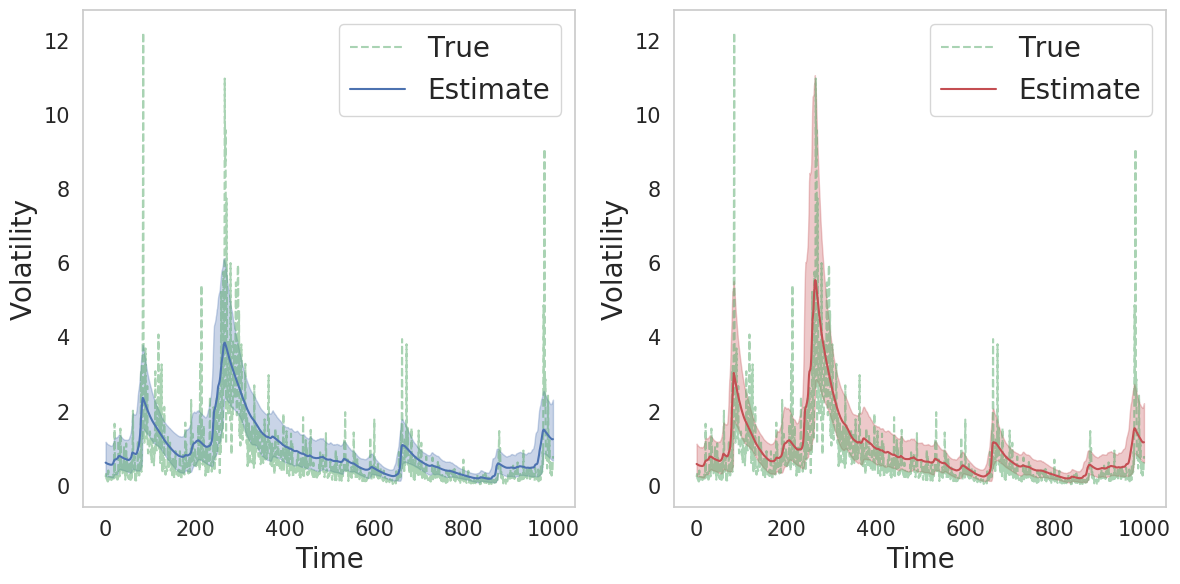}
\vspace*{-2mm}
%\caption{AORD}
\end{subfigure}
%\vspace*{3mm}
\caption{Posterior estimates (solid line) and 95\% credible intervals of the integrated volatility under the NG (left) and NGGP (right) models for the AORD index. The true integrated volatility is represented by a green dashed line.}
\label{fig:oxford_states}
\vspace{-2em}
\end{figure}

\paragraph{Results. } For every stock data, we run three independent chains of particle MCMC with $5\,000$ iterations ($2\,500$ burn-in) and $1\,000$ particles. The estimated parameters and credible intervals are given in Table~\ref{tab:oxford_parameters} (more datasets in Appendix \ref{sec:OUadditional}).  The comparisons between the fits of the two models in terms of KS statistics and empirical loss, for different loss functions, are reported in Table~\ref{tab:complex_oxford}. The model with NGGP marginal outperformed the one with NG marginal for all used metrics. Especially, since NGGP better captures the heavy-tails, the performance gap becomes more significant for the metrics emphasizing the heavy-tail regime ($\ell_{1,0.95}, \ell_{1,0.99}$).
This is well highlighted in Figure~\ref{fig:oxford_states}, which gives the estimated volatility and credible intervals under both models for the AORD stock indices. NGGP in general better captures ``spikes'' in the log-return values while NG often fails to get credible interval with good coverage.

\begin{table}[h]
\scriptsize
\centering
\setlength{\tabcolsep}{0.5pt}
\caption{Posterior mean and 95\% credible intervals of the parameters of the NG and NGGP marginals for the different indices.}
\begin{tabular}{@{}c ccc cccc}
\toprule
& \multicolumn{3}{c}{NG} & \multicolumn{4}{c}{NGGP} \\
\cmidrule(lr){2-4}\cmidrule(lr){5-8}
& $\eta$ & $\lambda$ & $c$ & $\eta$ & $\lambda$ & $c$ & $\tau$ \\
\midrule
AEX & \qnts{1.43}{0.71}{2.50} & \qnts{0.02}{0.01}{0.04} & \qnts{1.41}{0.66}{2.56} & \qnts{2.58}{1.17}{4.47} & \qnts{0.03}{0.01}{0.05} & \qnts{9.34}{3.00}{23.04} & \qnts{1.49}{1.05}{2.37} \\
\addlinespace[0.5em]
AORD & \qnts{1.97}{0.92}{3.42} & \qnts{0.02}{0.01}{0.03} & \qnts{2.06}{0.88}{3.60} &  \qnts{3.62}{1.52}{6.87} & \qnts{0.03}{0.01}{0.04} &  \qnts{11.26}{2.89}{29.19} & \qnts{1.65}{1.09}{3.05} \\
\addlinespace[0.5em]
DJI & \qnts{1.28}{0.68}{2.05} & \qnts{0.02}{0.01}{0.03} & \qnts{1.26}{0.65}{2.10} & \qnts{2.19}{1.10}{3.69} & \qnts{0.03}{0.01}{0.04} & \qnts{8.53}{2.30}{19.95} & \qnts{1.38}{1.02}{2.25} \\
\addlinespace[0.5em]
FTSE  & \qnts{1.34}{0.68}{2.22} & \qnts{0.02}{0.01}{0.02} & \qnts{1.46}{0.73}{2.37} & \qnts{3.19}{1.35}{6.67} & \qnts{0.02}{0.01}{0.04} & \qnts{14.36}{3.55}{38.67} & \qnts{1.36}{1.04}{2.08} \\
\addlinespace[0.5em]
GSPTSE & \qnts{1.23}{0.59}{2.13} & \qnts{0.01}{0.01}{0.02} & \qnts{1.28}{0.57}{2.31} & \qnts{2.31}{0.92}{4.26} & \qnts{0.02}{0.01}{0.03} & \qnts{10.33}{2.42}{25.48} & \qnts{1.42}{1.02}{2.29} \\
\bottomrule
\end{tabular}
\label{tab:oxford_parameters}
\end{table}

\begin{table}[h]
\centering
\scriptsize
\setlength{\tabcolsep}{3pt}
\caption{Comparison of the fit of the NG and NGGP models under different metrics. }
\label{tab:complex_oxford}
\begin{tabular}{@{} c cc cc cc cc cc}
\toprule
& \multicolumn{2}{c}{$\mathrm{KS}((\zeta_k), U(0,1))$}
& \multicolumn{2}{c}{$\ell_2$}
& \multicolumn{2}{c}{$\ell_{1,0.5}$}
& \multicolumn{2}{c}{$\ell_{1,0.95}$}
& \multicolumn{2}{c}{$\ell_{1,0.99}$}\\
Data & {\scriptsize NG} & {\scriptsize NGGP}
& {\scriptsize NG} & {\scriptsize NGGP}
& {\scriptsize NG} & {\scriptsize NGGP}
& {\scriptsize NG} & {\scriptsize NGGP}
& {\scriptsize NG} & {\scriptsize NGGP}\\
\midrule
AEX & 0.237 & 0.200 & 0.920 & 0.950 & 0.398 & 0.396 & 0.127 & 0.113 & 0.074 & 0.053 \\
AORD & 0.531 & 0.511 & 0.688 & 0.680 & 0.465 & 0.453 & 0.117 & 0.101 & 0.065 & 0.047 \\
DJI & 0.371 & 0.341 & 1.859 & 1.677 & 0.476 & 0.456 & 0.162 & 0.138 & 0.107 & 0.077\\
FTSE & 0.269 & 0.241 & 2.590 & 2.510 & 0.479 & 0.480 & 0.186 & 0.159 & 0.134 & 0.092\\
GSPTSE & 0.450 & 0.432 & 13.993 & 13.656 & 0.615 & 0.612 & 0.277 & 0.259 & 0.226 & 0.198 \\
HSI & 0.351 & 0.335 & 1.081 & 1.056 & 0.426 & 0.411 & 0.133 & 0.126 & 0.086 & 0.072 \\
IBEX & 0.264 & 0.245 & 0.824 & 0.788 & 0.422 & 0.413 & 0.122 & 0.106 & 0.067 & 0.045 \\
IXIC & 0.433 & 0.421 & 0.849 & 0.898 & 0.412 & 0.418 & 0.102 & 0.092 & 0.055 & 0.042\\
KS11 & 0.237 & 0.177 & 1.740 & 1.207 & 0.405 & 0.358 & 0.151 & 0.095 & 0.096 & 0.039 \\
MXX & 0.580 & 0.553 & 1.030 & 1.142 & 0.523 & 0.518 & 0.095 & 0.087 & 0.047 & 0.036\\
N225 & 0.283 & 0.230 & 0.674 & 0.807 & 0.360 & 0.362 & 0.087 & 0.073 & 0.045 & 0.030 \\
RUT & 0.570 & 0.392 & 1.217 & 1.192 & 0.449 & 0.454 & 0.079 & 0.072 & 0.031 & 0.025 \\
SPX & 0.388 & 0.337 & 1.318 & 1.317 & 0.440 & 0.435 & 0.131 & 0.115 & 0.082 & 0.058 \\
SSMI & 0.276 & 0.259 & 1.420 & 1.292 & 0.438 & 0.437 & 0.160 & 0.135 & 0.098 & 0.062 \\
\midrule
Mean & 0.374 & \bf 0.344 & 2.157 & \bf 2.084 & 0.451 & \bf 0.443 & 0.138 & \bf 0.119 & 0.087 & \bf 0.063 \\
\bottomrule
\end{tabular}
\vspace{-1em}
\end{table}

\subsubsection{Comparison to ARMA-GARCH}
\label{sec:armagarch}

We compare the OU-based stochastic volatility model to ARMA$(1,1)$-GARCH$(1,1)$ on data from the Oxford Realized library. For this, we used 2,000 log-return values from the Oxford dataset, and split them into 1,100 training time steps and 900 test time steps. We fit the OU-based model with NG and NGGP marginals with our sampler. We report in Table~\ref{tab:ll} the marginal log-likelihood on the training data for the three models (more datasets are left to Appendix \ref{sec:OUadditionalsimus}). NGGP outperforms NG but performs slightly worse than ARMA-GARCH. We further compare the models in terms of prediction. We generate one-step predictions for the 900 test time-steps, compute Value at Risk (VaR) values for each time step prediction, and counted the fraction of actual test data less than or equal to the negative of VaR values. This is to see whether estimated VaR values fit the test data well by checking
\[
\Pr(-Y \leq \mathrm{VaR}_\alpha) = 1-\alpha \iff \Pr(Y > - \mathrm{VaR}_\alpha) = 1-\alpha \iff \Pr(Y \leq -\mathrm{VaR}_\alpha) = \alpha,
\]
so that the fraction of test data less than or equal to $-\mathrm{VaR}_\alpha$ being closer to $\alpha$ means better prediction. For all models, we collected posterior samples, conducted prediction for each posterior samples using corresponding model parameters and state estimates, computed empirical CDFs using those samples, and computed VaR values. As summarised in Table~\ref{tab:var} (more datasets are left to Appendix \ref{sec:OUadditionalsimus}), the results of the different methods are comparable.

\begin{table}
\centering
\caption{Comparison of the marginal log-likelihood values of OU-NG, OU-NGGP and ARMA-GARCH models on data from the Realized-library.}
\scriptsize
\begin{tabular}{@{}cccc@{}}
\toprule
Data & NG & NGGP & ARMA-GARCH \\
\midrule
  AEX   & -1466.021 & -1465.811  & -1459.768 \\
  AORD  & -1495.441  & -1494.830 & -1492.148 \\
  DJI   & -1442.606 & -1438.013 & -1424.026 \\
  FTSE  & -1448.833 & -1445.721 & -1437.881 \\
 GSPTSE & -1455.563 & -1454.125 & -1445.307 \\
\bottomrule
\end{tabular}
\label{tab:ll}
\end{table}

\begin{table}
\centering
\caption{The results of VaR test on predicted sequence. Values closer to $\alpha$ mean better VaR prediction.
}
\scriptsize
\setlength{\tabcolsep}{4pt}
\begin{tabular}{@{}ccccccc@{}}
\toprule
 & NG & NGGP & \pbox{4cm}{ARMA\\-GARCH} & NG & NGGP & \pbox{3cm}{ARMA\\-GARCH} \\
\cmidrule(lr){2-4}\cmidrule(lr){5-7}
Data & \multicolumn{3}{c}{$\alpha=0.95$} & \multicolumn{3}{c}{$\alpha=0.99$} \\
\midrule
  AEX   & 0.962 & 0.958 &   0.969   & 0.993 & 0.993 & 0.992 \\
  AORD  & 0.962 & 0.960  &    0.970  & 0.996 & 0.993 & 0.991 \\
  DJI   & 0.956 & 0.959 &   0.968  & 0.996 & 0.994 & 0.989 \\
  FTSE  & 0.956 & 0.953 &    0.950  & 0.994 & 0.993 & 0.983 \\
 GSPTSE & 0.967 & 0.968 &   0.978  & 0.997 & 1.000 & 0.996\\
\bottomrule
\end{tabular}
\label{tab:var}
\end{table}

 \begin{acks}[Acknowledgments]
 The authors thank Cian Naik, Lancelot James and Matthias Winkel for useful feedback on an earlier version of this article.
 \end{acks}

%%%%%%%%%%%%%%%%%%%%%%%%%%%%%%%%%%%%%%%%%%%%%%
%% Supplementary Material, if any, should   %%
%% be provided in {supplement} environment  %%
%% with title and short description.        %%
%%%%%%%%%%%%%%%%%%%%%%%%%%%%%%%%%%%%%%%%%%%%%%
%\begin{supplement}

%\stitle{???}
%\sdescription{???.}
%\end{supplement}

%% if your bibliography is in bibtex format, uncomment commands:
\bibliographystyle{imsart-nameyear} % Style BST file (imsart-number.bst or imsart-nameyear.bst)
\bibliography{levy}       % Bibliography file (usually '*.bib')

%% or include bibliography directly:
%\begin{thebibliography}{}
% \bibitem{b1}
% \end{thebibliography}
\newpage

\appendix

\section{Useful identities}

\revision{We state here some useful identities on the incomplete gamma function $\gamma(s,x)$, which appears in the definition of the L\'evy intensity of the GGP process in Equation \eqref{eq:GBFRYintensity}.} We have, for $s,x>0$
\begin{align}
%\gamma(1,x)&=1-e^{-x}\\
\gamma(s,x)&=x^s\int_0^1 v^{s-1}e^{-vx}dv\label{eq:identityincgamma1}\\
s\gamma(s,x)&=\gamma(s+1,x)+x^s e^{-x}\label{eq:identityincgamma3}\\
\gamma(s,x)&\obsim{x\to 0} \frac{x^s}{s}\label{eq:incgammaasymp0}
\end{align}
For any $\tau,\kappa,c>0$,
$
\int_0^\infty x^{-1-\tau}\gamma(\kappa+\tau,cx)dx=\frac{c^{\tau}\Gamma(\kappa)}{\tau}
$.

\section{Generalised BFRY distribution}
\label{sec:GBFRY}

\revision{In the finite-activity case, the GGP process is a compound Poisson process whose compound distribution is the generalised BFRY distribution.} The generalised BFRY was introduced by \cite{Ayed2019} as a generalisation to the BFRY distribution~\citep{Pitman1997,Winkel2005,Bertoin2006}. The term BFRY was coined after Bertoin, Fujita, Roynette and Yor by~\cite{Devroye2014}. \revision{We describe in this section some properties of this distribution.}

A positive random variable $X$ with generalised BFRY distribution has probability density function
\begin{align}
\GBFRY(x\ ;\kappa, \tau, c)=\frac{\tau}{c^\tau\Gamma(\kappa)}x^{-1-\tau}\gamma(\kappa+\tau,cx)\label{eq:GBFRYpdf}
\end{align}
for some parameters  $\kappa,\tau,c>0$. $c$ is an inverse scale parameter; $\kappa$ controls the behavior at 0 as, using Karamata's theorem~\cite[Proposition 1.5.10]{Bingham1989} and Equation~\eqref{eq:incgammaasymp0}
$$
\Pr(X<x)\obsim{x\to 0}\frac{\tau c^\kappa}{\kappa(\kappa+\tau)\Gamma(\kappa)}x^{\kappa}
$$
 and $\tau$ is a power-law exponent controlling the tails of the distribution as, using Karamata's theorem
$$
\Pr(X>x)\obsim{x\to\infty} \frac{\Gamma(\kappa+\tau)}{c^\tau\Gamma(\tau)}x^{-\tau}.
$$
The moments are given by
\begin{align}
\mathbb E[X^m]=\frac{\tau\Gamma(m+\kappa)}{c^m(\tau-m)\Gamma(\kappa)}
\end{align}
for $m<\tau$, and $\mathbb E[X^m]=\infty$ otherwise. \revision{Using the identity \eqref{eq:identityincgamma1}, the pdf \eqref{eq:GBFRYpdf} admits the alternative representation
\begin{align*}
\GBFRY(x\ ;\kappa, \tau, c)=\int_0^1 v\times \frac{(vx)^{\kappa-1}e^{-cvx}c^\kappa}{\Gamma(\kappa)}\times\tau v^{\tau-1}dv.
\end{align*}}
 A GBFRY random variable therefore admits the following representation
$$
X\overset{d}{=}Y/Z
$$
where $Y\sim \Gammadist(\kappa, c)$ random variable and $Z\sim \Betadist(\tau, 1)$ where $\Gammadist(\kappa, c)$ denotes the gamma distribution with shape parameter $\kappa>0$ and inverse scale parameter $c>0$, and $\Betadist(a, b)$ denotes the beta distribution with parameters $a,b>0$. Note that $Z^{-1}\sim \Pareto(\tau,1)$.

\section{Tauberian-Abelian theorem for subordinated Brownian}
\label{sec:app:taubersubbrownian}
\revision{The following proposition states that the regular variation properties of the tail L\'evy intensity of a subordination Brownian process can be deduced from those of the tail L\'evy intensity of the corresponding subordinator. This proposition is used to derive the properties of the normal GGP process in Section~\ref{sec:NGGP} from those of the GGP subordinator. The proof relies on standard properties of regularly varying functions.}

\begin{proposition}\label{prop:tauberian}

Let $\varrho$ be a L\'evy measure on $(0,\infty)$ and let $\overline\rho(x)=\int_x^\infty \varrho(dw)$ be its tail L\'evy intensity. For $x>0$, let
\begin{align}
\overline\nu(x)=\int_{|s|>x}\nu(s)ds
\end{align}
where, for $s\in\mathbb R$,
$$
\nu(s)=\int_0^\infty (2\pi w)^{-1/2}e^{-\frac{s^2}{2w}}\varrho(dw).
$$
If $\overline\rho(x)$ is a regularly-varying function at 0 with
\begin{align}
\overline\rho(x)&\obsim{x\to 0} x^{-\alpha}\ell_1(1/x)\label{eq:RV0rho}
\end{align}
where $\alpha\in[0,1]$ and $\ell_1$ is a slowly varying function, then
$\overline\nu(x)$ is also regularly varying at 0, with
\begin{align}
\overline\nu(x)&\obsim{x\to 0} \frac{2^{\alpha+1}\Gamma(\alpha+1/2)}{\sqrt{\pi}} x^{-2\alpha} \ell_1(1/x^2).\label{eq:RV0nu}
\end{align}
If $\overline\rho(x)$ is a regularly-varying function at infinity with
\begin{align}
\overline\rho(x)&\obsim{x\to \infty} x^{-\tau}\ell_2(x)\label{eq:RVinfrho}
\end{align}
where $\tau\geq 0$ and $\ell_2$ is a slowly varying function, then
$\overline\nu(x)$ is also regularly varying at infinity, with
\begin{align}
\overline\nu(x) &\obsim{x\to \infty}  \frac{2^{\tau+1}\Gamma(\tau+1/2)}{\sqrt{\pi}} x^{-2\tau} \ell_2(x^2)\label{eq:RVinfnu}
\end{align}

If additionally $x\to x^{1/2}\overline\rho(2x)$ is ultimately monotone, then \eqref{eq:RV0nu} also implies \eqref{eq:RV0rho}; if $x\to x^{-1/2}\overline\rho(1/(2x))$ is ultimately monotone, then \eqref{eq:RVinfnu} also implies \eqref{eq:RVinfrho}.

\end{proposition}

\begin{proof}
We have
\begin{align*}
\overline\nu(x)&=\sqrt{\frac{2}{\pi}}\int_{\mu_0+x}^\infty\int_0^\infty w^{-1/2}e^{-\frac{(s-\mu_0)^2}{2w}}\varrho(dw)ds.
\end{align*}
With the change of variable $u=\frac{(s-\mu_0)^2}{2wx^2}$, we obtain
\begin{align*}
\overline\nu(x)&=\frac{2x}{\sqrt{\pi}}\int_0^\infty\int_{1/(2u)}^\infty u^{-1/2}e^{-ux^2 }\varrho(dw)du=\frac{2x}{\sqrt{\pi}}\int_0^\infty u^{-1/2}e^{-ux^2 }\overline\rho(1/(2u))du
\end{align*}
We have
\begin{align*}
u^{-1/2}\overline\rho(1/(2u))&\obsim{u\to 0} 2^\tau u^{\tau-1/2} \ell_2(1/u)\\
u^{-1/2}\overline\rho(1/(2u))&\obsim{u\to \infty} 2^\alpha u^{\alpha-1/2} \ell_1(u).
\end{align*}
It follows from Karamata Abelian theorem~\cite[Chapter XIII, Section 5]{Feller1971} that
\begin{align*}
\overline\nu(x) &\obsim{x\to \infty}  \frac{2^{\tau+1}\Gamma(\tau+1/2)}{\sqrt{\pi}} x^{-2\tau} \ell_2(x^2)\\
\overline\nu(x)&\obsim{x\to 0} \frac{2^{\alpha+1}\Gamma(\alpha+1/2)}{\sqrt{\pi}} x^{-2\alpha} \ell_1(1/x^2).
\end{align*}
The reverse Tauberian result holds under monotonicity conditions near 0 or infinity of the function $u\to u^{-1/2}\overline\rho(1/(2u))$, see \cite[Theorem 1.7.2 page 39]{Bingham1989}.
\end{proof}

\section{Additional Details and Simulation results for the Exponentiated L\'evy model}
\label{sec:sensitivity}

\par \revision{\textbf{Computational time:} For $n=1500$ observations, using $1500$ particles and running the chain for $5000$ iterations requires approximately 8 hours on a single CPU.}

\textbf{Sensitivity to prior:} To assess the sensitivity of the posterior to the choice of prior we conduct experiments with two different choices of prior on $\tau$: $(\tau-1)\sim \Gammadist(1,1)$ and $\tau\sim \Unif(1, 5)$. In Figure  \ref{fig:gbfry_simulated}, we consider a setting with a relatively large number $n=5000$ of samples  generated from a NGGP model. We can see that the choice of the prior has a relatively limited impact on the posterior. In Figure \ref{fig:ghd_simulated}, we consider a more challenging setting with smaller number $n=1500$ of samples generated from a GHD model (misspecified case). Here we can see that the choice of the prior has a non negligible effect, and more specifically the tail of the posterior of $\tau$ under the uniform prior is heavier. However, we can see that the posterior of other parameters remain unchanged, and that the one on $\tau$ still concentrates around the same mode.

\textbf{Goodness of fit to the tails:} In Figure \ref{fig:app_simple_rank} we report the ranked squared increments on the tech companies dataset.

\begin{figure}
     \centering
     \begin{subfigure}[b]{0.49\textwidth}
         \centering
         \includegraphics[width=\textwidth]{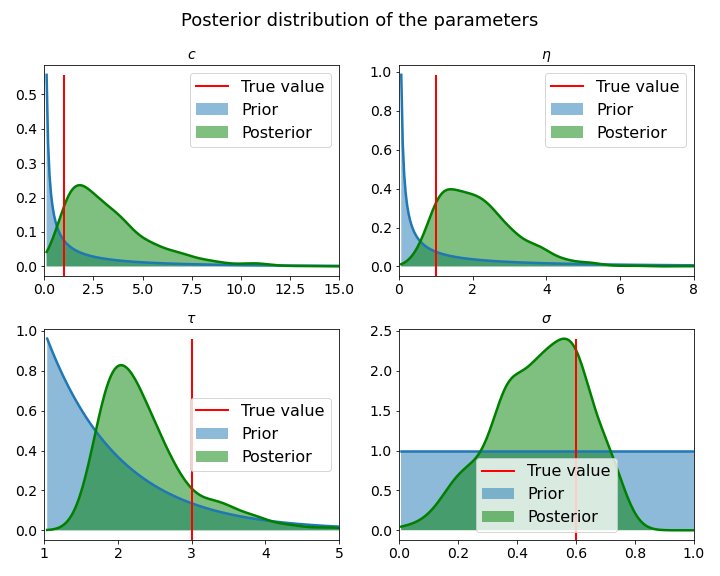}
         \caption{Posterior with Gamma prior for $\tau$.}
         \label{fig:gbfry_gamma}
     \end{subfigure}
     \hfill
     \begin{subfigure}[b]{0.49\textwidth}
         \centering
         \includegraphics[width=\textwidth]{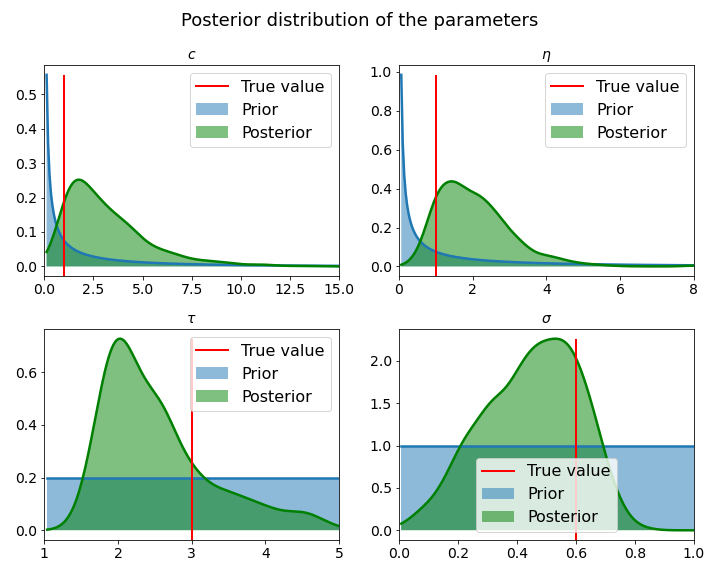}
         \caption{Posterior with Uniform prior for $\tau$}
         \label{fig:gbfry_uniform}
     \end{subfigure}
        \caption{Posterior distribution of the parameters of the NGGP parameters with different priors for the parameter $\tau$. Data simulated from a NGGP model (vertical red lines) with $5000$ training samples.}
        \label{fig:gbfry_simulated}
\end{figure}

\begin{figure}
     \centering
     \begin{subfigure}[b]{0.49\textwidth}
         \centering
         \includegraphics[width=\textwidth]{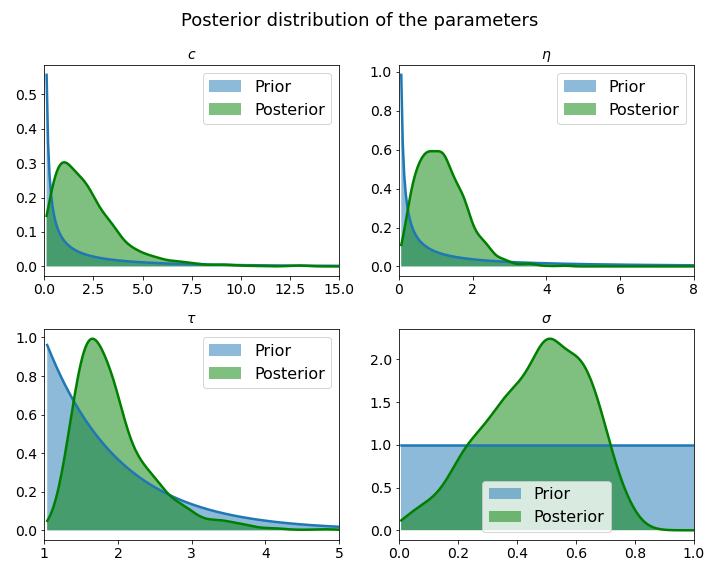}
         \caption{Posterior with Gamma prior for $\tau$.}
         \label{fig:ghd_gamma}
     \end{subfigure}
     \hfill
     \begin{subfigure}[b]{0.49\textwidth}
         \centering
         \includegraphics[width=\textwidth]{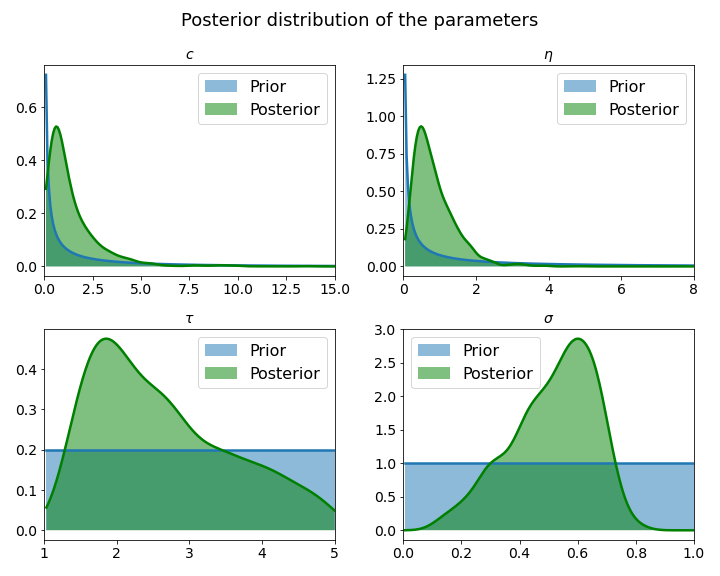}
         \caption{Posterior with Uniform prior for $\tau$.}
         \label{fig:ghd_uniform}
     \end{subfigure}
        \caption{Posterior distribution of the parameters of the NGGP parameters with different priors for the parameter $\tau$. Data simulated from a GHD model with $1 500$ training samples.}
        \label{fig:ghd_simulated}
\end{figure}

\begin{figure}
\centering
\begin{subfigure}[t]{0.24\textwidth}
\centering
\includegraphics[width=\linewidth]{figures/simple/AAPL_min_train/gbfry/Posterior_predictive_ordered_y_square_test.png}
%\caption[]{NGGP}
\end{subfigure}
\begin{subfigure}[t]{0.24\textwidth}
\centering
\includegraphics[width=\linewidth]{figures/simple/AAPL_min_train/ghd/Posterior_predictive_ordered_y_square_test.png}
%\caption{GHD}
\end{subfigure}
\begin{subfigure}[t]{0.24\textwidth}
\centering
\includegraphics[width=\linewidth]{figures/simple/AAPL_min_train/ns/Posterior_predictive_ordered_y_square_test.png}
%\caption{NS}
\end{subfigure}
\begin{subfigure}[t]{0.24\textwidth}
\centering
\includegraphics[width=\linewidth]{figures/simple/AAPL_min_train/student/Posterior_predictive_ordered_y_square_test.png}
%\caption{NS}
\end{subfigure}

\begin{subfigure}[t]{0.24\textwidth}
\centering
\includegraphics[width=\linewidth]{figures/simple/AMZN_min_train/gbfry/Posterior_predictive_ordered_y_square_test.png}
%\caption{GBFRY}
\end{subfigure}
\begin{subfigure}[t]{0.24\textwidth}
\centering
\includegraphics[width=\linewidth]{figures/simple/AMZN_min_train/ghd/Posterior_predictive_ordered_y_square_test.png}
%\caption{GHD}
\end{subfigure}
\begin{subfigure}[t]{0.24\textwidth}
\centering
\includegraphics[width=\linewidth]{figures/simple/AMZN_min_train/ns/Posterior_predictive_ordered_y_square_test.png}
%\caption{NS}
\end{subfigure}
\begin{subfigure}[t]{0.24\textwidth}
\centering
\includegraphics[width=\linewidth]{figures/simple/AMZN_min_train/student/Posterior_predictive_ordered_y_square_test.png}
%\caption{NS}
\end{subfigure}

\begin{subfigure}[t]{0.24\textwidth}
\centering
\includegraphics[width=\linewidth]{figures/simple/FB_min_train/gbfry/Posterior_predictive_ordered_y_square_test.png}
%\caption{GBFRY}
\end{subfigure}
\begin{subfigure}[t]{0.24\textwidth}
\centering
\includegraphics[width=\linewidth]{figures/simple/FB_min_train/ghd/Posterior_predictive_ordered_y_square_test.png}
%\caption{GHD}
\end{subfigure}
\begin{subfigure}[t]{0.24\textwidth}
\centering
\includegraphics[width=\linewidth]{figures/simple/FB_min_train/ns/Posterior_predictive_ordered_y_square_test.png}
%\caption{NS}
\end{subfigure}
\begin{subfigure}[t]{0.24\textwidth}
\centering
\includegraphics[width=\linewidth]{figures/simple/FB_min_train/student/Posterior_predictive_ordered_y_square_test.png}
%\caption{NS}
\end{subfigure}

\begin{subfigure}[t]{0.24\textwidth}
\centering
\includegraphics[width=\linewidth]{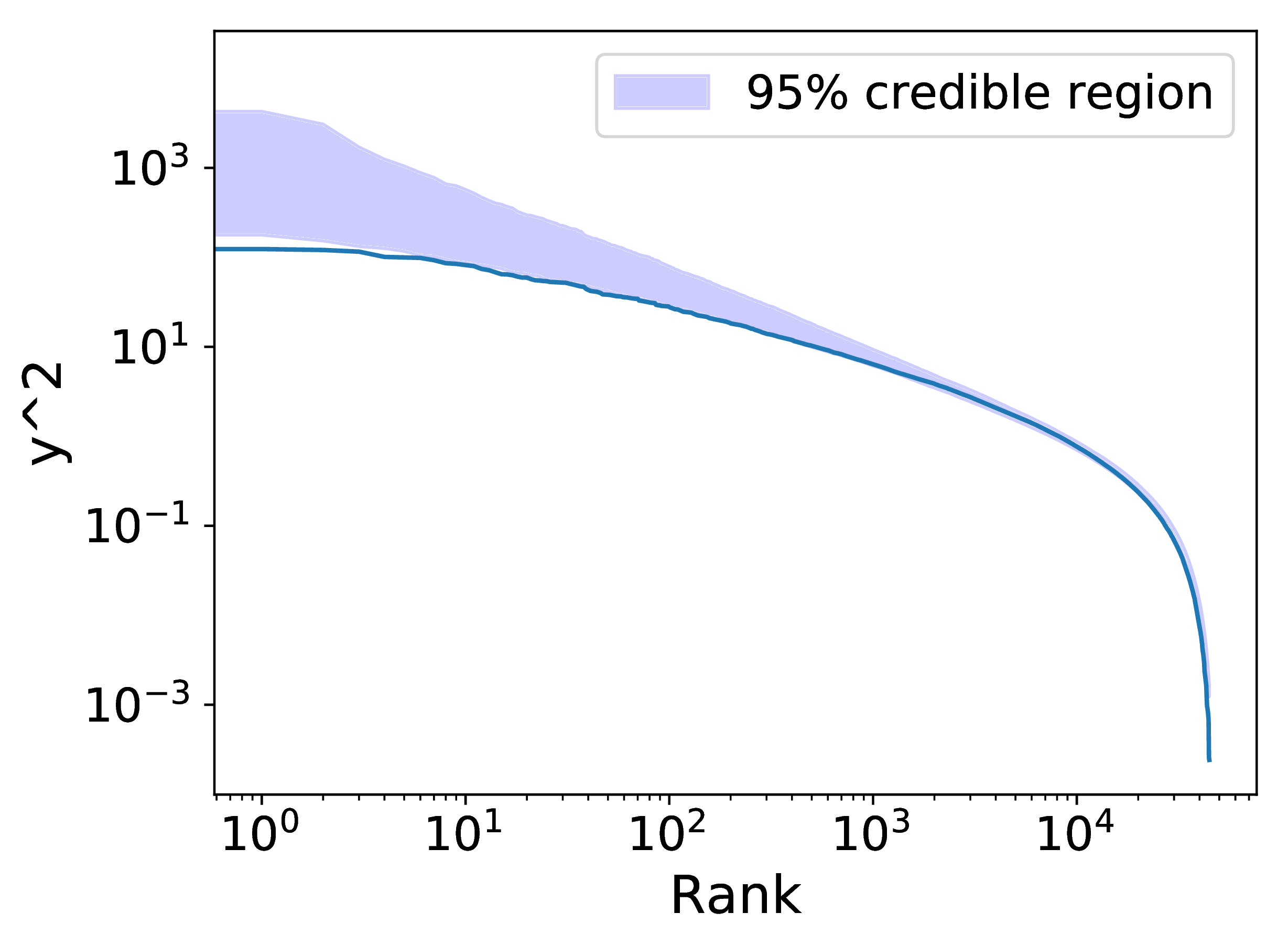}
%\caption{GBFRY}
\end{subfigure}
\begin{subfigure}[t]{0.24\textwidth}
\centering
\includegraphics[width=\linewidth]{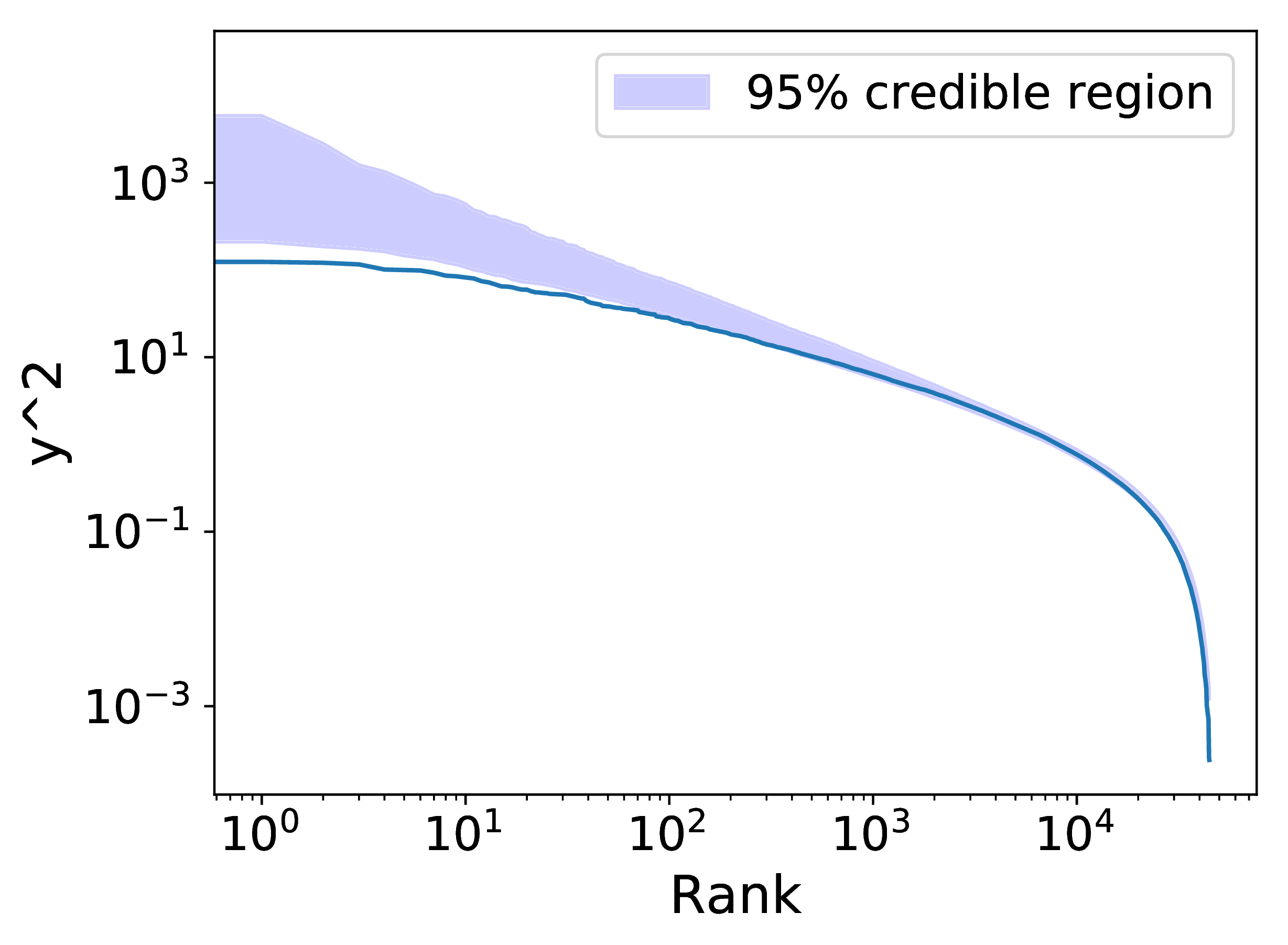}
%\caption{GHD}
\end{subfigure}
\begin{subfigure}[t]{0.24\textwidth}
\centering
\includegraphics[width=\linewidth]{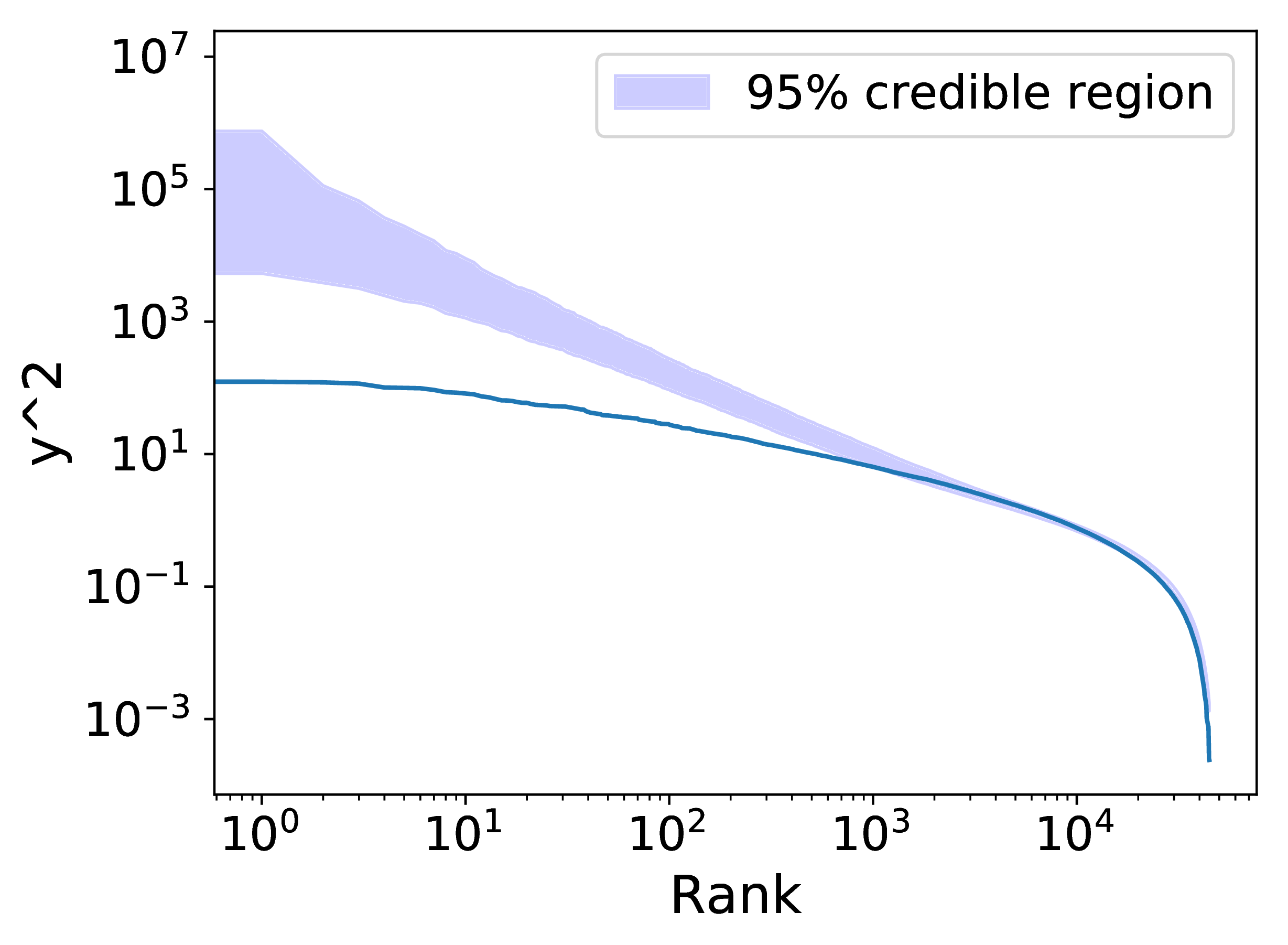}
%\caption{NS}
\end{subfigure}
\begin{subfigure}[t]{0.24\textwidth}
\centering
\includegraphics[width=\linewidth]{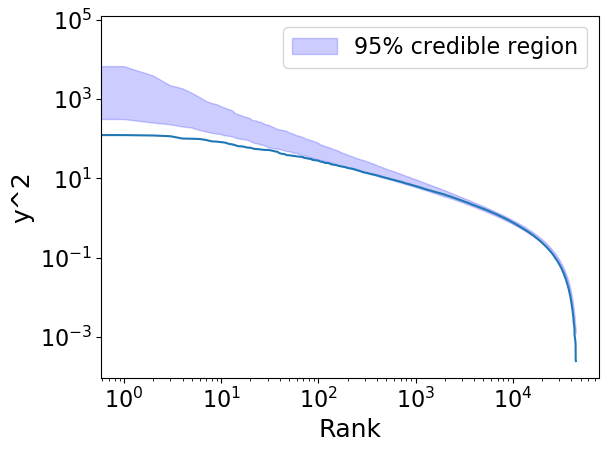}
%\caption{NS}
\end{subfigure}

\begin{subfigure}[t]{0.24\textwidth}
\centering
\includegraphics[width=\linewidth]{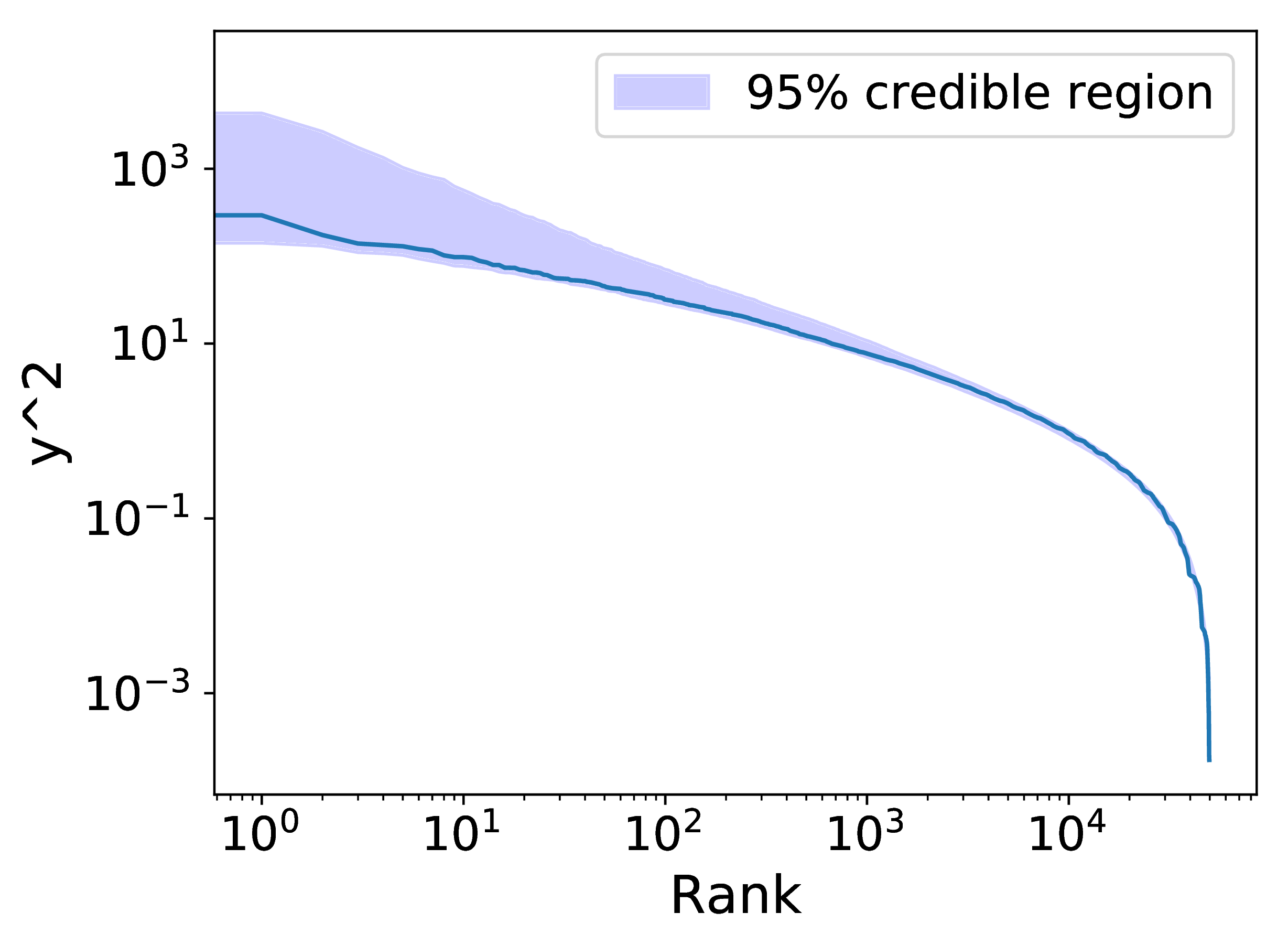}
%\caption{GBFRY}
\end{subfigure}
\begin{subfigure}[t]{0.24\textwidth}
\centering
\includegraphics[width=\linewidth]{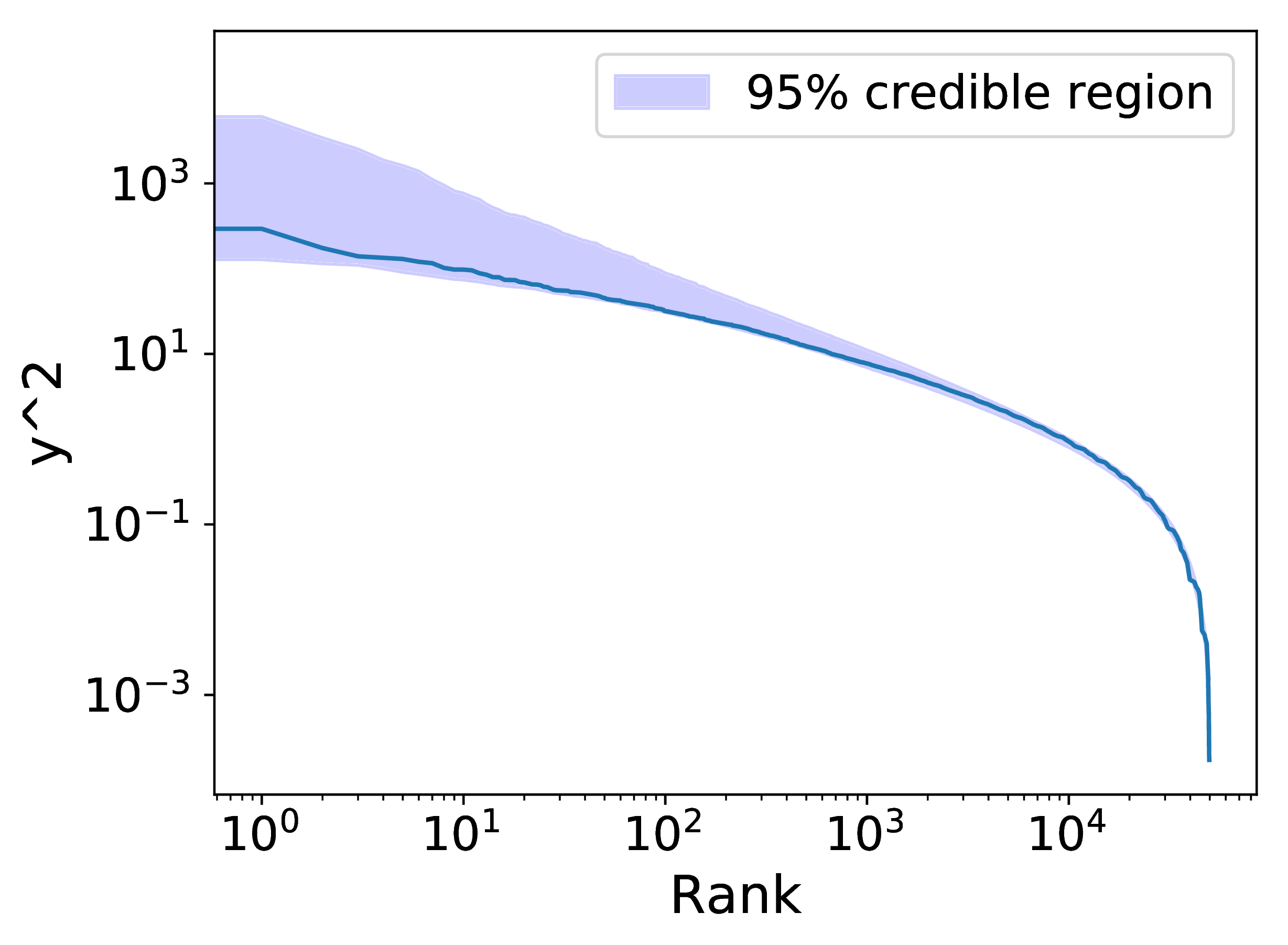}
%\caption{GHD}
\end{subfigure}
\begin{subfigure}[t]{0.24\textwidth}
\centering
\includegraphics[width=\linewidth]{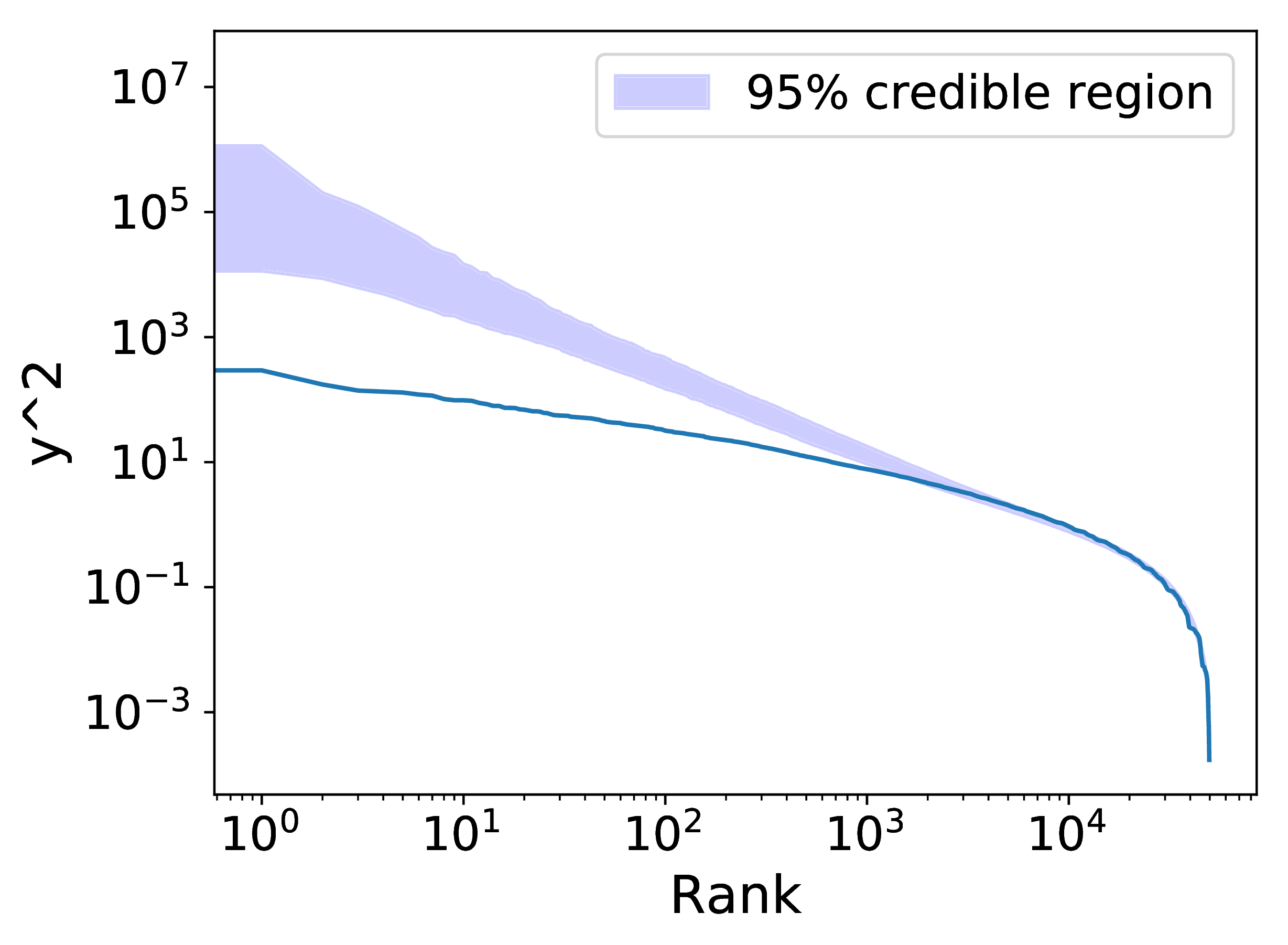}
%\caption{NS}
\end{subfigure}
\begin{subfigure}[t]{0.24\textwidth}
\centering
\includegraphics[width=\linewidth]{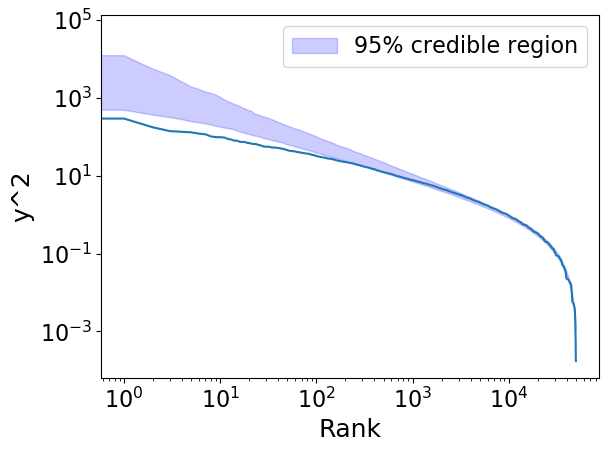}
%\caption{NS}
\end{subfigure}

\begin{subfigure}[t]{0.24\textwidth}
\centering
\includegraphics[width=\linewidth]{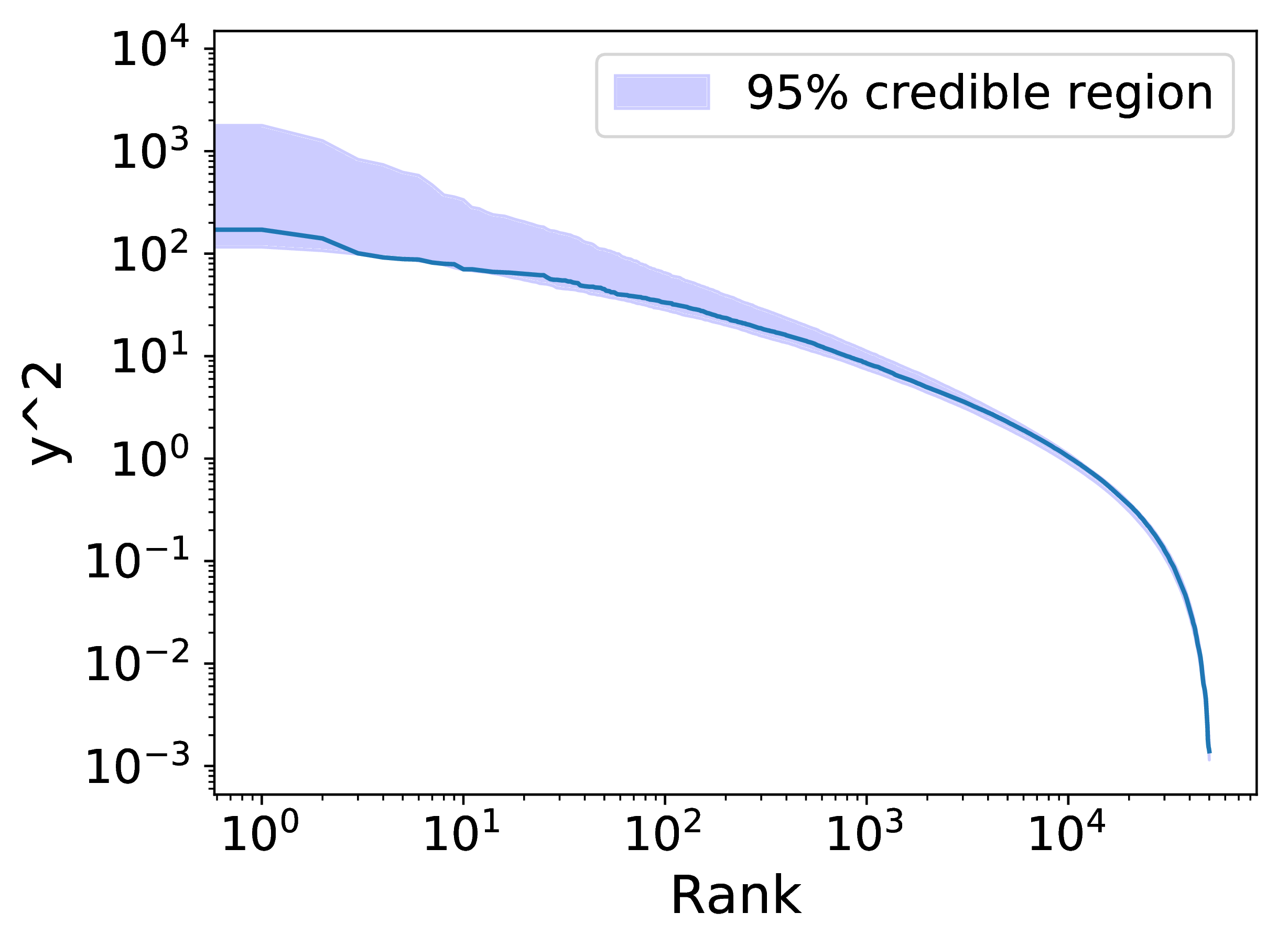}
\caption{NGGP}
\end{subfigure}
\begin{subfigure}[t]{0.24\textwidth}
\centering
\includegraphics[width=\linewidth]{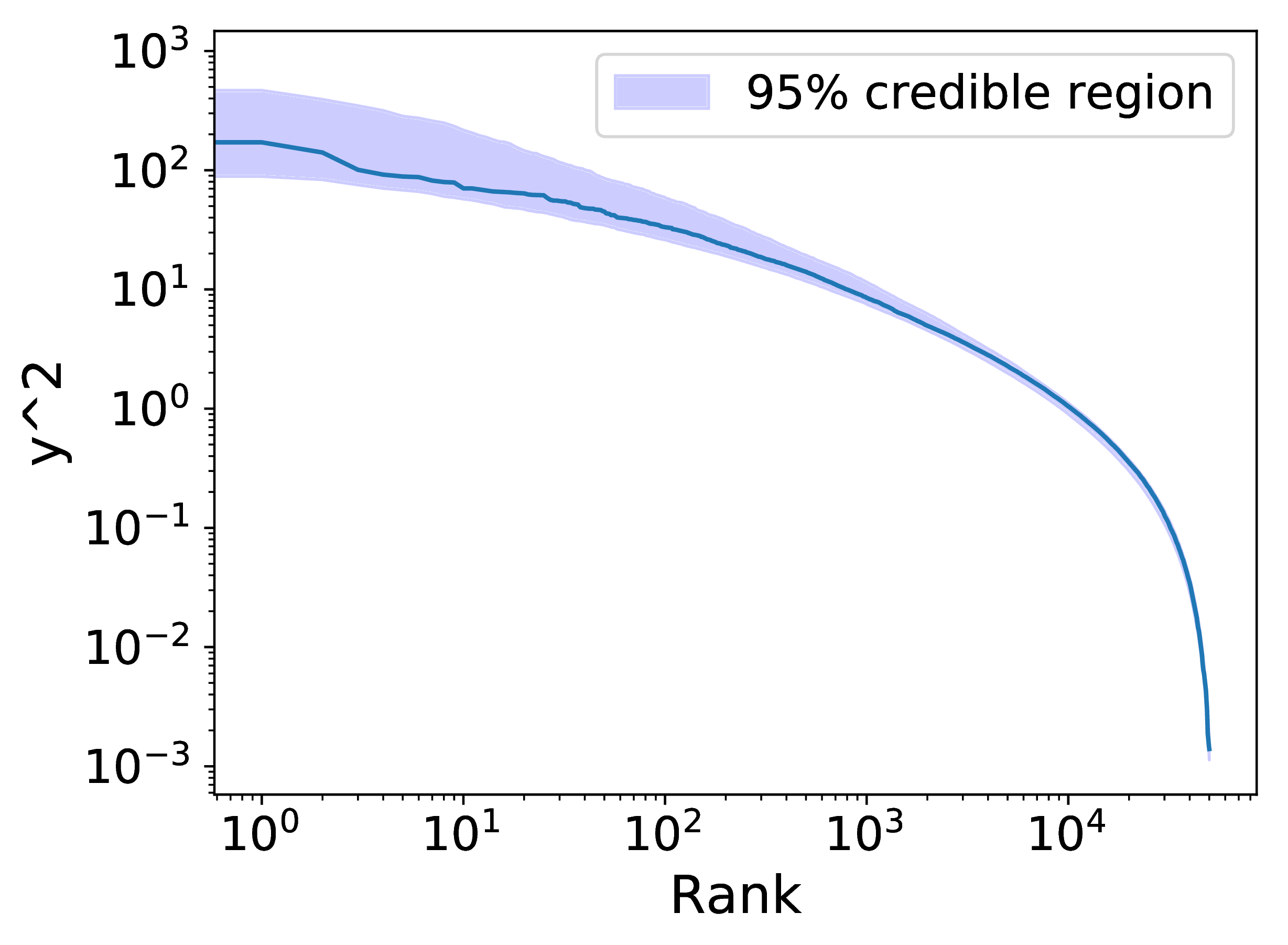}
\caption{GH}
\end{subfigure}
\begin{subfigure}[t]{0.24\textwidth}
\centering
\includegraphics[width=\linewidth]{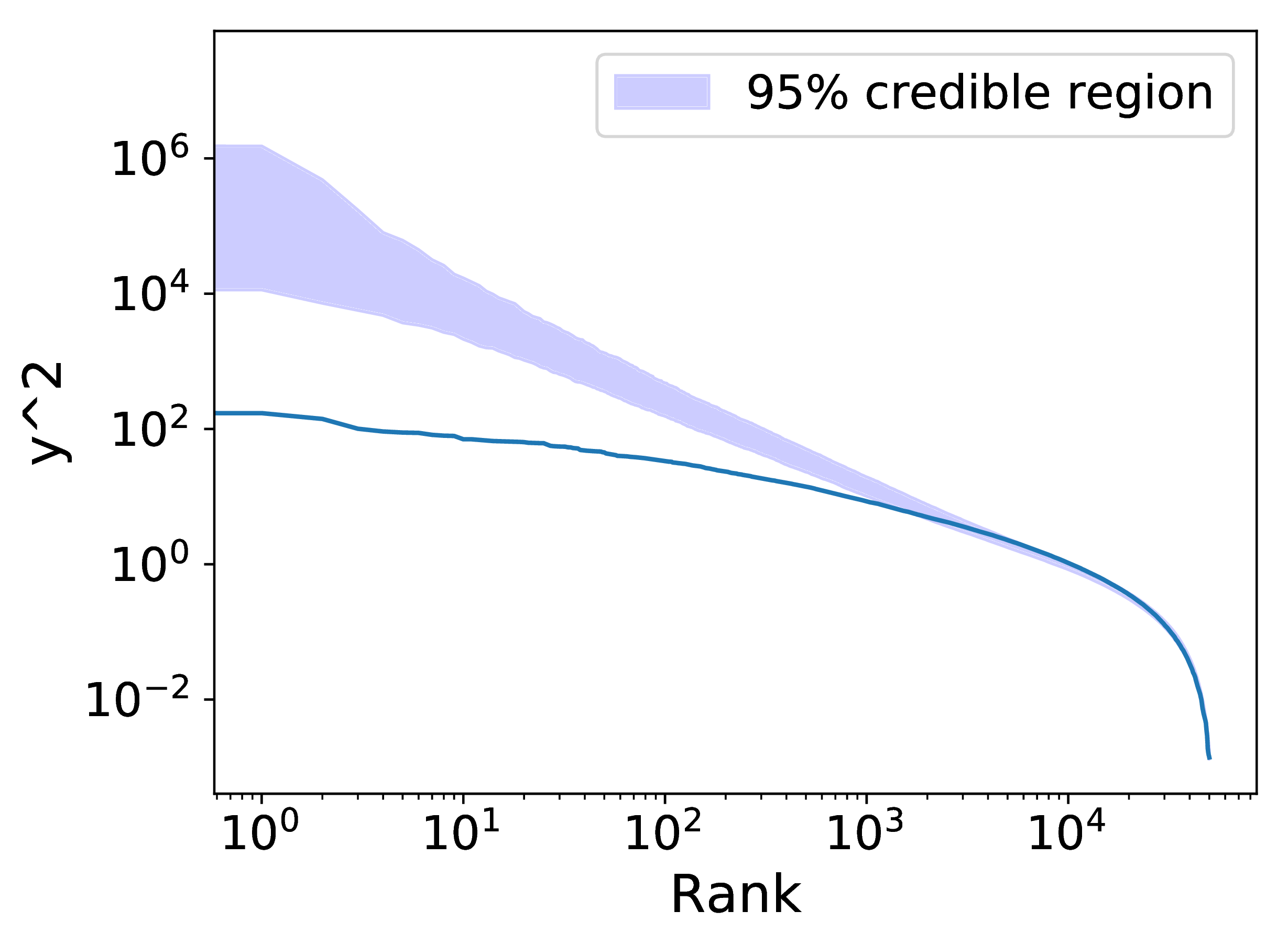}
\caption{NS}
\end{subfigure}
\begin{subfigure}[t]{0.24\textwidth}
\centering
\includegraphics[width=\linewidth]{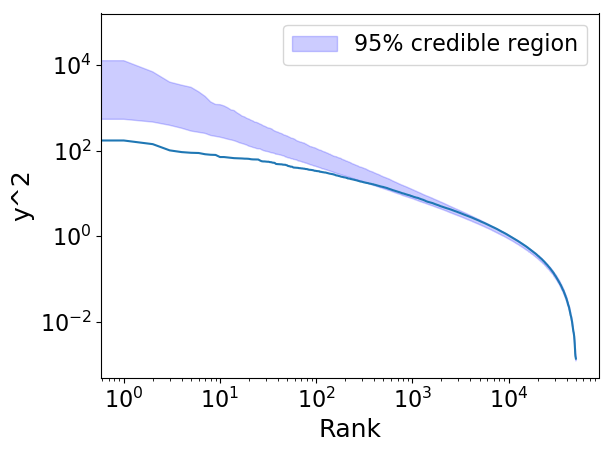}
\caption{Student}
\end{subfigure}
\caption{Ranked squared increments on the tech companies dataset. From top to bottom row: Apple, Amazon, Facebook, Google, Microsoft, Netflix. The line represents the ranked $y^2$ in the test dataset; the shaded area represent the 95\% credible region. Results are given for the  NGGP, GH,  NS and Student models in this order.}
\label{fig:app_simple_rank}
\end{figure}

\section{Additional simulation results for OU-based models}
\label{sec:OUadditionalsimus}

In this section, we provide additional simulation results for OU-based models with NGGP marginal. To assess the sensitivity of the proposed method to the choice of the prior, we consider an alternative prior on $\tau$ by choosing $\tau-1 \sim \mathrm{Unif}(0, 3)$. Figure~\ref{fig:complex_simluated_unif_3.0_app} shows the inference results for the simulated data as described in Section~\ref{sec:OUsimulated}, but with the uniform prior on $\tau$. The marginal posterior distributions for the parameters of interest are similar to those obtained with a gamma prior on $\tau$ (see Figure~\ref{fig:complex_simulated_3.0}).

Figure~\ref{fig:complex_simluated_exp_1.5_app} and \ref{fig:complex_simluated_unif_1.5_app} shows the inference results for data simulated with $\tau=1.5$ (all the other parameters were kept the same as above), both with $\mathrm{Gamma}(1,1)$ prior and $\mathrm{Unif}(0,3)$ prior on $\tau-1$. As for $\tau=3$, the posterior concentrates around the values used for simulation, and the results are rather insensitive to the choice of the prior for $\tau$.

\begin{figure}
\centering
\includegraphics[width=0.24\linewidth]{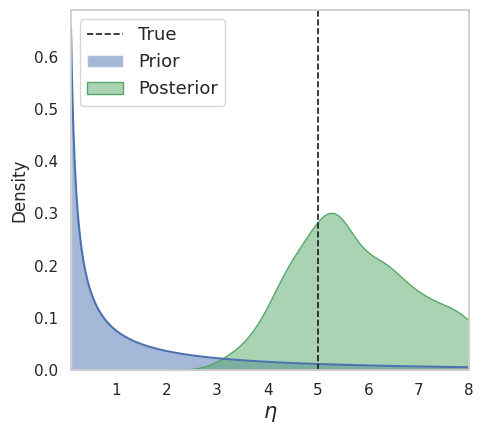}
\includegraphics[width=0.24\linewidth]{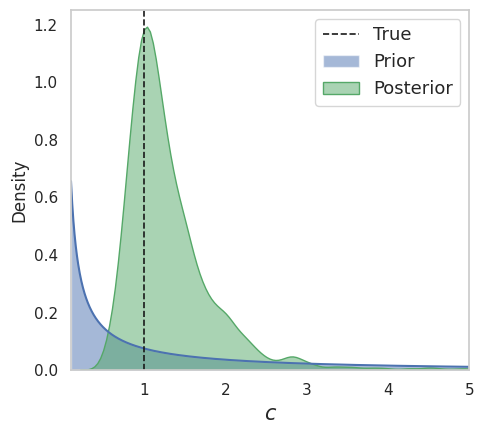}
\includegraphics[width=0.24\linewidth]{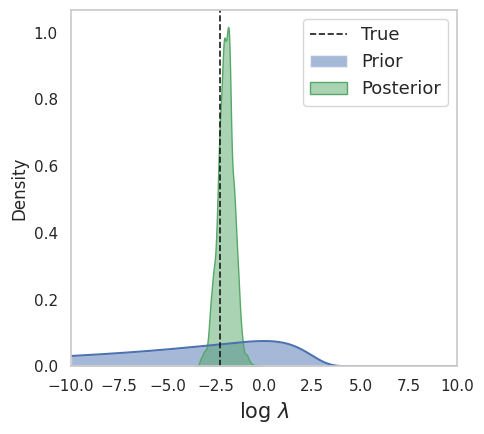}
\includegraphics[width=0.24\linewidth]{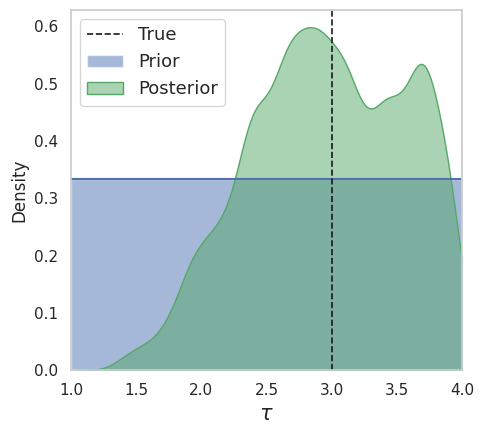}\\\includegraphics[width=0.24\linewidth]{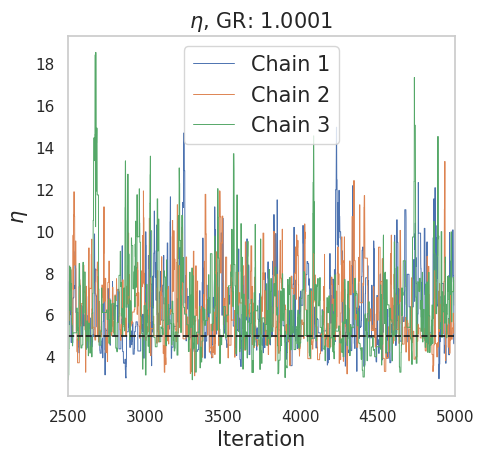}
\includegraphics[width=0.24\linewidth]{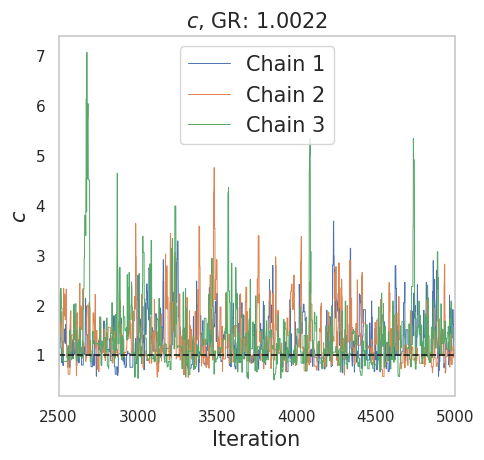}
\includegraphics[width=0.24\linewidth]{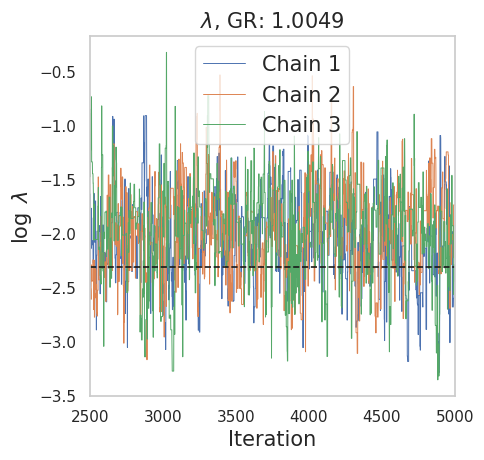}
\includegraphics[width=0.24\linewidth]{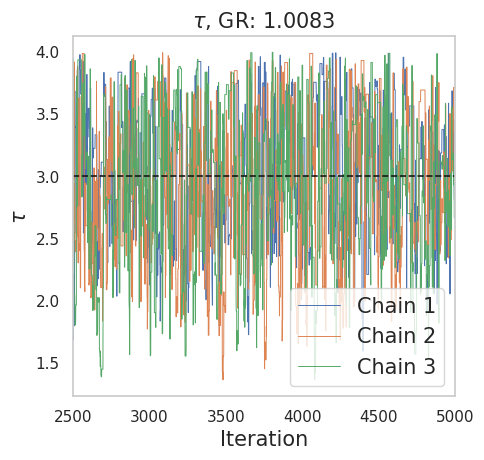}
\caption{Posterior samples of the parameters on simulated data from OU-based stochastic volatility model with NGGP marginal. Data generated with $\tau=3.0$, $\mathrm{Unif}(0, 3)$ prior on $\tau-1$.}
\label{fig:complex_simluated_unif_3.0_app}
\end{figure}

\begin{figure}
\centering
\includegraphics[width=0.24\linewidth]{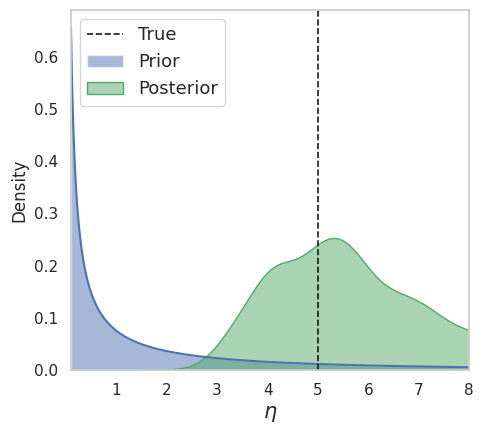}
\includegraphics[width=0.24\linewidth]{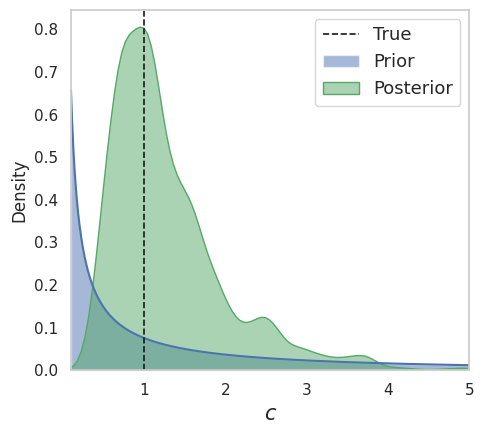}
\includegraphics[width=0.24\linewidth]{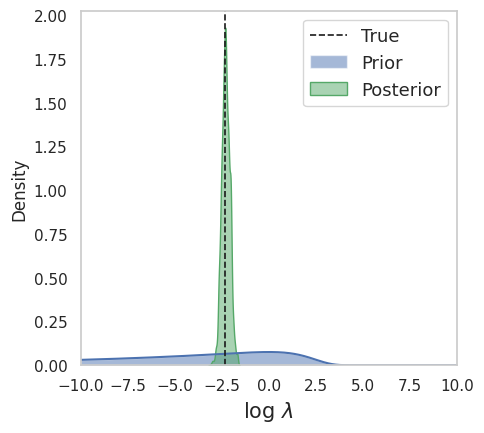}
\includegraphics[width=0.24\linewidth]{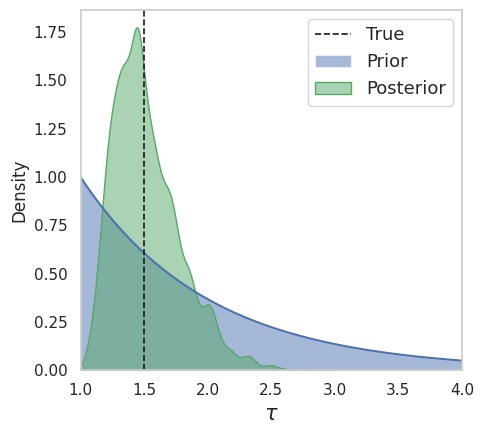}\\\includegraphics[width=0.24\linewidth]{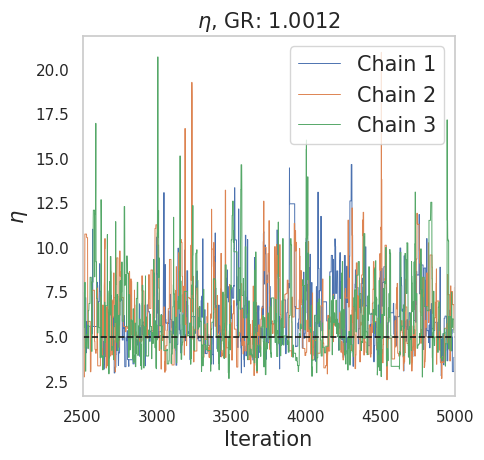}
\includegraphics[width=0.24\linewidth]{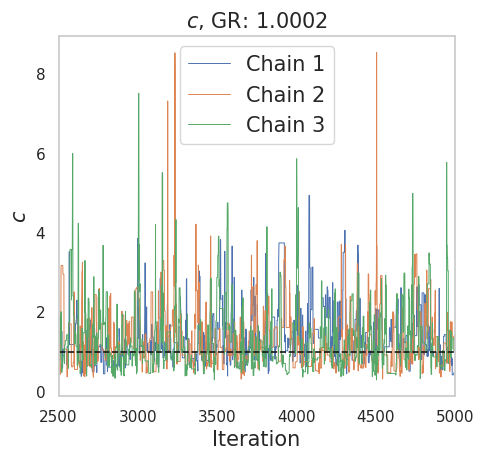}
\includegraphics[width=0.24\linewidth]{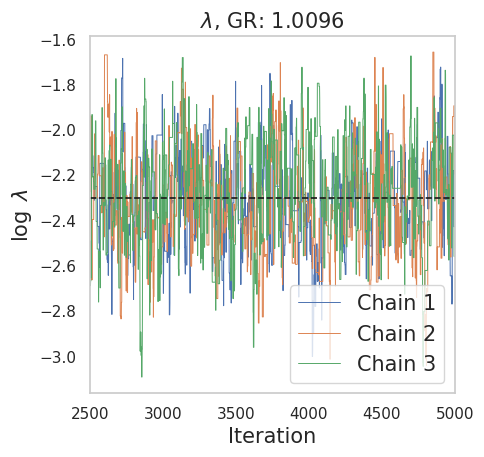}
\includegraphics[width=0.24\linewidth]{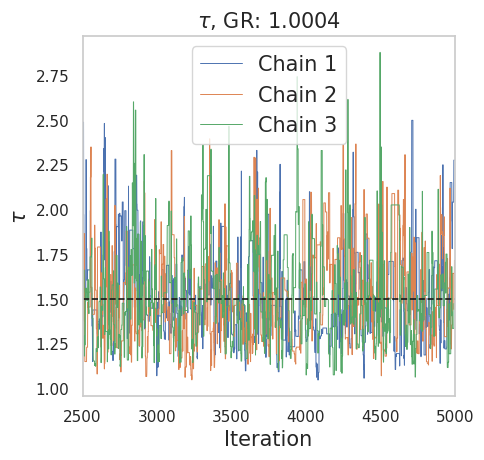}
\caption{Posterior samples of the parameters on simulated data from OU-based stochastic volatility model with NGGP marginal. Data generated with $\tau=1.5$, $\mathrm{Gamma}(1, 1)$ prior on $\tau-1$.}
\label{fig:complex_simluated_exp_1.5_app}
\end{figure}

\begin{figure}
\centering
\includegraphics[width=0.24\linewidth]{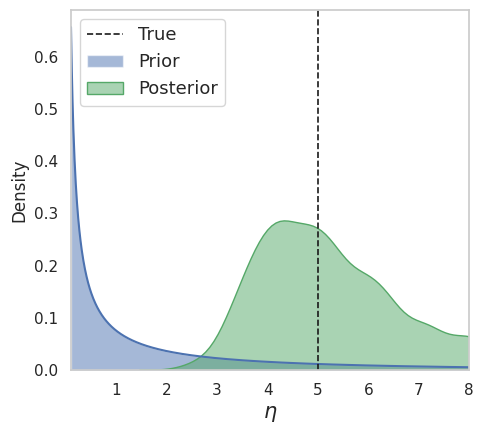}
\includegraphics[width=0.24\linewidth]{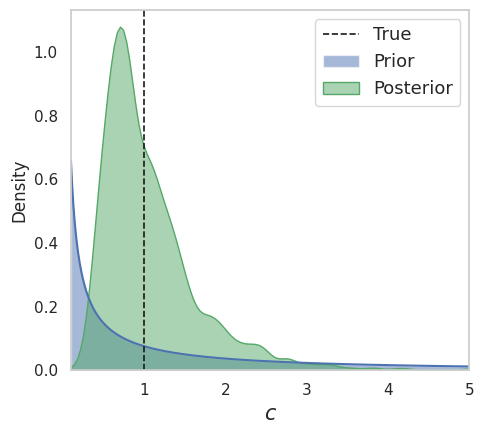}
\includegraphics[width=0.24\linewidth]{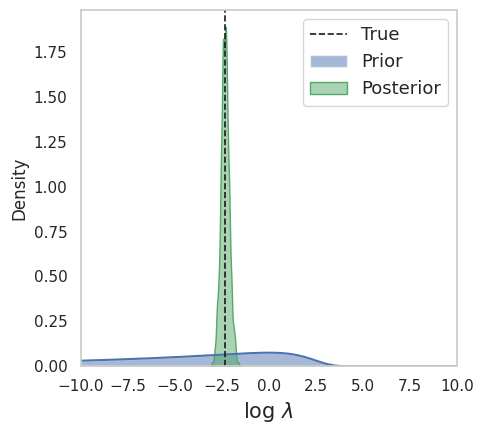}
\includegraphics[width=0.24\linewidth]{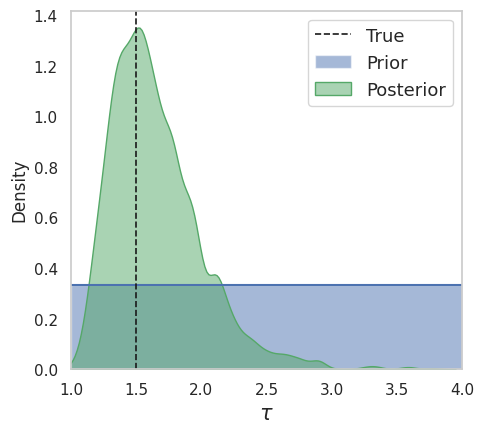}\\\includegraphics[width=0.24\linewidth]{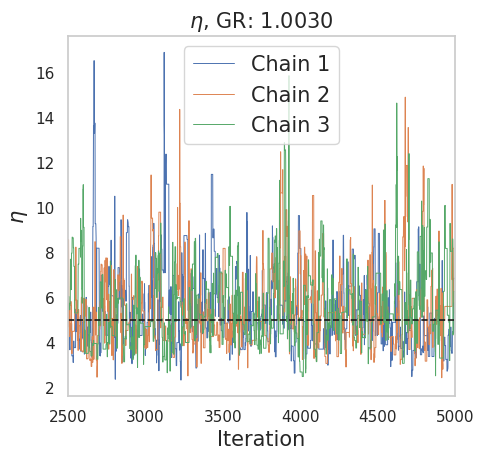}
\includegraphics[width=0.24\linewidth]{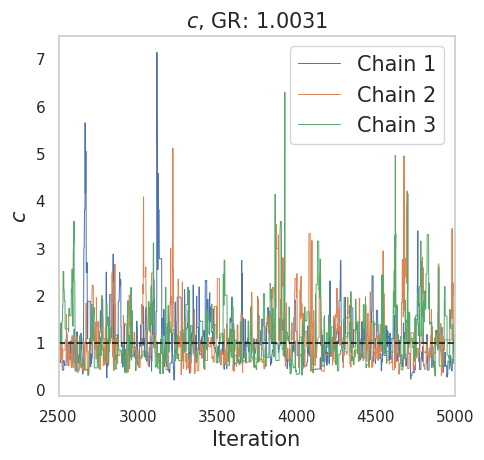}
\includegraphics[width=0.24\linewidth]{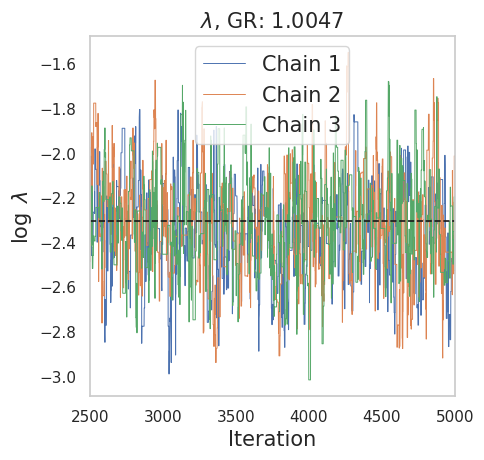}
\includegraphics[width=0.24\linewidth]{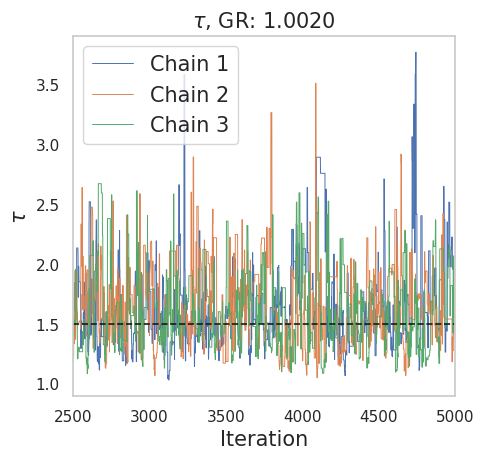}
\caption{Posterior samples of the parameters on simulated data from OU-based stochastic volatility model with NGGP marginal. Data generated with $\tau=1.5$, $\mathrm{Unif}(0, 3)$ prior on $\tau-1$.}
\label{fig:complex_simluated_unif_1.5_app}
\end{figure}

\section{Additional results for OU-based models with real data}
\label{sec:OUadditional}

\textbf{Estimated parameters} In Table  \ref{tab:app_oxford_parameters} we report the estimated parameters and credible intervals for more datasets.

\begin{table}
\scriptsize
\centering
\setlength{\tabcolsep}{0.5pt}
\caption{Posterior mean and 95\% credible intervals of the parameters of the NG and NGGP marginals for the different indices.}
\begin{tabular}{@{}c ccc cccc}
\toprule
& \multicolumn{3}{c}{NG} & \multicolumn{4}{c}{NGGP} \\
\cmidrule(lr){2-4}\cmidrule(lr){5-8}
& $\eta$ & $\lambda$ & $c$ & $\eta$ & $\lambda$ & $c$ & $\tau$ \\
\midrule
AEX & \qnts{1.43}{0.71}{2.50} & \qnts{0.02}{0.01}{0.04} & \qnts{1.41}{0.66}{2.56} & \qnts{2.58}{1.17}{4.47} & \qnts{0.03}{0.01}{0.05} & \qnts{9.34}{3.00}{23.04} & \qnts{1.49}{1.05}{2.37} \\
\addlinespace[0.5em]
AORD & \qnts{1.97}{0.92}{3.42} & \qnts{0.02}{0.01}{0.03} & \qnts{2.06}{0.88}{3.60} &  \qnts{3.62}{1.52}{6.87} & \qnts{0.03}{0.01}{0.04} &  \qnts{11.26}{2.89}{29.19} & \qnts{1.65}{1.09}{3.05} \\
\addlinespace[0.5em]
DJI & \qnts{1.28}{0.68}{2.05} & \qnts{0.02}{0.01}{0.03} & \qnts{1.26}{0.65}{2.10} & \qnts{2.19}{1.10}{3.69} & \qnts{0.03}{0.01}{0.04} & \qnts{8.53}{2.30}{19.95} & \qnts{1.38}{1.02}{2.25} \\
\addlinespace[0.5em]
FTSE  & \qnts{1.34}{0.68}{2.22} & \qnts{0.02}{0.01}{0.02} & \qnts{1.46}{0.73}{2.37} & \qnts{3.19}{1.35}{6.67} & \qnts{0.02}{0.01}{0.04} & \qnts{14.36}{3.55}{38.67} & \qnts{1.36}{1.04}{2.08} \\
\addlinespace[0.5em]
GSPTSE & \qnts{1.23}{0.59}{2.13} & \qnts{0.01}{0.01}{0.02} & \qnts{1.28}{0.57}{2.31} & \qnts{2.31}{0.92}{4.26} & \qnts{0.02}{0.01}{0.03} & \qnts{10.33}{2.42}{25.48} & \qnts{1.42}{1.02}{2.29} \\
\addlinespace[0.5em]
HSI & \qnts{1.34}{0.58}{2.38} & \qnts{0.01}{0.01}{0.02} & \qnts{1.61}{0.73}{2.91} & \qnts{2.51}{0.96}{4.60} & \qnts{0.01}{0.01}{0.03} &  \qnts{10.00}{2.70}{23.93} & \qnts{2.08}{1.07}{3.17} \\
\addlinespace[0.5em]
IBEX & \qnts{1.67}{0.88}{2.68} & \qnts{0.03}{0.02}{0.06} &  \qnts{1.72}{0.86}{2.90} & \qnts{2.92}{1.23}{5.91} & \qnts{0.04}{0.02}{0.08} & \qnts{9.17}{1.87}{26.00} & \qnts{1.68}{1.04}{3.53} \\
\addlinespace[0.5em]
IXIC & \qnts{1.43}{0.77}{2.37} & \qnts{0.02}{0.01}{0.03} & \qnts{1.42}{0.72}{2.43} & \qnts{2.06}{1.02}{3.84} & \qnts{0.02}{0.01}{0.03} &  \qnts{6.15}{1.74}{15.39} & \qnts{1.65}{1.04}{3.07} \\
\addlinespace[0.5em]
KS11 & \qnts{1.55}{0.91}{2.42} & \qnts{0.03}{0.02}{0.06} & \qnts{1.69}{1.00}{2.57} & \qnts{2.78}{1.20}{4.86} & \qnts{0.06}{0.03}{0.09} & \qnts{9.48}{2.12}{22.46} & \qnts{1.69}{1.08}{3.54} \\
\addlinespace[0.5em]
MXX & \qnts{1.09}{0.43}{1.85} & \qnts{0.01}{0.01}{0.02} &  \qnts{0.99}{0.38}{1.86} & \qnts{2.04}{0.84}{3.99} & \qnts{0.02}{0.01}{0.04} & \qnts{7.31}{1.40}{19.96} & \qnts{1.48}{1.04}{2.81} \\
\addlinespace[0.5em]
N225 & \qnts{1.05}{0.86}{1.53} &  \qnts{0.03}{0.02}{0.04} & \qnts{1.23}{0.78}{1.53} &  \qnts{1.47}{0.71}{2.37} & \qnts{0.05}{0.03}{0.07} &  \qnts{3.67}{0.98}{8.77} & \qnts{1.95}{1.27}{3.30} \\
\addlinespace[0.5em]
RUT & \qnts{1.41}{0.70}{2.35} & \qnts{0.02}{0.01}{0.03} &  \qnts{1.33}{0.56}{2.31} &  \qnts{2.70}{1.33}{4.99} & \qnts{0.03}{0.02}{0.05} & \qnts{10.43}{3.19}{25.73} & \qnts{1.39}{1.03}{2.31} \\
\addlinespace[0.5em]
SPX &  \qnts{1.22}{0.67}{2.00} & \qnts{0.02}{0.01}{0.03} & \qnts{1.19}{0.59}{2.03} & \qnts{2.07}{0.93}{3.52} & \qnts{0.02}{0.01}{0.04} & \qnts{8.91}{2.84}{19.89} & \qnts{1.37}{1.02}{2.09}\\
\addlinespace[0.5em]
SSMI & \qnts{1.75}{0.95}{2.79} & \qnts{0.03}{0.01}{0.05} & \qnts{1.86}{0.96}{3.06} & \qnts{2.86}{1.32}{5.25} & \qnts{0.04}{0.02}{0.07} & \qnts{9.18}{2.60}{24.71} & \qnts{1.60}{1.06}{2.89}\\
\bottomrule
\end{tabular}
\label{tab:app_oxford_parameters}
\end{table}

\textbf{Comparison to the ARMA-GARCH model:} In Tables \ref{app_tab:ll} and \ref{app_tab:var}, we report further results with more datasets.

\begin{table}
\centering
\caption{Comparison of the marginal log-likelihood values of OU-NG, OU-NGGP and ARMA-GARCH models on data from the Realized-library.}
\scriptsize
\begin{tabular}{@{}cccc@{}}
\toprule
Data & NG & NGGP & ARMA-GARCH \\
\midrule
  AEX   & -1466.021 & -1465.811  & -1459.768 \\
  AORD  & -1495.441  & -1494.830 & -1492.148 \\
  DJI   & -1442.606 & -1438.013 & -1424.026 \\
  FTSE  & -1448.833 & -1445.721 & -1437.881 \\
 GSPTSE & -1455.563 & -1454.125 & -1445.307 \\
  HSI   & -1509.007  & -1506.942 & -1495.847 \\
  IBEX  & -1464.315  & -1462.805 & -1455.762 \\
  IXIC  & -1470.094  & -1466.618 & -1454.834 \\
  KS11  & -1426.109 & -1419.883  & -1407.862 \\
  KSE   & -1486.406 & -1485.798 & -1478.232 \\
  MXX   & -1475.771 & -1474.586  & -1469.124 \\
  N225  & -1423.252 & -1416.556 & -1404.764 \\
  RUT   & -1453.471 & -1449.010 & -1438.184 \\
  SPX   & -1431.223 & -1424.749 & -1407.002 \\
  SSMI  & -1362.787 & -1354.640 & -1342.874 \\
\bottomrule
\end{tabular}
\label{app_tab:ll}
\end{table}

\begin{table}
\centering
\caption{The results of VaR test on predicted sequence. We trained the models using 1100 time-steps of the data and predicted remaining 900 time-steps to compute VaR with $\alpha=0.95$ and $\alpha=0.99$ for the predictions. Then we computed the fraction of actual test data less than negative of the computed VaR values. Values closer to $\alpha$ mean better VaR prediction.
}
\scriptsize
\setlength{\tabcolsep}{4pt}
\begin{tabular}{@{}ccccccc@{}}
\toprule
 & NG & NGGP & \pbox{4cm}{ARMA\\-GARCH} & NG & NGGP & \pbox{3cm}{ARMA\\-GARCH} \\
\cmidrule(lr){2-4}\cmidrule(lr){5-7}
Data & \multicolumn{3}{c}{$\alpha=0.95$} & \multicolumn{3}{c}{$\alpha=0.99$} \\
\midrule
  AEX   & 0.962 & 0.958 &   0.969   & 0.993 & 0.993 & 0.992 \\
  AORD  & 0.962 & 0.960  &    0.970  & 0.996 & 0.993 & 0.991 \\
  DJI   & 0.956 & 0.959 &   0.968  & 0.996 & 0.994 & 0.989 \\
  FTSE  & 0.956 & 0.953 &    0.950  & 0.994 & 0.993 & 0.983 \\
 GSPTSE & 0.967 & 0.968 &   0.978  & 0.997 & 1.000 & 0.996\\
  HSI   & 0.978 & 0.972 &   0.952  & 0.996 & 0.996 & 0.988  \\
  IBEX  &  0.960  & 0.959 &   0.981 & 0.996 & 0.996 & 0.997  \\
  IXIC  & 0.964 & 0.961 &   0.953  & 0.994 & 0.993 & 0.986 \\
  KS11  &  0.970  & 0.969 &   0.962  & 0.996 & 0.997 & 0.994 \\
  KSE   &  0.960  & 0.966 &   0.932  & 0.996 & 0.996 & 0.980 \\
  MXX   & 0.962 & 0.964 &   0.966  & 0.993 & 0.990 & 0.987 \\
  N225  & 0.963 & 0.963 &   0.954  & 0.994 & 0.991 & 0.990 \\
  RUT   &  0.960  & 0.959 &   0.952  & 0.996 & 0.993 & 0.994 \\
  SPX   &  0.970  & 0.969 &   0.962  & 0.994 & 0.997 & 0.993 \\
  SSMI  & 0.963 & 0.967 &    0.960  & 0.994 & 0.992 & 0.993  \\
\bottomrule
\end{tabular}
\label{app_tab:var}
\end{table}

\end{document}